\documentclass[acmsmall, nonacm]{acmart}

\AtBeginDocument{%
	\providecommand\BibTeX{{%
			\normalfont B\kern-0.5em{\scshape i\kern-0.25em b}\kern-0.8em\TeX}}}

\acmJournal{TOPS}

\usepackage{lineno}
\usepackage{boldline} 
\usepackage{multirow}
\usepackage{amsmath}
\usepackage{amsfonts}
\usepackage{algorithmic}
\usepackage{mathtools}
\usepackage{verbatim} 
\usepackage{stmaryrd}
\usepackage{enumerate}
\usepackage{setspace}
\usepackage{url}
\usepackage{bbm} 
\allowdisplaybreaks
\usepackage{basic}
\usepackage{float}
\usepackage{ambients}
\usepackage{graphicx}
\usepackage{placeins}
\usepackage{hyperref}
\usepackage{todonotes}
\usepackage{paralist}
\usepackage{epsf}
\usepackage{graphics} 
\usepackage{epsfig} 

\renewcommand{\timeout}[2]{\lfloor #1 \rfloor #2  }
\renewcommand{\tick}{{\scriptstyle \mathsf{tick}}}
\newcommand{\maxa}{{\scriptstyle \mathsf{maxa}}}
\newcommand{\fineC}{{\scriptstyle \mathsf{end}}}

\newif\ifdraft\drafttrue

\newcommand{\ntrans}[1]{\mathrel{{\trans{#1}}\makebox[0em][r]{$\not$\hspace{2ex}}}{\!}}

\newcommand{\on}{\mathsf{on}}
\newcommand{\off}{\mathsf{off}}

\newcommand{\confCPS}[2]{#1 \, {\Join} \, #2}

\newcommand{\defn}{\triangleq}

\renewcommand{\operatorname}[1]{\mathit{#1}}

\newcommand\restrict[1]{\raise-.5ex\hbox{\ensuremath|}_{#1}}

\usepackage{marvosym}

\makeatletter
\renewcommand{\xRightarrow}[2][]{\ext@arrow 0359\Rightarrowfill@{#1}{#2}}
\newcommand{\xRightarrowBis}[2][]{\ext@arrow 0359\Rightarrowfill@{#1}{#2}}
\makeatother

\newcommand{\Edit}{\mathsf{E}}
\newcommand{\Fdit}{\mathsf{F}}

\definecolor{darkred}{RGB}{128,0,0}
\definecolor{darkgreen}{RGB}{0,128,0}
\definecolor{lightgreen}{RGB}{224,255,224}

\newcommand{\funEdit}[1]{ \langle \!| \ #1\ |\! \rangle}

\newcommand{\ActSet}{\mathsf{Act}}
\newcommand{\SensSet}{\mathsf{Sens}}
\newcommand{\ChanSet}{\mathsf{Chn}}
\newcommand{\insertE}[2]{\eact{#2}{#1}}
\newcommand{\allowE}[1]{#1}
\newcommand{\suppressE}[1]{{^{-}{#1}}}

\newcommand{\synthES}[2]{ \langle \! | #1 | \! \rangle_{#2}}

\newcommand{\synthESP}[3]{{ \langle \! | #1 | \! \rangle_{#2}^{#3}}}
\newcommand{\funEditP}[2]{{\langle \!| \!{#1}\! | \! \rangle^{#2} }}

\newcommand{\PSet}{\mathcal{P}}

\newcommand{\Times}[3]{\mathsf{Prod}_{#3}(#1,#2)}
\newcommand{\TimesP}[3]{\mathsf{Prod}^{\mathcal P}_{#3}(#1,#2)}

\newcommand{\regSemantics}[1]{ \big\llbracket#1\big\rrbracket}

\newcommand{\CND}[2]{\mathrm{Cnd}(\mathit{#1, #2})}
\newcommand{\PCND}[3]{\mathrm{PCnd}_{\mathit{#3}}(\mathit{#1, #2})}
\newcommand{\caseCND}[1]{\mathrm{Case}{#1}}
\newcommand{\tpl}[2]{\mathrel{(#1, #2)}}

\newcommand{\BE}[2]{\mathrm{BE}_{\mathit{#2}}(\mathit{#1})}
\newcommand{\BP}[2]{\mathrm{BP}_{\mathit{#2}}(\mathit{#1})}
\newcommand{\BA}[2]{\mathrm{BA}_{\mathit{#2}}(\mathit{#1})}
\newcommand{\CBE}[4]{\mathrm{CBE}_{[\mathit{#3, #4}]}(\mathit{#1, #2})}
\newcommand{\CBP}[4]{\mathrm{CBP}_{[\mathit{#3, #4}]}(\mathit{#1, #2})}
\newcommand{\CBA}[4]{\mathrm{CBA}_{[\mathit{#3, #4}]}(\mathit{#1, #2})}
\newcommand{\MinD}[4]{\mathrm{MinD}(\mathit{#1, #2,#3, #4 })}
\newcommand{\MaxD}[4]{\mathrm{MaxD}(\mathit{#1, #2, #3, #4})}
\newcommand{\BR}[5]{\mathrm{BR}(\mathit{#1, #2, #3, #4, #5})}
\newcommand{\BI}[5]{\mathrm{BI}(\mathit{#1, #2, #3, #4, #5})}

\newcommand{\BME}[2]{\mathrel{\mathrm{BME}_{\mathit{#2}}{#1}}}

\newcommand{\sysAct}[1]{\mathit{out}(#1)}
\newcommand{\trigger}[1]{\mathit{trigger}(#1)}

\usepackage{tikz}
\usetikzlibrary{arrows,positioning,shapes}
\usetikzlibrary{positioning}
\tikzset{
	>=stealth',
	punkt/.style={
		rectangle,
		rounded corners,
		draw=black, very thick,
		text width=6.5em,
		minimum height=2em,
		text centered},
	pil/.style={
		->,
		thick,
		shorten <=2pt,
		shorten >=2pt,}
}
\tikzstyle{line} = [draw, -latex,thick,
shorten <=2pt,
shorten >=2pt]

\newcommand{\Xrec}{\mathsf{X}}

\newcommand{\Zrec}{\mathsf{Z}}

\newcommand{\Events}{\mathsf{Events}}
\newcommand{\PEvents}{\mathsf{PEvents}}
\newcommand{\PUEvents}{\mathsf{PUEvents}}
\newcommand{\UEvents}{\mathsf{PUEvents}}

\newcommand{\Sys}{N}

\renewcommand{\confCPS}[2]{#1 \! \bowtie \!  {\boldsymbol{ \{ }}#2{\boldsymbol{ \} }}}

\def\eact#1#2{#2 \prec #1}

\makeatletter
\newcommand{\subalign}[1]{%
	\vcenter{%
		\Let@ \restore@math@cr \default@tag
		\baselineskip\fontdimen10 \scriptfont\tw@
		\advance\baselineskip\fontdimen12 \scriptfont\tw@
		\lineskip\thr@@\fontdimen8 \scriptfont\thr@@
		\lineskiplimit\lineskip
		\ialign{\hfil$\m@th\scriptstyle##$&$\m@th\scriptstyle{}##$\crcr
			#1\crcr
		}%
	}
}
\makeatother

\renewcommand{\on}[1]{ {\mathsf{\scriptstyle on_{#1}}}  }
\renewcommand{\off}[1]{ {\mathsf{\scriptstyle off_{#1}}}  }
\newcommand{\open}{ {\mathsf{\scriptstyle open}}  }
\newcommand{\close}{ {\mathsf{\scriptstyle close}}  }
\newcommand{\ask}[1]{ {\mathsf{\scriptstyle {#1}{\_}req}}  }

\newcommand{\Fill}{ \uparrow  }
\newcommand{\Empty}{ \downarrow}

\newcommand{\PropG}{\mathbbm{PropG}}
\newcommand{\PropL}{\mathbbm{PropL}}

\newcommand{\ctrlDim}[1]{ \mathsf{dim}(#1)  }
\newcommand{\propDim}[1]{ \mathsf{dim}(#1)  }

\newcommand{\events}[1]{ \mathsf{events}(#1)  }

\newtheorem{definition}{Definition}
\newtheorem{theorem}{Theorem}
\newtheorem{remark}{Remark}
\newtheorem{notation}{Notation}
\newtheorem{proposition}{Proposition}
\newtheorem{corollary}{Corollary}
\newtheorem{lemma}{Lemma}

\definecolor{BrickRed}{RGB}{204, 0, 0}
\definecolor{ForestGreen}{RGB}{0, 128, 0}

\begin{document}

	\title[Runtime  Enforcement of PLCs]{Runtime  Enforcement of Programmable Logic Controllers}
	
\author{Ruggero Lanotte}
\orcid{0000-0002-3335-234X}
\affiliation{%
	\institution{
		Universit\`a degli Studi dell'Insubria}
	\streetaddress{Dipartimento di Scienze Umane e dell'Innovazione per il Territorio, via Sant'Abbondio 12}
	\city{Como}
	\postcode{22100}
	\country{Italy}
}
\email{ruggero.lanotte@uninsubria.it}

\author{Massimo Merro}
\orcid{0000-0002-1712-7492}
\affiliation{%
	\institution{Universit\`a degli Studi di Verona}
	\streetaddress{Dipartimento di Informatica, strada Le Grazie 15}
	\city{Verona}
	\postcode{37134}
	\country{Italy}
}
\email{massimo.merro@univr.it}

\author{Andrei Munteanu}
\affiliation{%
	\institution{Universit\`a degli Studi di Verona}
	\streetaddress{Dipartimento di Informatica, strada Le Grazie 15}
	\city{Verona}
	\postcode{37134}
	\country{Italy}
}
\email{andrei.munteanu@univr.it}

\renewcommand{\shortauthors}{R. Lanotte, M. Merro, A. Munteanu}

\begin{abstract}
		 With the advent of \emph{Industry 4.0}, industrial facilities and critical infrastructures are transforming into an ecosystem of heterogeneous physical and cyber components, such as \emph{programmable logic controllers}, increasingly interconnected and therefore exposed to \emph{cyber-physical attacks}, \emph{i.e.},  security breaches in cyberspace that may adversely affect the physical processes underlying \emph{industrial control systems}. 

                      In this paper, we  propose a  \emph{formal approach} based on  \emph{runtime enforcement} to ensure specification compliance in  networks of  controllers,  possibly compromised by \emph{colluding malware} that may tamper with 
        actuator commands,  
        sensor readings,  and 
        inter-controller communications.  
        Our approach relies on  an ad-hoc sub-class of Ligatti et al.'s \emph{edit automata} to enforce controllers  represented in Hennessy and Regan's \emph{Timed Process Language}. We define a synthesis algorithm that, given an alphabet $\PSet$ of observable actions and  a 
        timed correctness property $e$,  returns a  monitor that enforces the property $e$ during the  execution of any (potentially corrupted) controller with alphabet $\PSet$, and complying with the property $e$. Our  monitors \emph{correct} and \emph{suppress} incorrect actions coming from  corrupted controllers and  \emph{emit} actions in full autonomy when the controller under scrutiny is not able to do so in a correct manner.  Besides classical requirements, such as \emph{transparency} and \emph{soundness}, the proposed enforcement enjoys   \emph{deadlock- and diverge-freedom} of monitored controllers, together with  \emph{scalability}    when dealing with networks of controllers. Finally, we test the proposed enforcement mechanism on a non-trivial case study, taken from the context of industrial water treatment systems, in which the controllers are injected with  different malware with different malicious goals. 
	\end{abstract}

\ccsdesc[500]{Security and privacy~Formal security models}
\ccsdesc[500]{Security and privacy~Cyber-physical systems security}

\keywords{Runtime enforcement, control systems security, PLC malware}
\maketitle

	\section{Introduction}
	\emph{Industrial Control Systems} (ICS{s}) are physical and engineered systems whose operations are monitored, coordinated, controlled, and integrated by a computing and communication core \cite{CPS-DEF}.
        They represent the backbone of Critical Infrastructures for safety-critical applications such as electric power distribution, nuclear power production, and water supply.

        The growing connectivity and integration 
in \emph{Industry 4.0} has triggered a
        dramatic increase \nolinebreak in the number of \emph{cyber-physical attacks}~\cite{HCALTS2009} targeting ICSs, \emph{i.e.},  security breaches in cyberspace \nolinebreak  that adversely affect the physical processes.  
	%
	Some notorious examples are: (i) the \emph{Stuxnet} worm, which reprogrammed Siemens PLCs of nuclear centrifuges in  the nuclear facility of Natanz in Iran~\cite{stuxnet}; 
	(ii) the \emph{CRASHOVERRIDE} attack on the Ukrainian power grid, otherwise known as Industroyer~\cite{chrashoverride};
	(iii)  the recent \emph{TRITON/TRISIS} malware  that targeted a petrochemical plant in Saudi Arabia~\cite{Triton}. 

           One of the key components of ICSs  are \emph{Programmable Logic Controllers}, better known as PLCs. They  control mission-critical electrical hardware such as pumps or centrifuges, effectively serving as a bridge between the cyber and the physical worlds.
	PLCs  have an ad-hoc architecture to execute simple repeating processes known as \emph{scan cycles} (IEC 61131-3~\cite{61131-3}). Each scan cycle consists  of three phases: (i) reading of sensor measurements of the physical process;
        (ii) execution of the controller code to compute  how  the physical process should evolve; 
        (iii) transmission of commands 
                to the actuator devices to govern the physical process as desired.

	Due to their sensitive role in controlling industrial processes, successful exploitation of  PLCs can have severe consequences on ICSs. In fact, although modern controllers provide security mechanisms to allow only legitimate firmware to be uploaded, the running code can be typically altered by anyone with network or USB access to the controllers (see Figure~\ref{fig:Sys}). 	Published scan data shows how thousands of  PLCs are directly accessible from the Internet to improve efficiency~\cite{Radvanovsky2013}. 	Thus, despite their responsibility, controllers   
	 are vulnerable to several kinds of attacks, including PLC-Blaster worm~\cite{BLACKHAT2016},  Ladder Logic Bombs~\cite{Ladder-logic-bombs}, and  PLC PIN Control attacks~\cite{Abbasi2016}. 

	\begin{figure}[t]
		\centering
		\includegraphics[width=8.5cm,keepaspectratio=true,angle=0]{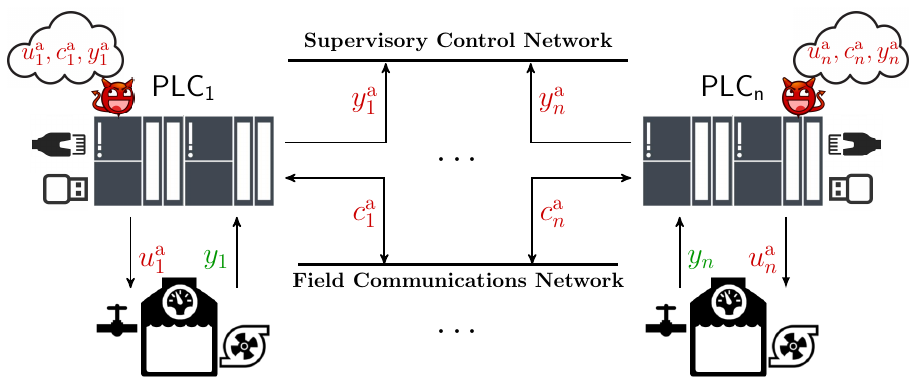}
		\caption{A  network of compromised PLCs:  \textcolor{ForestGreen}{$y_i$}  denote genuine sensor measurements,  \textcolor{BrickRed}{$y_i^{\mathrm a}$} are corrupted sensor measurements, \textcolor{BrickRed}{$u_i^{\mathrm a}$}  corrupted actuator commands, and \textcolor{BrickRed}{$c_i^{\mathrm a}$} denote corrupted inter-controller communications.}
		\label{fig:Sys}
	\end{figure}



        Extra \emph{trusted hardware components} have been  proposed to enhance the security  of  PLC architectures~\cite{McLaughlin-ACSAC2013,Mohan-HiCONS2013}. 
	For instance, McLaughlin~\cite{McLaughlin-ACSAC2013} proposed  a policy-based \emph{enforcement mechanism}  to mediate  the actuator commands transmitted by the PLC to the physical  plant.
        Mohan  et al.~\cite{Mohan-HiCONS2013} introduced a different architecture,   in which every PLC runs under the scrutiny  of a 
        \emph{monitor} which looks for deviations with respect to safe behaviours.
	Both  architectures  have been validated by means of simulation-based techniques. However, as far as we know, \emph{formal methodologies}  have  been
   {\color{black} rarely}     used to model and formally verify security-oriented architectures for ICSs.

\emph{Runtime enforcement} \cite{Schneider2000,Ligatti2005,Falcone3} 
is a  formal verification/validation technique
aiming at correcting possibly-incorrect executions  of a system-under-scrutiny (SuS). It employs a kind of monitor~\cite{Franc21} 
	that acts as a \emph{proxy}  between the SuS  and the environment interacting with it. 
	At runtime, the monitor \emph{transforms} any incorrect executions exhibited by the SuS into correct ones by either \emph{replacing}, \emph{suppressing} or \emph{inserting}  observable actions on behalf of the system.
	%
	%
	The  effectiveness of the enforcement depends on the achievement of the two following general principles~\cite{Schneider2000,Ligatti2005}:
	\begin{itemize}
		\item 
		\emph{transparency}, \emph{i.e.}, the enforcement  must not alter correct executions of the SuS; 
		\item \emph{soundness}, \emph{i.e.},   incorrect executions of the SuS must be prevented. 
	\end{itemize}

           In this paper, we  propose a  \emph{formal approach} based on \nolinebreak \emph{runtime enforcement} to ensure specification compliance in  networks of  controllers  possibly compromised by \emph{colluding malware} that may tamper with 
        actuator commands,  
        sensor readings,  and 
        inter-controller communications. 
        combined with \emph{automatic recovery mechanisms}.

	Our \emph{goal} is
        to enforce 
        potentially corrupted controllers  using
        \emph{secure proxies} based on a sub-class of Ligatti et al.'s edit automata~\cite{Ligatti2005}. These automata  will be \emph{synthesised} from enforceable \emph{timed correctness  properties} to form networks of \emph{monitored controllers}, as in Figure~\ref{fig:Sys-structure}. The proposed  enforcement will  enjoy both 
        transparency and soundness together  with the following features:
	\begin{itemize} 
		\item \emph{determinism preservation}, \emph{i.e.}, the enforcement should not introduce \emph{nondeterminism}; 
		\item \emph{deadlock-freedom},  \emph{i.e.}, the enforcement should not introduce deadlocks; 
                \item \emph{divergence-freedom}, \emph{i.e.}, the enforcement should not introduce divergencies; 
	        \item \emph{mitigation} of incorrect/malicious activities;
		\item \emph{scalability}, \emph{i.e.}, the enforcement mechanism should scale to networks of controllers. 
	\end{itemize}
	Obviously,   when a controller is compromised,  these  objectives can be achieved only with the introduction of a physically independent \emph{secure proxy}, as advocated by McLaughlin  and  Mohan et al.~\cite{McLaughlin-ACSAC2013,Mohan-HiCONS2013},  which does not have any Internet or USB access, and which is connected with the monitored controller via  \emph{secure channels}.
	This may seem like we just moved the problem over to securing the proxy. However,  this is not the case because the proxy only needs to enforce a  \emph{timed correctness  property} of the system, while the controller does the whole job of controlling the physical process relying on potentially dangerous communications via the Internet or the USB ports. Thus, any upgrade of the control system will be made to the controller and not to the secure proxy. Of course, by no means runtime reconfigurations of the secure proxy should be allowed 
        {\color{black} as its enforcing should be based on the  physics of the plant itself and not on the controller code. 
          }

	\begin{figure}[t]
		\centering
		\includegraphics[width=8.5cm,keepaspectratio=true,angle=0]{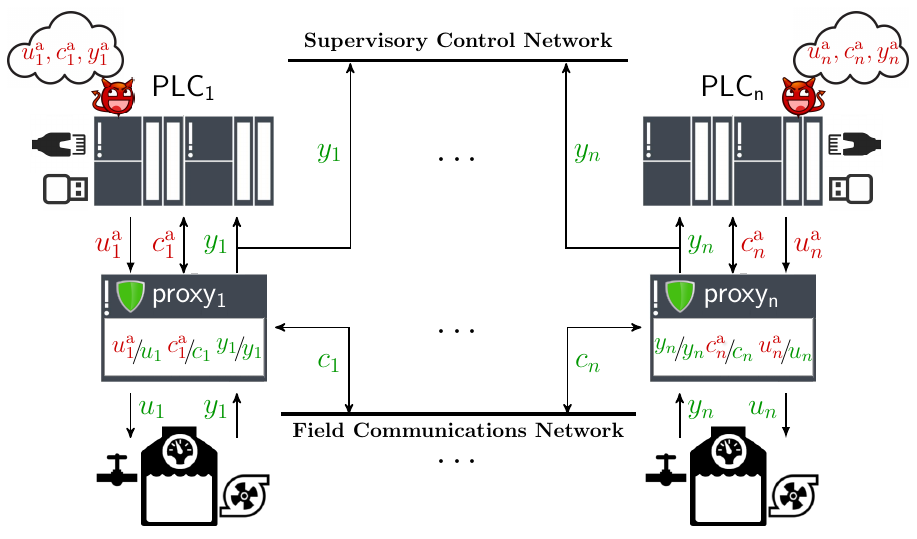}
		\caption{A  network of monitored controllers.}
		\label{fig:Sys-structure}
	\end{figure}

	\vspace*{-1mm}
	\paragraph*{Contribution}
	{\color{black}Fist of all, we define the attacker model and the attacker objectives in an enforced ICS architecture such as that depicted in Figure~\ref{fig:Sys-structure}. } Then, we introduce a formal language to specify controller programs. For this very purpose, we resort to  \emph{process calculi}, a successful and widespread formal approach in \emph{concurrency theory} for representing complex systems, such as \nolinebreak mobile systems~\cite{Ambients} and cyber-physical systems~\cite{LaMeTi18}, and used in many  areas, including verification \nolinebreak of security protocols~\cite{appliedpi,spi} and  security analysis of \emph{cyber-physical attacks}~\cite{LMMV20}.
        Thus, we define \nolinebreak a  \nolinebreak simple timed process calculus, based on  Hennessy and Regan's \emph{Timed Process Language} (TPL)~\cite{HR95}, 
        for specifying controllers, finite-state enforcers, and networks of communicating \nolinebreak  monitored \nolinebreak controllers.

	Then, we define a simple description language to express \emph{timed  correctness properties} that should hold {\color{black}for a (possibly unbounded) number} 
        of scan cycles of the monitored controller. 
	This will allow us to abstract over controllers implementations, focusing on general properties which may even be shared \nolinebreak by completely different controllers. In this regard, we might resort to one of the several logics existing in the li\-te\-rature for mo\-ni\-toring timed concurrent systems, and in particular cyber-physical systems (see, \emph{e.g.},  \cite{SurveyBartocciSTL2018,toolchain}).  
                \enlargethispage{.9\baselineskip}
        {\color{black}However, the  peculiar  iterative behaviour of controllers convinced us to adopt the  sub-class of \emph{regular expressions} that can be recognised by finite automata whose cycles always contain at least one final state; this is the largest class of regular properties that  can be enforced by  finite-state Ligatti et at.'s edit automata (see  Beauquier et al.'s work~\cite{BCL13}).
          }  
        In Section~\ref{sec:logics}, we express a wide class of  correctness properties for controllers in terms of (our) regular properties.

	  After defining a formal language to describe controller properties, we  provide  a \emph{synthesis function}  $\funEditP{-}\PSet$ that, given
          an alphabet $\PSet$ of observable actions (sensor readings, actuator commands, and inter-controller communications) and 
          a deterministic regular property $e$ combining events of $\PSet$,  returns  an edit automaton $\funEditP{e}\PSet \!$. The resulting enforcement mechanism will   ensure the required features mentioned before:   transparency, soundness, determinism preservation, deadlock-freedom, divergence-freedom, mitigation and scalability. 
            Then, we propose a non-trivial case study,  taken from the context of industrial water treatment systems, and implemented as follows: (i) the physical plant is simulated in \emph{Simulink}~\cite{MATLAB}; (ii) the open source  PLCs are implemented in \emph{OpenPLC}~\cite{OpenPLC} and executed on Raspberry Pi; (iii) the enforcers run on connected FPGAs. In this setting, we test our enforcement mechanism when  injecting the PLCs with $5$ different malware, with different \nolinebreak goals.

	    \paragraph*{Outline}
            Section~\ref{sec:attacker-model} describes the attacker model and the attacker objectives. 
            Section~\ref{sec:calculus} gives a formal language  for monitored  controllers.  Section~\ref{sec:case-study} defines the  case  study. Section~\ref{sec:logics} provides a language of regular properties to express controller behaviours; it also contains a taxonomy of properties expressible in the language. 
        Section~\ref{sec:synthesis} contains the algorithm \nolinebreak to synthesise monitors from  regular properties, together with  the main results.	
  Section~\ref{sec:implementation} discusses the implementation of the case study when exposed to five different attacks.  Section~\ref{sec:related} is devoted to related work. 
        Section~\ref{sec:conclusion} draws conclusions and discusses future work. Technical proofs  can be found in the appendix. 



{\color{black}

  \section{Attacker model and attacker objectives}
  \label{sec:attacker-model}
  In the following sections, we will propose  an enforcement-based architecture for ICSs (as those depicted in Figure~\ref{fig:Sys-structure})  to counter attacks complying with the following \emph{attacker model}: 

\begin{itemize}
\item malware injected in one or more PLCs may forge/drop actuator commands, modify sensor readings coming from the plant, forge/drop inter-controller communications;
\item malware injected in different PLCs of the same field communications network may collaborate/communicate with each other  to achieve common objectives;
  \item the attacker runtime behaviour  may vary as it may depend on the received sensor signals and the communications with other PLCs; 

  \item malicious alterations of sensor signals at network level, or within the sensor devices, are not allowed (they are out of the scope of this paper);

 \item  scan cycles  must be completed within a specific time,  called \emph{maximum cycle limit}, which depends on the controlled physical process; if this time limit is violated then the controller stops and throws an exception~\cite{BLACKHAT2016}; we assume that the injected malware  never  violates  the maximum cycle limit because not interested in causing the immediate shutdown of a  PLC;

 \item the enforcers added in the architecture are  physically independent \emph{secure proxies} with no Internet or USB access, and connected with the controller via  \emph{secure channels};  as a consequence, the measurements transmitted to the \emph{supervisory control network} will not be corrupted; 

   \item runtime reconfigurations of  secure proxies are not allowed.

\end{itemize}
Thus, in general, the attacker objectives can be resumed in alteration/forgery of  PLC actuator commands and/or  communication messages between PLCs to eventually  affect the 
evolution \nolinebreak  of \nolinebreak  the controlled physical processes, and/or  transmit fake signals to the supervisory control network.

	
}

\section{A Formal language for monitored controllers}
\label{sec:calculus}

In this section, we introduce the \emph{Timed Calculus of Monitored Controllers}, called \cname{},  as an abstract formal language to express networks of controllers  integrated with edit automata sitting on the network interface of each controller  to monitor/correct their interactions  with the rest  of  the system.
Basically, \cname{} extends Hennessy and Regan's \emph{Timed Process Language} (TPL)~\cite{HR95} with monitoring edit automata. 
Like TPL 
time proceeds in  \emph{discrete time slots} separated by   $\tick$-actions. 
	
	Let us start with some preliminary notation. We use $s, s_k \in \mathsf{Sens}$ to name \emph{sensor signals};  $a,a_k \in \mathsf{Act}$ to indicate \emph{actuator commands};   $c, c_k \in\mathsf{Chn}$ for \emph{channels};  $z, z_k$ for  \emph{generic names}.

	\paragraph*{Controllers}
	In our setting, controllers are nondeterministic sequential  timed processes evolving through three main  phases: \emph{sensing} of  sensor signals, \emph{communication} with other controllers, and \emph{actuation}. For convenience,  we use five different syntactic categories to distinguish the five main states of a controller: $\mathbbm{Ctrl}$  for  initial states, $\mathbbm{Sleep}$ for sleeping states,  $\mathbbm{Sens}$ for  sensing states, $\mathbbm{Com}$ for  communication states, and $\mathbbm{Act}$ for actuation states. In its initial state, a controller is a recursive process   \emph{waiting} for signal stabilisation in order to start the sensing phase:
	\begin{displaymath}
	\begin{array}{rcl}
	  \mathbbm{Ctrl} \ni P & \Bdf &   X \\ 
          \mathbbm{Sleep} \ni W  & \Bdf & \tick.W  \Bor     S
	\end{array}
	\end{displaymath}

        The main process describing a controller consists of some 
        \emph{recursive process}  defined via  equations of the form $X = \tick.W$, with $ W \in \mathbbm{Sleep}$; here, $X$ is a \emph{process variable} that may occur (free) in $W$. 
        For convenience,  our controllers always have at least one  initial timed action $\tick$  to ensure \emph{time-guarded recursion}, {\color{black} thus avoiding undesired \emph{zeno behaviours}~\cite{Zeno}: the number of untimed actions between
        two $\tick$-actions is always finite.} Then, after a determined sleeping period, when sensor signals get stable,  the sensing phase can start. 
	
	During the sensing phase, the controller waits  for a \emph{finite} number of admissible sensor signals. If none of those signals arrives in the current time slot then the controller will \emph{timeout} moving to the following time slot (we adopt the TPL construct $\timeout{\cdot}{\cdot}$ for timeout).  The syntax is the following:
	\begin{displaymath}
	\begin{array}{rcl}
	\mathbbm{Sens} \ni S   & \Bdf & \timeout{\sum_{i \in I} s_i.S_i}{S} 
	\Bor C 
	\end{array}
	\end{displaymath}

        \noindent 
	where $\sum_{i \in I} s_i.S_i$  denotes the standard construct for nondeterministic choice. 
	Once the sensing phase is concluded, the controller starts its calculations that may  depend  on \emph{communications} with other controllers governing different physical processes. Controllers communicate with each other for mainly two reasons: either to receive notice about the state of other physical sub-processes or to require an actuation on a  physical process which is out of their control. As in TPL, we adopt a \emph{channel-based} \emph{handshake point-to-point}  communication paradigm. Note that, in order to avoid starvation,  communication is always under timeout. The syntax for the communication phase is:
        %
	\begin{displaymath}
	\begin{array}{rcl}
	\mathbbm{Comm} \ni C & \Bdf & \timeout{\sum_{i \in I}c_i.C_i}{C}\Bor \timeout{\overline{c}.C}{C} \Bor A 
	\end{array}
	\end{displaymath}
        
	In the actuation phase a controller eventually transmits a \emph{finite} sequence of commands to  actuators, and then, it emits a \emph{fictitious system signal}  $\fineC$ to denote the end of the scan cycle. After that, the whole scan cycle can restart. Formally,
	\begin{displaymath}
	\begin{array}{rcl}
	  \mathbbm{Act} \ni A  & \Bdf & \overline{a}.A  \Bor  \fineC.X
	\end{array}
	\end{displaymath}

        \begin{remark}[Scan cycle duration and maximum cycle limit]
          \label{rem:maximum-time}

          The scan cycle of a PLC must be completed within a specific time,  called \emph{maximum cycle limit}, which depends on the controlled physical process; if this time limit is violated  the the controller stops and throws an exception~\cite{BLACKHAT2016}. Thus, the signal $\fineC$ must occur well before the \emph{maximum cycle limit} of the controller. We actually work under the assumption that our controllers successfully complete their scan cycle in
          less than half of the maximum cycle limit. The reasons for this assumption will be clarified in Remark~\ref{rem:mitigation-time}. {\color{black}Please, notice that it is easy to statically derive the maximum duration of a scan cycle expressed in our calculus by simply counting the maximum number of $\tick$-prefixes occurring between two subsequent $\fineC$-prefixes. }
        \end{remark}

\begin{table}[t]
		{\small 
			\begin{displaymath}
			\begin{array}{l@{\hspace*{8mm}}l} 
			\Txiom{Sleep}
			{ - }
			{ {  \tick.W \trans{\tick}  W }}
			&
			\Txiom{Rec}
			{ X = \tick.W }
			{ {  X \trans{\tick}  W }}
			\\[10pt]
			\Txiom{ReadS}
			{j \in I}  
			{ \timeout{\sum_{i\in I} s_i.S_i}{S} \trans{s_j}  {S_j}}
			&
			\Txiom{TimeoutS}
			{ - }  
			{ \timeout{\sum_{i\in I} s_i.S_i}{S} \trans{\tick}  {S}}
			\\[10pt]
			\Txiom{InC}
			{j \in I}  
			{ \timeout{\sum_{i\in I} c_i.C_i}{C} \trans{c_j}  {C_j}}
			
			&
			\Txiom{TimeoutInC}
			{- }  
			{ \timeout{\sum_{i\in I} c_i.C_i}{C} \trans{\tick}  {C} 
			}
			\\[10pt]
			\Txiom{OutC}
			{ - }  
			{\timeout{\overline{c}.C}{C'} \trans{\overline{c}}  {C}}
			&
			\Txiom{TimeoutOutC}
			{- }  
			{\timeout{\overline{c}.C}{C'} \trans{\tick}  {C'}
			}
			\\[10pt]
			\Txiom{WriteA}
			{-}
			{ { \overline{a}.A  \trans{\overline{a}}   A}}
			&
			\Txiom{End}
			{-}
			{ { \fineC.X \trans{\fineC}  X}}
			\end{array}
			\end{displaymath}
		}
		\caption{LTS for controllers.}
		\label{tab:sem-ctrl}
	\end{table}

	\enlargethispage{.3\baselineskip}
	The operational semantics  in Table~\ref{tab:sem-ctrl} is along the lines of Hennessy and Regan's TPL~\cite{HR95}.
	In the following, we use the metavariable $\alpha$   to range over the set of all \emph{(observable) controller actions}: $\{   s, \overline{a},  \overline{c}, c,   \tick, \fineC \}  $. These actions denote:  sensor  readings, actuator commands, channel transmissions, channel receptions,    passage of time, and end of  scan cycles, respectively.   {\color{black} Notice that at our level of abstraction we represent only the observable behaviour of PLCs:  internal computations are not modelled 
          within PLCs; although, we do have $\tau$-actions to express communications between  two   PLCs, as the reader will notice in Table~\ref{tab:fnet}. 
            
        \begin{remark}[Attacker model and $\fineC$-signal]
          \label{rem:end-action}
          In our abstract representation of PLCs, the $\fineC$-signal is not really part of the (possibly compromised) PLC program but it is rather a system signal denoting the end of a scan cycle. As a consequence, in accordance with our attacker model, we assume that this fictitious signal cannot be dropped or forged by the attacker.
        \end{remark}
        }

	\paragraph*{Monitored controllers}

        The core of our 
        enforcement relies on (timed)
        finite-state Ligatti et al.'s \emph{edit automata}~\cite{Ligatti2005}, \emph{i.e.}, a particular class of automata specifically designed to allow/suppress/insert actions in a generic system in order to preserve its correct behaviour. The syntax is as follows:
	\begin{displaymath}
{\small
	\begin{array}{rcl}
	\mathbbm{Edit} \ni \Edit, \Fdit & \Bdf &  \go  \Bor  
	\sum_{i \in I} \lambda_i.\Edit_i \Bor \mathsf X 
	\end{array}
}
	\end{displaymath}

%
  
        The special automaton $\go$ will admit any action of the monitored system. 
        The edit automaton $\sum_{i \in I} \lambda_i.\Edit_i$  enforces an action   $\lambda_i$, and then continues as $\Edit_i$, for any $i \in I$, with $I$ finite. Here,  the symbol $\lambda$ ranges over: (i) $\alpha$ to \emph{allow} the action $\alpha$, (ii) $^-\alpha$ to \emph{suppress} the action $\alpha$, and (iii) $\eact{\alpha_2}{\alpha_1}$, for $\alpha_1 \neq \alpha_2$, to \emph{insert} the action $\alpha_1$ before the action $\alpha_2$. \emph{Recursive automata} $\mathsf X$ are defined via equations of the form $\mathsf X = \Edit$, where the automata variable  $\mathsf X$ may occur (free) in $\Edit$.

 The   operational semantics of our  edit automata is given via the following transition rules:
	\begin{displaymath}
	{\small
		\begin{array}{c}
		\Txiom{Go}
		{-}
		{ { \go  \trans{\alpha}  \go }}
		\Q\Q
		\Txiom{Edit}
		{j \in I}
		{ \sum_{i \in I} \lambda_i.\Edit_i  \trans{\lambda_j}   \Edit_j }
		\Q\Q
		\Txiom{recE}
		{ {\mathsf X} = \Edit \Q \Edit \trans{\lambda}  {\Edit'} }  
		{  \mathsf X \trans{\lambda}  {\Edit'}}
		\end{array}
	}
	\end{displaymath}

	Our \emph{monitored controllers}, written $\confCPS{\Edit}{J}$, consist of a  controller $J$,  for $J \in \mathbbm{Ctrl} \cup \mathbbm{Sleep} \cup  \mathbbm{Sens}  \cup  \mathbbm{Comm}  \cup  \mathbbm{Act} $, and an edit automaton $\Edit$ enforcing the behaviour of  $J$,   according to the 
        following transition rules, presented in the style of~\cite{MartinelliMatteucci2007}: 
\begin{displaymath}
          	{\small
	          \begin{array}{c}
                    \! \! \!
	\Txiom{Allow}
	{ \Edit \trans{\alpha}  \Edit' \q J \trans{\alpha} J' }
	{  \confCPS {\Edit}{J} \trans{\alpha}  \confCPS {\Edit'}{J'}}
	\q\,
	\Txiom{Suppress}
	{ \Edit \trans{^{-}\alpha}  \Edit' \q J \trans{\alpha} J' }
	{  \confCPS {\Edit}{J} \trans{\tau}  \confCPS {\Edit'}{J'}} \q\, 
        {\color{black}
	\Txiom{Insert}
	      {  \Edit \trans{\eact{\,\alpha_2}{\alpha_1}}  \Edit' \q  J \trans{\alpha_2} J'   
                }
	      {  \confCPS {\Edit}{J} \trans{\alpha_1}  \confCPS {\Edit'}{J}
              } }	
	\end{array}
}
\end{displaymath}

  Rule \rulename{Allow} is used for allowing observable actions  emitted by the controller under scrutiny.
  By an application of Rule \rulename{Suppress}, incorrect actions $\alpha$  emitted by  (possibly corrupted) controllers $J$ are suppressed, \emph{i.e.}, converted into (silent) $\tau$-actions. Rule \rulename{Insert} is used to insert an action $\alpha_1$ before an action $\alpha_2$ of the controller.
  {\color{black}In the following, the metavariable  $\beta$ will range over the same set of actions as $\alpha$, together with the  \emph{silent action} $\tau$.}

        {\color{black}  
          Here, we wish to stress that, like  Ligatti et al.~\cite{Ligatti2005}, we are interested in  deterministic (and hence implementable) enforcement.
          With the following technical definitions we extract from  enforcer actions $\lambda$ both: (i) the controller triggering actions, and (ii) the resulting output actions. 
                      	\begin{definition}
                          Let $\lambda \in \{\allowE{\alpha}, \suppressE{\alpha}, \insertE{\alpha_1}{\alpha_2}  \}$
                          be an arbitrary action for  edit automata, we write     $\trigger{\lambda}$  to denote the  controller action triggering $\lambda$, defined as: 
$\trigger{\allowE{\alpha}} = \alpha$,  $\trigger{\suppressE{\alpha}}=  \alpha$ and 
                          $\trigger{\insertE{\alpha_1}{\alpha_2}} =  \alpha_2$. Similarly,
                          we write     $\mathit{out}({\lambda}$)  to denote the output action prescribed by $\lambda$, defined as:  
$\mathit{out}({\allowE{\alpha}}) = \alpha$,  $\mathit{out}({\suppressE{\alpha}})=  \tau$ and 
                          $\mathit{out}({\insertE{\alpha_1}{\alpha_2}}) =  \alpha_1$. Given a trace $t= \lambda_1 \cdots \lambda_n$, we write $\mathit{out}(t)$ for the trace $\mathit{out}(\lambda_1) \cdots   \mathit{out}(\lambda_n)$.
                     	\end{definition}
                        Now, we provide a definition of deterministic enforcer along the lines of Pinisetty at al.~\cite{Pinisetty_2017}.
\begin{definition}[Deterministic enforcer]
		\label{def:semantic-deterministic-enfrocement}
	A edit automaton $\Edit\in\mathbbm{Edit}$ is said to be \emph{deterministic} iff 
        in every term $\sum_{i \in I} \lambda_i.\Edit_i $  that appears in $\Edit$
        there are no $\lambda_k$ and  $\lambda_j$, for $k,j \in I$, $k\neq j$, such that
                $\trigger{\lambda_k} =  \trigger{\lambda_j}$ and 
        $\mathit{out}({\lambda_k}) =  \mathit{out}({\lambda_j})$. 
\end{definition}
}


	Finally, we can easily generalise the concept of monitored controller to a \emph{field communications network} of parallel monitored controllers, each one acting on different actuators, 
	and  exchanging information via channels. These networks are formally defined via a straightforward grammar:
	\[
	\mathbbm{FNet} \ni \Sys \Bdf \confCPS{\Edit}{J} \Bor \Sys \parallel \Sys 
	\]
	with the  operational semantics defined in Table~\ref{tab:fnet}.
        
        \begin{table}
	  \begin{displaymath}
            {\small
			\begin{array}{ll}
			\Txiom{ParL}
			{ \Sys_1 \trans{\alpha} \Sys_1'}  
			{ \Sys_1 \parallel \Sys_2 \trans{\alpha} \Sys_1' \parallel \Sys_2}
                        &
                        	\Txiom{ChnSync}
			{\Sys_1 \trans{c} \Sys_1' \q\, \Sys_2 \trans{\overline{c}} \Sys_2'}  
			{ \Sys_1 \parallel \Sys_2 \trans{\tau} \Sys_1' \parallel \Sys_2'\\
				\Sys_2 \parallel \Sys_1 \trans{\tau} \Sys_2' \parallel \Sys_1'}
                  
				\\[20pt]
			\Txiom{ParR}
			{ \Sys_2 \trans{\alpha} \Sys_2'}  
			{ \Sys_1 \parallel \Sys_2 \trans{\alpha} \Sys_1 \parallel \Sys_2'}
			&
		
			\Txiom{TimeSync}
			{\Sys_1 \trans{\tick} \Sys_1' \Q \Sys_2 \trans{\tick} \Sys_2' \Q  \Sys_1 \parallel \Sys_2 \ntrans{\tau}}  
			{ \Sys_1 \parallel \Sys_2 \trans{\tick} \Sys_1' \parallel \Sys_2'}
			\end{array}
                        	}	
	  \end{displaymath}
          \caption{LTS for field communications networks of monitored controllers.}
                    \label{tab:fnet}
         \end{table}

          Notice that monitored controllers may interact with each other via channel synchronisation (see Rule \rulename{ChnSync}). Moreover, via rule \rulename{TimeSync} they may evolve in time only when channel synchronisation may not occur (our controllers do not admit zeno behaviours). This ensures  \emph{maximal progress}~\cite{HR95}, a desirable time property  when modelling real-time systems: channel communications will never be postponed.

          {\color{black}
	\begin{definition}[Execution traces]
	  Given three finite execution traces {\small $t_{\mathrm{c}}=\alpha_1 \ldots \alpha_k$,  $t_{\mathrm{e}}=\lambda_1 \ldots \lambda_l$, and  $t_{\mathrm{m}}=\beta_1 \ldots \beta_n$}, for controllers, edit automata and monitored controllers, respectively. We  write: (i) {\small $P \trans{t_{\mathrm{c}} } P'$}, as an abbreviation for {\small $P=P_0\trans{\alpha_1} \cdots \trans{\alpha_k}P_k=P'$}; (ii) {\small $\Edit \trans{t_{\mathrm{e}}} \Edit'$}, as an abbreviation for {\small $\Edit=\Edit_0\trans{\lambda_1} \cdots \trans{\lambda_l}\Edit_l=\Edit'$}; (iii)
          {\small $N \trans{t_{\mathrm{m}}} N'$}, as an abbreviation for {\small $N=N_0\trans{\beta_1} \cdots \trans{\beta_n}N_n=N'$}.
        \end{definition}
        }

        In the rest of the paper we adopt the following notations. 
        \begin{notation}
           \label{not:trace}
                As usual, we write $\epsilon$ to denote the  \emph{empty trace}.  
                Given a  trace $t$ we write $|t|$ to denote the  \emph{length} of $t$, \emph{i.e.}, the number of actions occurring in $t$. Given a  trace $t$ we write $\hat{t}$    to denote the trace obtained by removing the
$\tau$-actions. 
                Given two traces $t'$ and $t''$, we write  $t' \cdot t''$ for the trace resulting from the  \emph{concatenation} of $t'$ and $t''$. 
                For $t= t' \cdot t''$ we say that  $t'$ is a  \emph{prefix} of $t$ and $t''$ is a \emph{suffix} of $t$.  
 \end{notation}

	\section{Use case: the SWaT System}
        	\label{sec:case-study}

	In this section,  we describe 
	how to specify in \cname{} a non-trivial network of PLCs to control 
	(a simplified version of) the \emph{Secure Water Treatment system} (SWaT)~\cite{SWaT}.

	SWaT represents a scaled down version of a real-world industrial water treatment plant. The system consists of $6$ stages, each of which deals with a different treatment, including:  chemical dosing, filtration,  dechlorination, and reverse osmosis.  
	For simplicity, in our use case, depicted in Figure~\ref{fig:case_study}, we consider only three stages. In the first stage, raw water is \emph{chemically dosed} and pumped in a tank $T_1$, via two pumps $\mathit{pump}_1$ and $\mathit{pump}_2$. A \emph{valve} connects $T_1$ with a \emph{filtration unit} that releases the treated water in a second tank $T_2$.  Here, we assume that the flow of the incoming water in $T_1$ is greater than the outgoing flow passing through the valve. The water in $T_2$  flows into a  \emph{reverse osmosis unit} to reduce inorganic impurities. In the  last stage, the  water coming from the reverse osmosis unit is either distributed as clean water, if required standards are met, or stored in a backwash tank $T_3$ and then pumped back, via a pump $\mathit{pump}_3$, to the filtration unit. Here, we assume that tank $T_2$ is large enough to receive the whole content of tank $T_3$ at any moment.

		\begin{figure}[t]
	\centering
	\includegraphics[width=10cm,keepaspectratio=true,angle=0]{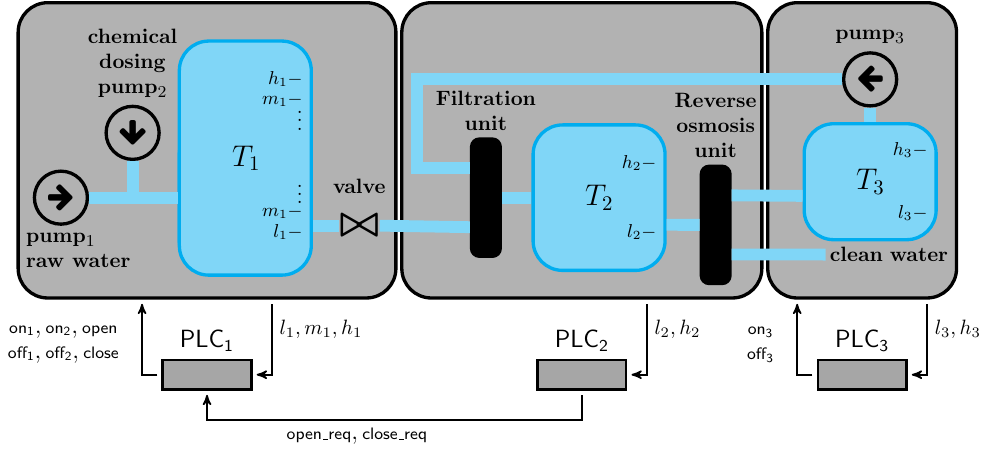}
	\caption{A simplified Industrial Water Treatment System.}
	\label{fig:case_study}
\end{figure}
        
	The SWaT system has been used to provide a dataset containing physical and network data recorded during $11$ days of activity~\cite{swat-dataset}. Part of this dataset contains information about the execution of the system in isolation, while a second part records the effects on the system when exposed to different kinds of cyber-physical attacks. Thus, for instance, 
     (i) \emph{drops} of commands to activate $\textit{pump}_2$ may affect the  quality of the water, as they would affect the correct functioning of  the chemical dosing pump; 
	(ii) \emph{injections} of commands   to close the valve between $T_1$ and $T_2$,  may  give rise to an overflow of tank $T_1$ if this tank is full;  
	(iii) \emph{integrity attacks} on the signals coming from the sensor of the tank $T_3$ 
        may result in damages of the pump  $\mathit{pump}_3$ if it is activated when $T_3$ is  empty.

	Each tank is controlled by its  own PLC. The  programs of the three PLCs, expressed in terms of ladder logic, are given in Figure \ref{fig:ladder-logic-plcs}. 
        In the following, we give their descriptions in 
        \cname{}. 

        Let us start with the code of the controller $\mathrm{PLC}_1$ managing the tank $T_1$. Its definition is given in terms of two  equations  to deal with the case when the two pumps, $\mathit{pump}_1$ and $\mathit{pump}_2$, are both off and both on, respectively. Intuitively, when the pumps are off, the level of water in $T_1$ drops until it reaches its low level  (event $l_1$); when this happens  both pumps are turned on and they remain so until the tank is refilled, reaching its high level (event $h_1$). Formally, 
		\begin{displaymath}
	{\footnotesize
		\begin{array}{l}
		{P}^{\off{ }}_1	   = \tick.  \big(  
		\lfloor l_1. 
		\overline{\on{1}}.\overline{\on{2}}.\overline{\close}.\fineC.{P^{\on{ }}_1}
		\\ [1pt]
		\Q\Q\Q\Q\q	   + \, m_1.\timeout{\ask{open}. \overline{\off{1}}.\overline{\off{2}}.\overline{\open }.\fineC.{P^{\off{ }}_1} +  \ask{close}. \overline{\off{1}}.\overline{\off{2}}.\overline{\close}.\fineC.{P^{\off{ }}_1}}{(\overline{\off{1}}.\overline{\off{2}}.\overline{\close}.\fineC.{P^{\off{ }}_1})}\\ [1pt]
		\Q\Q\Q\Q\q		
		+ \, 	h_1.\timeout{
			\ask{\open}.\overline{\off{1}}.\overline{\off{2}}.\overline{\open}.\fineC.{P^{\off{ }}_1}
			+
			\ask{\close}.\overline{\off{1}}.\overline{\off{2}}.\overline{\close}.\fineC.{P^{\off{ }}_1}
		}{(\overline{\off{1}}.\overline{\off{2}}.\overline{\close}.\fineC.{P^{\off{ }}_1})}
		\\ [1pt]
		\Q\Q\Q\Q
		\rfloor{(\overline{\off{1}}.\overline{\off{2}}.\overline{\close}.\fineC.{P^{\off{ }}_1})}\big)
		\\ [2pt]
		
		{P}^{\on{ }}_1	   = \tick.  \big(  
		\lfloor l_1. 
		\overline{\on{1}}.\overline{\on{2}}.\overline{\close}.\fineC.{P^{\on{ }}_1}
		\\ [1pt]
		\Q\Q\Q\Q\q	   + \, m_1.\timeout{\ask{open}. \overline{\on{1}}.\overline{\on{2}}.\overline{\open }.\fineC.{P^{\on{ }}_1} +  \ask{close}. \overline{\on{1}}.\overline{\on{2}}.\overline{\close}.\fineC.{P^{\on{ }}_1}}{(\overline{\on{1}}.\overline{\on{2}}.\overline{\close}.\fineC.{P^{\on{ }}_1})}\\ [1pt]
		\Q\Q\Q\Q\q		
		+ \, 	h_1.\timeout{
			\ask{\open}.\overline{\off{1}}.\overline{\off{2}}.\overline{\open}.\fineC.{P^{\off{ }}_1}
			+
			\ask{\close}.\overline{\off{1}}.\overline{\off{2}}.\overline{\close}.\fineC.{P^{\off{ }}_1}
		}{(\overline{\off{1}}.\overline{\off{2}}.\overline{\close}.\fineC.{P^{\off{ }}_1})}
		\\ [1pt]
		\Q\Q\Q\Q
		\rfloor{(\overline{\off{1}}.\overline{\off{2}}.\overline{\close}.\fineC.{P^{\on{ }}_1})}\big)
		\end{array}
	}
	\end{displaymath}
	Thus, for instance, when the pumps are off the $\mathrm{PLC}_1$ waits for one time slot (to get stable sensor signals) and then checks the water level of the tank $T_1$, distinguishing between three possible states. 
 	If $T_1$ reaches its low level (signal $l_1$) then the pumps are turned on (commands 	$\overline{\on{1}}$ and $\overline{\on{2}}$) and the valve is closed (command 	$\ask{\open}$). Otherwise, if the tank $T_1$ is at some intermediate level  between low and high (signal $m_1$) then $\mathrm{PLC}_1$ listens for  requests arriving from $\mathrm{PLC}_2$ to open/close the valve. Precisely,  if the PLC gets an $\ask{open}$ request then it opens the valve,  letting the water  flow from $T_1$ to $T_2$, otherwise, if it gets a $\ask{close}$ request then it closes the valve; in both cases the pumps remain off.  
	If the level of the tank is  high (signal $h_1$) then the requests of water coming from $\mathrm{PLC}_2$ are served as before, but the two pumps are eventually turned off (commands  $\scriptstyle\overline{\mathsf{off}_1}$ and $\scriptstyle\overline{\mathsf{off}_2}$). 

        	\begin{figure}[t]	
			\centering
		\begin{minipage}{0.27\textwidth}
			\includegraphics[
			width=3.8cm,keepaspectratio=true,angle=0]{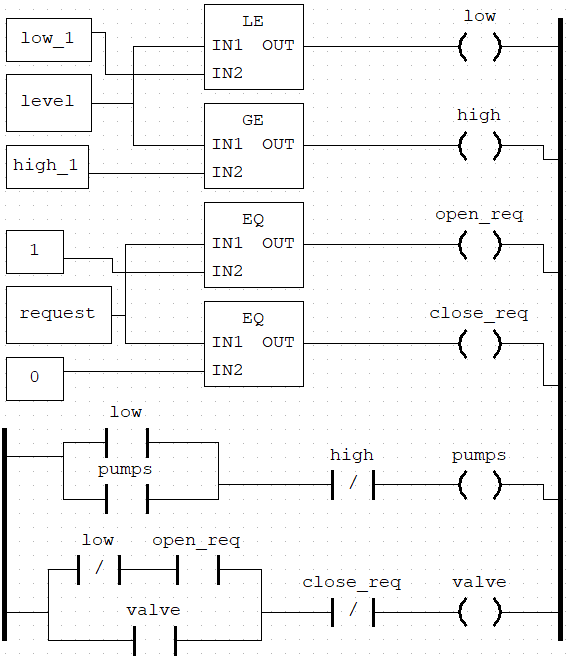}
		\end{minipage}\Q
		\begin{minipage}{0.285\textwidth}
			\includegraphics[
			width=4cm,keepaspectratio=true,angle=0]{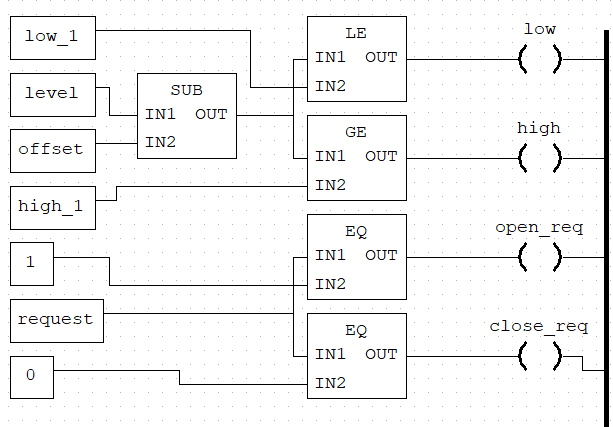}
		\end{minipage}\Q
		\begin{minipage}{0.295\textwidth}
			\includegraphics[
			width=4.55cm,keepaspectratio=true,angle=0]{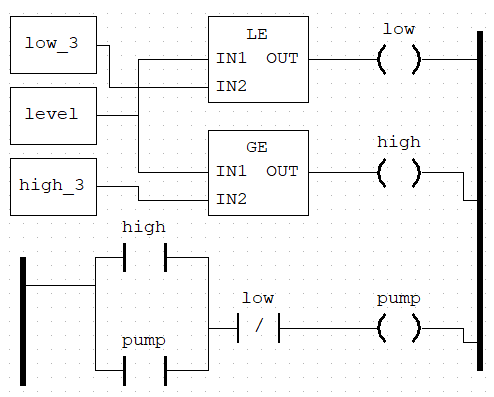}
		\end{minipage}
			\caption{Ladder logics of the three PLCs controlling the system in Figure~\ref{fig:case_study}.}
		\label{fig:ladder-logic-plcs}
	\end{figure}

	$\mathrm{PLC_2}$ manages the water level of tank  $T_2$. Its code consists of  the  two equations to model the filling (state $\uparrow$) and the emptying (state $\downarrow$) of the tank. Formally, 
	\begin{displaymath}
           {\footnotesize 
	\begin{array}{c}
	P^{\Fill}_2 =
	\tick.(  
	\timeout{l_{2}.\timeout{\overline{ \ask{open}}.\fineC.{P^{\Fill}_2} }{\fineC.{P^{\Fill}_2}} 
		\, + \, m_{2}.\timeout{\overline{ \ask{open}}.\fineC.{P^{\Fill}_2} }{\fineC.{P^{\Fill}_2}} 
		\, + \,
		h_{2}.\timeout{\overline{\ask{close}}.\fineC.{P^{\Empty}_2}}{\fineC.{P^{\Fill}_2}}}{\fineC.{P^{\Fill}_2}} )
	\\ [5pt]
	P^{\Empty}_2 =
	\tick.(  
	\timeout{l_{2}.\timeout{\overline{ \ask{open}}.\fineC.{P^{\Fill}_2} }{\fineC.{P^{\Empty}_2}} 
		\, + \, m_{2}.\timeout{\overline{ \ask{close}}.\fineC.{P^{\Empty}_2} }{\fineC.{P^{\Empty}_2}} 
		\, + \,
		h_{2}.\timeout{\overline{\ask{close}}.\fineC.{P^{\Empty}_2}}{\fineC.{P^{\Empty}_2}}}{\fineC.{P^{\Empty}_2}} )
	\end{array}
}            
	\end{displaymath}

	Here, after one time slot, the level of $T_2$ is checked. If the level is low (signal $l_2$) then $\mathrm{PLC}_2$ sends a request to $\mathrm{PLC}_1$, via the channel $\ask{open}$, to open the valve that lets the water  flow from $T_1$ to $T_2$, and then returns. Otherwise, if the level of tank $T_2$ is high (signal $h_2$) then $\mathrm{PLC}_2$ asks  $\mathrm{PLC}_1$ to close the valve, via the channel $\ask{close}$, and then returns. Finally, if the tank $T_2$ is at some intermediate level  between $l_2$ and $h_2$ (signal $m_2$) then the valve remains open (respectively, closed) when the tank is refilling (respectively, emptying). 
	
	Finally, $\mathrm{PLC_3}$ manages  the water level of the backwash tank $T_3$. Its code consists of two equations to deal with the case when the pump $\mathit{pump}_3$ is off and on, respectively. Formally,  
%
        	\begin{displaymath}
	{\footnotesize
		\begin{array}{c}
		P^{\off{ }}_3 \: =  \: \tick. ( 
		\timeout{
			l_3.\overline{\off{3}}.\fineC. P^{\off{ }}_3
			\: + \: 	m_3.\overline{\off{3}}.\fineC. P^{\off{ }}_3
			\: + \: 
			h_3.\overline{\on{3}}.\fineC.P^{\on{ }}_3}{(\off{3}.\fineC.P^{\off{ }}_3)} )
		\\[5pt]
		P^{\on{ }}_3 \: =  \: \tick. ( 
		\timeout{
			l_3.\overline{\off{3}}.\fineC. P^{\off{ }}_3
			\: + \: 	m_3.\overline{\on{3}}.\fineC. P^{\on{ }}_3
			\: + \: 
			h_3.\overline{\on{3}}.\fineC.P^{\on{ }}_3}{(\off{3}.\fineC.P^{\off{ }}_3)} )
		\end{array}
	}
	        \end{displaymath}

        Here, after one time slot, the level of  tank $T_3$ is checked. If the level is low (signal $l_3$) then $\mathrm{PLC}_3$ turns off the pump $\mathit{pump}_3$ (command $\scriptstyle\overline{\off{3}}$), and then returns. Otherwise, if the level of $T_3$ is high (signal $h_3$) then the pump is turned on (command $\scriptstyle\overline{\on{3}}$) until the whole content of $T_3$ is pumped back into the filtration unit of $T_2$.

       {\color{black} \paragraph*{Examples of correctness properties and  attacks}
         In a system similar to that described above, one would expect a number of  properties capturing the correct functioning of system components. Let us provide a few examples of such correctness properties and some specific attacks that may potentially invalidate these properties. 

         A first property might say that if  $\mathrm{PLC}_1$  receives a request to open the valve
         between tanks $T_1$ and $T_2$ then the same valve will be eventually closed early enough to prevent water overflow in tank $T_2$.  This property certainly holds when the system is not exposed to any attack. However,  a malware injected in  $\mathrm{PLC}_1$ might try to undermine this property by tampering either with the actuator dedicated to the valve or with  the requests of  $\mathrm{PLC}_2$  to open/close the valve. 
         In particular, a malicious request to open the valve might be forged by an attacker  injected in $\mathrm{PLC}_2$. Thus, another desired correctness property might say that whenever the tank $T_2$ is full then $\mathrm{PLC}_2$  will never ask for incoming water from tank $T_1$. Finally, another expected property might say that $\mathit{pump}_3$ will never work without enough water in tank $T_3$.
         Again,  an attacker injected in  $\mathrm{PLC}_3$ might try to undermine this property by tampering either with the actuator dedicated to  the pump or with the sensor measuring the level of tank~$T_3$.

         In Section~\ref{sec:attacks} we will provide formal definitions for patterns template of structured correctness properties that are suitable for enforcing correct behaviours of our PLCs. 
}
	%
	%
	\section{A formal language for controller properties}
	\label{sec:logics}

	In this section, we provide a simple description language to express  \emph{correctness  properties} that we may wish to enforce at runtime in our controllers. 
	As discussed in the Introduction, we resort to (a sub-class of) \emph{regular properties} 
        as they allow us to  express interesting classes of properties referring to one or more scan cycles of a controller. 
	The proposed language distinguishes between two kinds of properties: (i) \emph{global properties}, $e \in \PropG$, to express   general controllers' execution traces; (ii) \emph{local properties}, $p \in \PropL$, to express traces confined to a finite number of consecutive scan cycles. 
	The two families of properties are formalised via the following regular  grammar:  

		\begin{displaymath}
	\begin{array}{lcl}
	e \in \PropG  &\Bdf& p^\ast \,|\, e_1 \cap e_2 \\ 
	p,q \in  \PropL &\Bdf&  \epsilon \,|\, p_1; p_2 \,|\,  \cup_{i\in I}\pi_i.p_i  \,|\, p_1 \cap p_2 
	\end{array}
	\end{displaymath}
	        %
	where $\pi_i \in \Events \defn \SensSet \cup \overline{\ActSet} \cup \ChanSet^{\ast} \cup \{\tick\} \cup \{ \fineC \}$ denote \emph{atomic properties},  called \emph{events}, that may occur as prefix of a property. With an abuse of notation, we use the symbol $\epsilon$ to denote both the \emph{empty property} and the \emph{empty trace}. 

        The \emph{semantics} of our logic is naturally defined in terms of sets of execution traces which satisfy a given property; its formal definition is given in Table~\ref{tab:semantics-logic}.

        \begin{table*}[t]
        	\begin{displaymath}
        	{\small 
        		\begin{array}{lcl}
        		
        		\regSemantics{p^\ast} &\defn&
        		\{\epsilon\} \cup \bigcup_{n\in \mathbb{N}^+} \{t \, \mid \,  t =t_1 \cdot \ldots \cdot t_n, \textrm{ with }   t_i \in \regSemantics{p}, \textrm{ for }  1\leq i\leq n \} \\[1pt]
        		\regSemantics{e_1 \cap e_2} &\defn&  \{t\, \mid \,t\in \regSemantics{e_1}  \text{ and } t \in \regSemantics{e_2}\}\\[1pt]
        		\regSemantics{\epsilon} &\defn& 
        		{\{ \epsilon \}}\\[1pt]
        		\regSemantics{p_1\cap p_2} &\defn& 	
        		\{t \, \mid \, t\in \regSemantics{p_1}  \text{ and } t \in \regSemantics{p_2} \}
        		\\[1pt]
        		\regSemantics{p_1;p_2} &\defn& 	
        		\{t \, \mid \, t=t_1\cdot t_2, \textrm{ with  $t_1 \in \regSemantics{p_1}$  and  $t_2 \in \regSemantics{p_2}$} \}
        		\\[1pt]
        		\regSemantics{\bigcup_{i\in I}\pi_i.p_i } &\defn& \bigcup_{i\in I} 
        		\{t \, \mid \, t=\pi_i \cdot t',  \textrm{ with $t'\in \regSemantics{p_i}$}\}                
        		\end{array}
        	}
        	\end{displaymath}
        	\caption{Trace semantics of our regular properties.}
        	\label{tab:semantics-logic}
        \end{table*}
     
   However, the syntax of our logic is a bit too permissive with respect to our intentions, as it allows us to describe partial scan  cycles,  \emph{i.e.}, cycles that have not completed. Thus, we restrict ourselves to considering properties which builds  on top of local properties associated to  \emph{complete scan cycles},  \emph{i.e.}, scan cycles whose last action is an  $\fineC$-action. Formally, 
   \begin{definition}
     \label{def:well-formedness}
     \emph{Well-formed properties} are  defined as follows: 
     \begin{itemize}
			\item the local property $\fineC. \epsilon$ is well formed; 
			\item a local property of the form $p_1 ; p_2$ is well formed if $p_2$ is well formed; 
			\item a local property of the form $p_1 \cap p_2$ is well formed if both $p_1$ and $p_2$ are well formed; 
			\item a local property of the form $\cup_{i\in I}\pi_i.p_i$ 
			is well formed  if   either $\pi_i . p_i = \fineC. \epsilon$ or $p_i$ is well formed, for any $i \in I$; 
			\item a global property $p^\ast$ is well formed if  $p$ is well-formed;
			\item a global property $e_1\cap e_2$ is well-formed if both $e_1$ and $e_2$ are well-formed. 
	\end{itemize}
	\end{definition}
   %
        In the rest of the paper, we always assume to work with well-formed properties.  Moreover, we adopt the following notations and/or abbreviations on properties. 
\begin{notation}
  We omit trailing empty properties, writing $\pi$ instead of $\pi.\epsilon$.
  For $k > 0$, we write $\pi^k.p$ as a shorthand for $\pi.\pi...\pi.p$, where prefix $\pi$ appears $k$ consecutive times. Given a local property $p$ we write $\events p \subseteq \Events$ to denote the set of events occurring in $p$; similarly, we write $\events e \subseteq \Events$ to denote the set of events occurring in a global property $e \in \PropG$. 
  Given a set of events ${\mathcal A} \subseteq \Events$ and a local property $p$,  we use ${\mathcal A}$ itself as an abbreviation for the property $\cup_{\pi \in {\mathcal A}}\pi.\epsilon$, and  ${\mathcal A}.p$  as an abbreviation for the property $\cup_{\pi \in {\mathcal A}}\pi.p$.
		Given a set of events ${\mathcal A}$, with  $\fineC \not \in {\mathcal A}$,
	  	we write ${\mathcal A}^{\le k}$, for $k \geq 0$, to denote the well-formed property defined  as
                follows: (i)
		${\mathcal A}^{\le 0}  \defn  \fineC $; 
          (ii) 	$ 
	 {\mathcal A}^{\le k}  \defn	\fineC  \cup  {\mathcal A}.{\mathcal A}^{\le k-1}$, 
	 for $k> 0$. 
         Thus, the property ${\mathcal A}^{\le k}$  captures all possible sequences  of events of  ${\mathcal A}$ whose length is at most $k$, for $k \in\mathbb{N}$.  We write $\PEvents$ to denote the set of pure events, \emph{i.e.}, $\Events \setminus \{\fineC \}$.
         Finally, we write $\PUEvents$ to denote the set of pure untimed events, \emph{i.e.}, $\Events \setminus \{\fineC, \tick \}$. 
	\end{notation}

{\color{black}
Note that our properties are in general nondeterministic. However, since we are interested in deterministic enforcers, in the following we will focus on deterministic enforcing properties. 
\begin{definition}[Deterministic properties]
  \label{def:deterministic-prop}
        A global property $e\in \PropG$ is said to be deterministic if for any sub-term $\cup_{i \in I}\pi_i.p_i$  appearing in $e$,  we have  $\pi_k \neq \pi_h$, for any $k,h\in I$,  $k\neq h$.
	\end{definition}
}




\subsection{Local properties}
\label{sec:local_prop}
As already said, local properties describe execution traces which are limited to a finite number \nolinebreak of \nolinebreak scan cycles. Let us present a number of significant local properties that can be expressed in our language of regular properties.
 {\color{black}
   In the following, we assume a fixed maximum number of actions, 
    $\maxa$,  that may occur within  a single scan cycle of our controllers, \emph{i.e.}, between two subsequent  $\fineC$-actions. }
\subsubsection{Basic properties}

They prescribe 
conditional, eventual and persistent behaviours.

\paragraph*{Conditional}
These properties say that 
when a (pure) untimed event $\pi$ occurs in the \emph{current scan cycle}  then some property $p$ should be satisfied. 
More generally, for $\pi_i \in \UEvents$ and $p_i\in \PropL$, we write $\caseCND{(\,\cup_{i\in I}\{ \tpl{\pi_i}{p_i}\})}$  to denote the property $q_k$, for $k=\maxa $,  defined as follows: 
\begin{itemize}
\item
   $q_k \defn \fineC \cup \bigcup_{i\in I} \pi_i.p_i \cup\ (\PEvents {\setminus}\bigcup_{i \in I} \{\pi_i\}).q_{k-1} $, for $0 < k \leq \maxa$ 
\item
   $q_0 \defn \fineC$.
\end{itemize}

 When there is only one triggering event 
 $\pi \in \PUEvents$ and one associated local property $p \in \PropL$, we have a simple conditional property defined as follow:   $\CND{\pi}{p} \defn \caseCND{(\{\tpl{\pi}{p}\})}$.

\begin{figure}[t]
	\centering
	\includegraphics[width=0.65\textwidth]{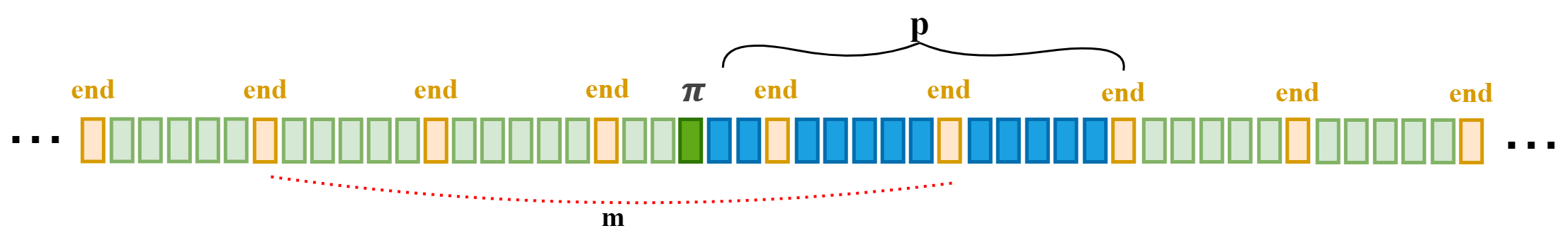}
	\caption{A trace satisfying a persistent conditional property $\PCND{\pi}{p}{m}$.}
	\label{fig:conditional}
\end{figure}

Conditional properties $\CND{\pi}{p}$ define a cause-effect relation in which the triggering event $\pi$ is searched in the current  scan cycle;  one may think of a more general property $\PCND{\pi}{p}{m}$, in which the cause-effect relation  \emph{persists} for $m > 0$ consecutive scan cycles, \emph{i.e.}, the search for the triggering event $\pi$ continues for at most $m$ consecutive scan cycles. Of course, the triggered local property $p$ may span over a finite number of scan cycles (see Figure~\ref{fig:conditional}). 
Formally, we write $\PCND{\pi}{p}{m}$, for $\pi \in \UEvents$, $p\in \PropL$ and $m > 0$, for the  property 
 $q^{m}_\maxa$ defined as follows:
\begin{itemize}
	\item $q^h_k\defn \fineC.q^{h-1}_\maxa \, \cup \,  \pi.p \, \cup \,  (\PEvents{\setminus}\{\pi\}).q^{h}_{k-1}$, for $1 < h \leq m$ and $0 < k \leq \maxa$
	\item $q^{h}_0\defn \fineC.q^{h-1}_\maxa$, for $1 < h \leq m$
	\item $q^{1}_k\defn  \fineC \, \cup \,  \pi.p \, \cup \, (\PEvents{\setminus}\{\pi\}).q^{1}_{k-1}$,   for $0 < k \leq \maxa$ 
          \item  $q^{1}_0 \defn \epsilon$. 
\end{itemize}
Obviously, $\CND{\pi}{p} =  \PCND{\pi}{p}{1}$.

\paragraph*{Bounded eventually}
In this case, \emph{an event $\pi$ must eventually occur within $m$ scan cycles}.
Formally, for $\pi \in \UEvents$ and $m > 0$,  we write ${\BE{\pi}{m}}$ to denote the property $q_\maxa^m $ defined as follows:
\begin{itemize}
	\item $q^h_k\defn \fineC.q_{\maxa}^{h-1} \cup \pi.{\PEvents}^{\leq k-1} \cup (\PEvents {\setminus}\{\pi\}).q_{k-1}^{h}   $, for $1 < h \leq m$ and $0 < k \leq \maxa$
	\item $q_0^h \defn\fineC.q^{h-1}_{\maxa}$, for $1 < h \leq m$ 
	\item $q^1_k\defn \pi.{\PEvents}^{\leq k-1} \cup (\PEvents {\setminus}\{\pi\}).q_{k-1}^{1}$,  for  $0 < k \leq \maxa$ 

	\item $q_0^1 \defn {\pi}.\fineC$.
\end{itemize}

\paragraph*{Bounded persistency}

While in  ${\BE{\pi}{m}}$ 
the event $\pi$ must eventually occur within $m$ scan cycles, bounded persistency prescribes that an  event $\pi$ \emph{must occur in all subsequent $m$ scan cycles}. 
%
Formally,  for $\pi \in \UEvents$ and $m > 0$, we write ${\BP{\pi}{m}}$ to denote the property $q_\maxa^{m} $ defined as follows:
\begin{itemize}
	\item $q_{k}^h \defn {\pi}. \PEvents^{\leq k-1};
	q^{h-1}_\maxa  \cup{(\PEvents{\setminus}\{\pi\})}.
	q_{k-1}^{h}$,  for $1 < h \leq m$ and $0 < k \leq \maxa$
	\item $q_{0}^h\defn \pi.\fineC.q^{h-1}_\maxa$,  for $1 < h \leq m$
	\item $q_{k}^1 \defn {\pi}. \PEvents^{\leq k-1}  \cup{(\PEvents{\setminus}\{\pi\})}.
	q_{k-1}^{1}$,  for $0 < k \leq \maxa$
	\item $q_0^1  \defn \pi.\fineC$.
\end{itemize}

\paragraph*{Bounded absence}
The negative counterpart of bounded persistency is bounded absence. This property says that \emph{an  event $\pi$ must not appear in all subsequent $m$ scan cycles}{.
Formally,  for $\pi \in \UEvents$ and $m> 0$, we write ${\BA{\pi}{m}}$ to denote the property $q_{m} $ defined as follows:
	\begin{itemize}

		\item  $q_{h} \defn (\PEvents{\setminus}\{\pi\})^{\leq \maxa};q_{h-1}$, for $0 < h \leq m$
		\item $q_{0} \defn \epsilon$.
\end{itemize}

\subsubsection{Compound conditional   properties}
The properties above can be combined together to express more detailed PLC behaviours. Let us see a few examples with the help of the use case of Section~\ref{sec:case-study}. 

\paragraph*{Conditional bounded eventually}
According to this property, 
if a triggering event $\pi_1$ occurs then a second event $\pi_2$ must eventually occur 
between the $m$-th  and the $n$-th scan cycle, with $1\leq m \leq n$.
Formally, for  $\pi_1,\pi_2 \in \UEvents$ and $1\leq m \leq n$, we define $\CBE{\pi_1}{\pi_2}{m}{n}$  as follows: 
	 \[ \CBE{\pi_1}{\pi_2}{m}{n} \defn \CND{\pi_1 \, }{(\PEvents^{\leq \maxa})^{m-1};\BE{\pi_2}{n-m+1}}. \]
         Intuitively, if the  event $\pi_1$ occurs then the event $\pi_2$
         must eventually occur between the scan cycles $m$ and $n$. In case we would wish that $\pi_2$ should not occur before the $m$-th scan cycle, then the property would become: $ \CND{\pi_1 \, }{\BA{\pi_2}{m-1};\BE{\pi_2}{n-m+1}}.$ 

As an example,  we might enforce a conditional bounded eventually property in  $\mathrm{PLC}_1$ of our use case in Section~\ref{sec:case-study} to prevent water overflow in the tank $T_2$ due to a misuse of the valve connecting the tanks $T_1$ and $T_2$. Assume that $z \in \mathbb{N}$ is the time (expressed in scan cycles) required to overflow the tank $T_2$ when the valve is open and the level of tank $T_2$ is low.  We might consider to enforce a property of the form  
\begin{math}
\CBE{\ask{open}}{\overline{\close}}{1}{w}
\end{math}, 
with $w < < z$, saying that if  $\mathrm{PLC}_1$ receives a request to open the valve, 
then the valve will be eventually closed
within at most $w$ scan cycles (including the current one). This will ensure that if a water request coming from $\mathrm{PLC}_2$ is received by $\mathrm{PLC}_1$ then the valve controlling the flaw from $T_1$ to $T_2$ will remain open for at most $w$ scan cycles, with $w < < z$,  preventing the overflow of $T_2$. 

\paragraph*{Conditional bounded persistency}
Another possibility is to combine conditional with bounded persistency to prescribe that if a triggering event $\pi_1$ occurs then the event $\pi_2$ must occur in the $m$-th  scan cycle and in all subsequent $n-m$ scan cycles, for $1\leq m \leq n$. 
Formally,  for  $\pi_1,\pi_2 \in \UEvents$ and $1 \leq m \leq n$, we write $\CBP{\pi_1}{\pi_2}{m}{n}$ to denote the property defined as:
\[ \CBP{\pi_1}{\pi_2}{m}{n} \defn \CND{\pi_1 \, }{(\PEvents^{\leq \maxa})^{m-1};\BP{\pi_2}{n-m+1}}. \]


As an example,  we might enforce a conditional bounded persistency property in  $\mathrm{PLC}_3$ of our use case in Section~\ref{sec:case-study} to prevent damages of  $\mathit{pump}_3$ due to lack of water in tank $T_3$.
Assume that $z \in \mathbb{N}$ is the minimum time (in terms of scan cycles) required to fill $T_3$, \emph{i.e.}, to pass from level $l_3$ to level $h_3$, when  $\mathit{pump}_3$ is off.   We might consider to enforce a property of the form 
\begin{math}
\CBP{l_3}{\overline{\off{3}}}{1}{z}
\end{math}, 
 to prescribe that if the tank reaches its  low level (event $l_3$) then  $\mathit{pump}_3$ must remain  off (event $\overline{\off{3}}$) for $z$ consecutive scan cycles. This will ensure enough water in tank $T_3$ to prevent damages on  $\mathit{pump}_3$.

  Notice that all previous properties have a single triggering event $\pi_1$;
 in order to deal with multiple triggering events it is enough to replace the conditional operator with the case construct. 

\paragraph*{Conditional bounded absence (also called Absence timed~\cite{toolchain})} Finally, we might consider to combine conditional with bounded absence to formalise a property saying that if a triggering event $\pi_1$ occurs then another event $\pi_2$  must not occur in 
the $m$-th  scan cycle and in all subsequent $n-m$ scan cycles, with $1\leq m \leq n$.
Formally,  for $\pi_1,\pi_2 \in \UEvents$ and  $1\leq m \leq n$, we write $\CBA{\pi_1}{\pi_2}{m}{n}$ to denote the property defined as follows:
\[ \CBA{\pi_1}{\pi_2}{m}{n} \defn \CND{\pi_1}{\, (\PEvents^{\leq \maxa})^{m-1}; \BA{\pi_2}{n-m+1}}. \]
%
Intuitively, if the triggering event $\pi_1$ occurs then the event $\pi_2$ must not occur in the time interval between  the  $m$-th and the  $n$-th scan cycle. 

As an example,  we might enforce a conditional bounded absence property in  $\mathrm{PLC}_2$ of our use case in Section~\ref{sec:case-study} to prevent water overflow in the tank $T_2$ due to a misuse of the valve connecting the tanks $T_1$ and $T_2$. Assume that $z \in \mathbb{N}$ is the time (expressed in scan cycles) required to empty the tank $T_2$ when the valve is closed and the tank $T_2$ reaches its high level $h_2$. Then,  we might consider to enforce a property of the form 
\begin{math}
\CBA{h_2}{\overline{\ask{open}}}{1}{w}
\end{math}, 
for $w < z$,
to prescribe that if the tank reaches its high level (event $h_2$) then $\mathrm{PLC}_2$ may not send a requests to open the valve (event $\overline{\ask{open}}$) for the subsequent $w$ scan cycles. This ensures us that when $T_2$ reaches its high level then it will not ask for incoming water for at least $w$ scan cycles, so preventing  tank overflow.

\subsubsection{Compound persistent conditional properties}
Now, we formalise in our language of regular properties a number of correctness properties  used by Frehse et al.\ for the verification of hybrid systems~\cite{toolchain}.  More precisely,   we formalise bounded versions of their properties. 
\paragraph*{Bounded minimum duration} When a triggering event $\pi_1$ occurs, if a second event $\pi_2$ occurs within $m$ scan cycles then this event \emph{must} appear in \emph{at least} all subsequent $n$ scan cycles (see Figure~\ref{fig:minimum-duration}). Formally, we can express this property as follows: 
	  \[ \MinD{\pi_1}{\pi_2}{m}{n} \defn \CND{\pi_1}{\PCND{\pi_2}{\BP{\pi_2}{n}}{m} }. \]

 Note that the property  $\MinD{\pi_1}{\pi_2}{m}{n}$ is satisfied each time $\CBP{\pi_1}{\pi_2}{m}{m+n}$ is. The vice versa does not hold as in $\CBP{\pi_1}{\pi_2}{m}{m+n}$ the event $\pi_2$ is required to occur in the whole time interval $[m, m{+}n]$, while, according to $\MinD{\pi_1}{\pi_2}{m}{n}$, the event $\pi_2$ might not occur at all.  
\begin{figure}[t]
	\centering
	\includegraphics[width=0.75\textwidth]{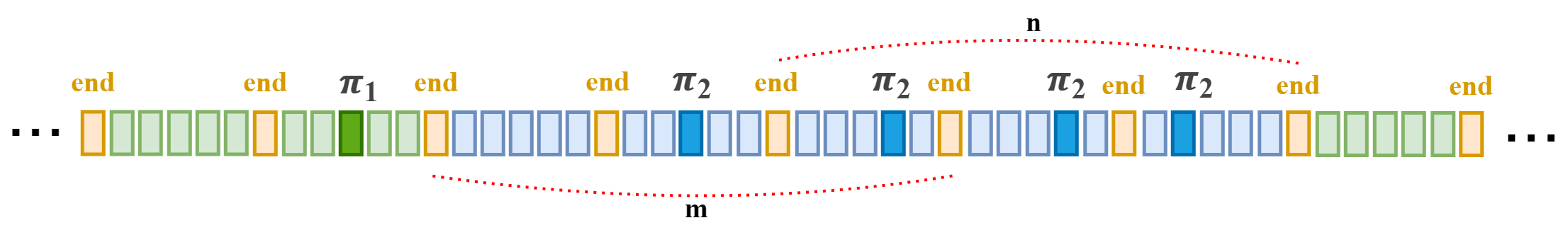}
	\caption{A trace satisfying a minimum duration property 
	 $\MinD{\pi_1}{\pi_2}{m}{n}$, for $m=n=3$.
	}
	\label{fig:minimum-duration}
\end{figure}

\paragraph*{Bounded maximum duration}
When an event $\pi_1$ occurs, if a second event $\pi_2$ occurs within $m$ scan cycles then the same event  $\pi_2$ \emph{may} occur in \emph{at most} all subsequent $n$ scan cycles. 
%
%
 Formally, we can represent this property as follows:
 	 \[ \MaxD{\pi_1}{\pi_2}{m}{n} \defn \CND{\pi_1}{ \PCND{\pi_2}{(\PEvents^{\leq \maxa})^{n};\BA{\pi_2}{1}}{m}}. \]
The property  $\MaxD{\pi_1}{\pi_2}{m}{n}$ is satisfied each time the property $\CBP{\pi_1}{\pi_2}{m}{m+n}; \BA{\pi_2}{1}$ is. Again, the vice versa does not hold. 

\paragraph*{Bounded response} When an event $\pi_1$ occurs, if
a second event $\pi_2$ occurs within $m$ scan cycles then a third event $\pi_3$  appears within $n$ scan cycles. 
Formally, we can express this property as follows: 
	 	\[ \BR{\pi_1}{\pi_2}{\pi_3}{m}{n} \defn \CND{\pi_1}{ \PCND{\pi_2}{ \BE{\pi_3}{n} }{m}}. \]

\paragraph*{Bounded invariance}
Whenever an event $\pi_1$ occurs,  if  $\pi_2$ occurs within $m$ scan cycles then  $\pi_3$ will persistently occur for at least $n$ scan cycles.
%
Formally,  we can express this property as follows:
 	 \[ \BI{\pi_1}{\pi_2}{\pi_3}{m}{n} \defn \CND{\pi_1}{\PCND{\pi_2}{\BP{\pi_3}{n} }{m}}. \]

\subsubsection{Bounded mutual exclusion}
A different class of properties prescribes the possible occurrence  of  events  $\pi_i \in \PEvents$, for $i \in I$, in mutual exclusion within $m$ consecutive scan cycles. 
%
%
%
Formally,  for  $\pi_i\in \UEvents$, $i\in I$ and $m\ge 1$, we write $\BME{(\,\bigcup_{i\in I}\{ \pi_i\} \,)}{m}$, for the  property \nolinebreak $q^{m}_\maxa$ \nolinebreak defined \nolinebreak  as:

%
\begin{itemize}
	\item $q^h_k\defn \fineC.q^{h-1}_\maxa \, \cup \,  \bigcup_{i\in I} \pi_i.(\bigcap_{j\in I\setminus\{i\}} \BA{\pi_j}{h}) \cup (\PEvents {\setminus}\bigcup_{i \in I} \{\pi_i\}).q^{h}_{k-1}$, for $1 < h \leq m$ and $0 < k \leq \maxa$
	\item $q^{h}_0\defn \fineC.q^{h-1}_\maxa$, for $1 < h \leq m$
	\item $q^{1}_k\defn  \fineC  \cup   \bigcup_{i\in I} \pi_i.(\bigcap_{j\in I\setminus\{i\}} \BA{\pi_j}{1}) \cup (\PEvents {\setminus}\bigcup_{i \in I} \{\pi_i\}).q^{1}_{k-1}$,   for $0 < k \leq \maxa$ 
	\item  $q^{1}_0 \defn \epsilon$. 
\end{itemize}
 As an example, we might enforce a bounded mutual exclusion property in the $\mathrm{PLC}_1$ of our use case of Section~\ref{sec:case-study} to prevent   chattering of the valve, \emph{i.e.},  rapid opening and closing which may cause  mechanical failures on the long run.
In particular, we might consider to enforce a property of the form
\begin{math} \BME{(\{ \overline{\open}, \overline{\close}  \})}{3}
\end{math}
saying that  within $3$ consecutive scan cycles the opening and the closing of the  valve (events $\overline{\open}$ and $\overline{\close}$, respectively) may only occur in mutual exclusion. 
 
\smallskip
In Table~\ref{tab:properties-overview}, we summarise all local properties represented and discussed in this section.  

%

\subsection{Global properties}
	As expected, the previously described local properties become global ones by applying the Kleene-operator $\ast$. Once in this form, we can put these properties in conjunction between them. Here, we show two global properties, the first one is built top of  conditional bounded persistency properties and the second one is built on top of a conditional bounded eventually property. 
	
	As a first example, we might consider a global property saying that whenever an event $\pi$ occurs then all events $\pi_i$ , for $i \in I$, must occur in the $m$-th  scan cycle and in all subsequent $n-m$ scan cycles, for $1\leq m \leq n$. 
	Formally, for $\pi,\pi_i \in \UEvents$,  $i\in I$, and $1\leq m\leq n$:
        $\bigcap_{i\in I} (\CBP{\pi}{\pi_i}{m}{n})^\ast$. 
	

\begin{table}[t]
\small
	\begin{center}
		\begin{tabular}{p{0.315\textwidth}p{0.6\textwidth}}
			\hlineB{3}
%
			Case: & if $\pi_i$ occurs then $p_i$ should be satisfied, for $i\in I$ \\
			
			
			Persistent conditional: &  for $m$ scan cycles, if  $\pi$  occurs then $p$  should be satisfied\\
				
			Bounded eventually: &  event $\pi$ must eventually occur within $m$ scan cycles \\
				
			Bounded persistency: &  event $\pi$ must occur in all subsequent $ m$ scan cycles \\
			
			Bounded absence: &  even $\pi$ must not occur in all subsequent $m$ scan cycles \\
			\hline	
			Conditional bounded eventually: & if  $\pi_1$ occurs then $\pi_2$ must eventually occur in the scan cycles $[m,n]$\\

			Conditional bounded persistency: & if $\pi_1$ occurs then  $\pi_2$ must occur in all  scan cycles of $[m,n]$\\

			Conditional bounded absence: & if $\pi_1$ occurs then  $\pi_2$  must not occur in all  scan cycles of $[m, n]$ \\

			\hline	
			(Bounded) Minimum duration: & when $\pi_1$,  if  $\pi_2$ in $[1,m]$  then $\pi_2$ persists for at least  $n$ scan cycles \\[1pt]
			
			(Bounded) Maximum duration: & when  $\pi_1$, if  $\pi_2$ in  $[1,m]$ then $\pi_2$ persists for at most  $n$ scan cycles \\[1pt]
			
			Bounded response: & when $\pi_1$, if $\pi_2$ in $[1,m]$ them $\pi_3$ appears within $n$ scan cycles \\[1pt]
			
			Bounded invariance: & when $\pi_1$, if  $\pi_2$ in $[1,m]$ 
			then $\pi_3$ persists  for at least $n$ scan  cycles \\ [1pt]
			\hline
			Bounded mutual exclusion & events $\pi_i$ may  only occur in mutual exclusion within $n$  scan cycles \\
			
			
			


			\hlineB{3}
		\end{tabular}
	\end{center}	
	\caption{Overview of  local properties.}
	\label{tab:properties-overview}
\end{table}

We might enforce this kind of property in  $\mathrm{PLC}_1$ of our use case of Section~\ref{sec:case-study}.
Assume  $z\in \mathbb{N}$ being the time (expressed in scan cycles) required to overflow the tank $T_1$ when the level of the tank $T_1$ is low and both pumps are on and the valve is closed.
Then, the property would be 
\begin{math}
(\CBP{l_1}{\overline{\on{1}}}{1}{w})^\ast\cap(\CBP{l_1}{\overline{\on{2}}}{1}{w})^\ast
\end{math},
with $w< z$,
saying that if the tank $T_1$ reaches its low level (event $l_1$) then both $\mathit{pump}_1$ and $\mathit{pump}_2$ must be  on (events $\overline{\on{1}}$ and $\overline{\on{2}}$) in all subsequent $w$ scan cycles, starting from the current one.

As a second example, we might consider a more involved global property relying on conditional bounded eventually, persistent conditional, and bounded persistency.   Basically, the property says that whenever an event $\pi_1$ occurs then a second event $\pi_2$ must eventually occur between the $m$-th scan cycle and the $n$-th scan cycle, with $1\leq m \leq n$; moreover, it must occur for $d$ consecutive scan cycles,  for $1 \leq d$ (see Figure~\ref{fig:cbe-ext2}).
%
Formally, the property is the following: 
		\[	\big(\CBE{\pi_1}{\pi_2}{m}{n}\big)^\ast \cap \, \big( \CND{\pi_1}{\PCND{\pi_2}{\PEvents^{\leq \maxa}; \BP{\pi_2}{d-1}}{n}}\big)^\ast \]
        for $\pi_1,\pi_2 \in \UEvents$, with $1\leq m\leq n$ and $d \ge 1$. Intuitively, the  property $(\CBE{\pi_1}{\pi_2}{m}{n})^\ast$ requires that when $\pi_1$ occurs the event $\pi_2$ must eventually occur between the $m$-th scan cycle and the $n$-th scan cycle. The remaining part of the property says  if the event $\pi_2$ occurs within the $n$-th scan cycle  (recall that $m \leq n$) then it must persist for $d$ scan cycles.

\begin{figure}[t]
	\centering
	\includegraphics[width=0.75\textwidth]{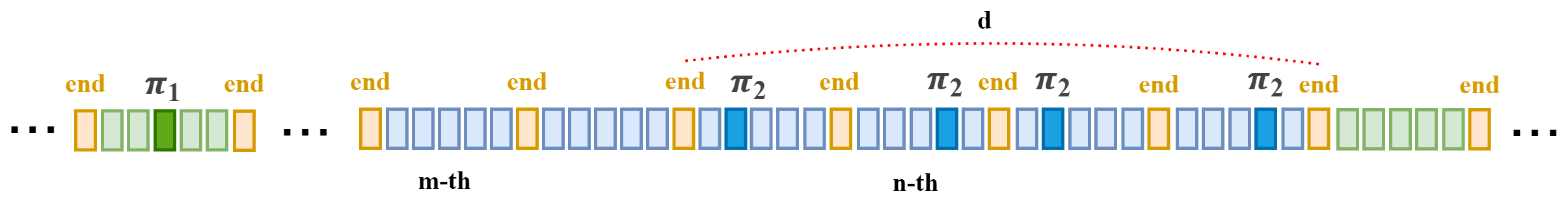}
	\caption{A trace satisfying the just mentioned property 
		for some $m$, $n = m+4$ and $d=4$.
	}
	\label{fig:cbe-ext2}
\end{figure}

In this manner,  we might strengthen the conditional bounded eventually property given in Section~\ref{sec:local_prop} for   $\mathrm{PLC}_1$ of our use case 
to prevent water overflow in the tank $T_2$.
Let  $z \in \mathbb{N}$ be the time (expressed in scan cycles) required to overflow the tank $T_2$ when the valve is open and the level of tank $T_2$ is low. The property is the following: 
\begin{displaymath}{\small
		\big(\CBE{\ask{open}}{\overline{\close}}{1}{w} \big)^\ast \, \cap \, \big( \CND{\ask{open}}{\PCND{\overline{\close}}{\PEvents^{\leq \maxa}; \BP{\overline{\close}}{d-1}}{w}}\big)^\ast 
	}
\end{displaymath}
where $w < < z$, and $d \in \mathbb{N}$ is the time (expressed in scan cycles) required to release in $T_3$ the (maximum) quantity of water that the tank $T_2$ may accumulate in $w$ scan cycles.
The first part of the property says that if  $\mathrm{PLC}_1$  receives a request to open the valve (event $\ask{open}$)  then the valve must be eventually closed (event $\overline{\close}$ must eventually occur) within at most $w$ scan cycles. The remaining part of the property says that  when  $\mathrm{PLC}_1$ receives a request to open the valve (event $\ask{open}$), if the valve gets closed (event $\overline{\close}$) within the $w$-th scan cycle, then it must remain closed for the  $d$ consecutive scan cycles.
Here, $d$ depends both on  the maximum level of water  reachable in $T_2$ in $w$ scan cycles and on the physical law governing the water flow from $T_2$ to $T_3$.

	\section{Monitor synthesis}
	\label{sec:synthesis}

	In this section, we provide an algorithm to synthesise monitors from  regular properties  whose events are contained in (the set  of events associated to) a fixed set   $ \PSet$ of   observable controller actions.
        More precisely, given  
        a global property $e \in \PropG$
       the algorithm  returns an edit automaton $\synthES{\! e \!}{}^{\PSet} \in \mathbbm{Edit}$
    that is capable to enforce 
    the  property $e$  during  the execution of a generic controller whose possible actions are confined to those in  $\PSet$. 
    The synthesis algorithm is defined in Table~\ref{tab:synthesis-logic} by induction on the structure of the global/local property given in input; 	as we distinguish global  properties from local ones,  we define our  algorithm in two steps.
    {\color{black}
      \begin{remark}
        We recall that, according to the operational semantics defined in Table~\ref{tab:sem-ctrl},  all controller actions $\alpha$  are observable and they basically coincide with the set $ \Events$ used to build up the enforcing properties defined in Section~\ref{sec:logics}. As a consequence,  we will synthesise enforcing monitors that may observe any action of the controller under scrutiny and may act consequently.
        \end{remark}
    }

    The monitor  $\synthES{\! p^{\ast} \!}{}^{\PSet}$ associated to a global property $p^{\ast}$ 
    is an edit automaton defined via the recursive equation $\mathsf X = \synthESP{p}{\mathsf X}{\PSet}$, to
recursively enforce the local property $p$. The monitor  $\synthES{\!e_1\cap e_2\!}{}^{\PSet}$ is given by the \emph{cross product} between the edit automata $\synthES{e_1}{}^{\PSet}$ and $\synthES{e_2}{}^{\PSet}$, to accept only traces that satisfy both $e_1$ and $e_2$;  the definition of the cross product 
between two edit automata 
recalls that for finite state automata,  and it is reported in the appendix in Table~\ref{tab:edit-automata-product}. 
%
The monitor ${\synthESP{\! p_1\cap p_2 \!}{\mathsf X}{\PSet}}$ 
        is given by the cross product 
        between the edit automata ${\synthESP{p_1}{\mathsf X}{\PSet}}$ and ${\synthESP{p_2}{\mathsf X}{\PSet}}$.
        {\color{black}
	The monitor ${\synthESP{p_1;p_2}{\mathsf X}{\PSet}}$  is given by  
        the  automaton ${\synthESP{p_1}{\mathsf Z}{\PSet}}$, 
        where $\mathsf Z  = {\synthESP{p_2}{\mathsf X}{\PSet}} $;  basically $\mathsf Z$ ties the final states of the automaton enforcing $p_1$ with the initial state of the automaton enforcing $p_2$ 
        (\emph{e.g.}, ${\synthESP{\epsilon;p_2}{\mathsf X}{\PSet}} = {\synthESP{\epsilon}{\mathsf Z}{\PSet}} = \mathsf Z$, for  $ {\mathsf Z} = {\synthESP{p_2}{\mathsf X}{\PSet}}$ ).  
	The monitor associated to a union property $\cup_{i\in I}\pi_i.p_i$ does the following: (i) 
	\emph{allows} all actions associated to  the events $\pi_i$, (ii) \emph{inserts} an action associated to  some admissible event $\pi_i$ only when the controller wishes to prematurely complete the scan cycle, i.e., it emits an $\fineC$-action, and (iii) \emph{suppresses} any  other action except for $\tick$- and $\fineC$-actions. }

        {\color{black}
          Thus, the \emph{mitigation} of the enforcement is actually implemented in the monitors synthesised from union properties. 
 In practise,  when  the controller under scrutiny complies with the property enforced by the monitor, 
 the two components, monitor and controller, evolve in a tethered fashion (by applying rule \rulename{Allow}), moving through related correct states.
 However, if the controller gets somehow corrupted (for instance, due to the presence of a malware)  then the two components will get misaligned  reaching unrelated states. 
 In this case, the enforcer mitigates the attack  by suppressing the remaining actions emitted by the controller   (by applying rule \rulename{Suppress}) until the controller reaches the end of the scan cycle, signalled by the emission of the   $\fineC$-action\footnote{As said in Section\ref{sec:attacker-model}, a  malware that aims to take control of the plant  has no interest in delaying the scan cycle and risking the violation of the maximum cycle limit  whose consequence would be the immediate shutdown of the controller~\cite{BLACKHAT2016}. }.
 After that, 
  if monitor and controller are not aligned the monitor will command the insertion of a safe trace,  without any involvement of the  controller, via one or more applications of the rule \rulename{Insert}.
 Safe traces inserted in full autonomy by our enforcers always terminate with an $\fineC$. 
 Thus, when both the controller and the monitor will be aligned, at the end of the scan cycle,  they will synchronise on the action $\fineC$, via an application of the  rule \rulename{Allow}, and from then on they may continue in a tethered fashion. 
	\begin{table}[t] 
		\begin{displaymath}
		{\small
			\begin{array}{rcl}
                          	\synthES{p^{\ast}}{ }^{\PSet} &\defn& 
			\mathsf X, \text{ for } \mathsf X = \synthES{p}{\mathsf X}^{\PSet} 
                        \\[2pt]
			 \synthES{e_1 \cap e_2}{ }^{\PSet} &\defn&  \TimesP{ \synthES{e_1}{ }^\PSet}{\synthES{e_2}{}^{\PSet}}{\Xrec},\q \Xrec \text{ fresh}

			 \\[1pt] 
		
			\synthES{\epsilon}{\mathsf X}^{\PSet} &\defn& \mathsf X \\[1pt]
			\synthES{p_1 \cap p_2}{\Xrec}^{\PSet} & \defn & \TimesP{\synthES{p_1}{\Xrec}^{\PSet}}{\synthES{p_2}{\Xrec}^{\PSet}}{\Xrec}  \\[1pt]
			\synthES{p_1;p_2}{\mathsf X}^{\PSet} &\defn& 
			\synthES{p_1}{\Zrec}^{\PSet},\textrm{ for } \Zrec = \synthES{p_2}{\mathsf X}^{\PSet}, \q \Zrec \text{ fresh} \\[1pt]
			
			%
			
			\synthES{\bigcup_{i\in I}\pi_i.p_i}{\mathsf X}^{\PSet} &\defn&  \Zrec,  \text{  for }  \\[1pt]
			&  
                        \multicolumn{2}{l}{
                        {\color{black}
			\Zrec =
                        \begin{cases}
                        \sum\limits_{\substack{{\scriptscriptstyle i\in I} }}
			\pi_i.\synthES{p_i}{\mathsf X}^{\PSet}
		
			+ \sum\limits_{\substack{{\scriptscriptstyle i\in I} }}
			{\scriptstyle \eact{\,\fineC}{\pi_i}}.\synthES{p_i}{\mathsf X}^{\PSet}
		
			+ \sum\limits_{\subalign{{\scriptscriptstyle \alpha \in {\mathcal Q} }}} {^{-}}\alpha.\Zrec, \text{  if } \fineC \not \in \cup_{i \in I}\pi_i
                        \\[4pt]
                         \sum\limits_{ \substack{{\scriptscriptstyle i\in I} }}
			\pi_i.\synthES{p_i}{\mathsf X}^{\PSet}
		
			+ \sum\limits_{\subalign{{\scriptscriptstyle \alpha \in {\mathcal Q} }}} {^{-}}\alpha.\Zrec, \text{  otherwise} 
                        
                        \end{cases}
                        }
                        }\\[4pt]

&& {\color{black} 
\q \textrm{ where }{\scriptstyle {\mathcal Q} \, = \,  {\PSet} \setminus  (\cup_{i \in I}\pi_i\cup\{\tick,\fineC\})}
}  
        
			\end{array}
		}
		\end{displaymath}
		
		\caption{Monitor synthesis  from properties in $\PropG$ and  $\PropL$.}
		\label{tab:synthesis-logic}
	\end{table}

 \begin{remark}
   \label{rem:mitigation-time}
   {\color{black}
     Note that even whe the controller is completely unreliable and the monitor inserts an entire safe trace, 	the assumption made in Remark~\ref{rem:maximum-time} ensures us that the enforced scan cycle always ends well before a violation of the maximum  cycle limit.
     }
 \end{remark}

 }

%
%

	Now, we calculate the complexity of the synthesis algorithm based on  the number of occurrences of the operator $\cap$ in $e$ and the dimension of $e$, $\ctrlDim{e}$, \emph{i.e.},  the number  of all operators occurring in $e$.  
     {\color{black}   Intuitively, the size of a property is given by the number of operators occurring in it.
	\begin{definition}
          Let $\propDim {}:\PropG\cup\PropL\rightarrow\mathbbm{N}$ be a property-size function  defined as: 
			\begin{displaymath}
				{\small
			\begin{array}{lcl@{\hspace*{0.5cm}}lcl}
						\propDim{p^*} &\defn& \propDim{ p } +1
						&
				\propDim{e_1 \cap e_2} &\defn&  \propDim{e_1}	+\propDim{e_2}	+1\\[1pt]		
			\propDim{\epsilon} &\defn& 1 
			&
			\propDim{p_1;p_2} &\defn& \propDim{p_1} + \propDim{p_2}+1\\[1pt]
				\propDim{p_1\cap p_2} &\defn& \propDim{p_1} + \propDim{p_2}+1
			&
			\propDim{\bigcup_{i\in I}\alpha_i.p_i} &\defn& |I| + \sum_{i \in I}\propDim{ p_i }. 
			\end{array}
		}
			\end{displaymath}	
	\end{definition}
        }

	\begin{proposition}[Complexity]
		\label{prop:poly2}
		Let $e \in \PropG$ be a global property and  $\PSet$ be a set of actions such that $\events e \subseteq \PSet$.  The complexity of the  algorithm to synthesise $\funEditP{e}{\PSet}$ is $\mathcal{O}( {|}\PSet{|} \cdot m^{k+1} )$, with $m=\ctrlDim{e}$ and $k$ being the number of occurrences of the operator $\cap$ in $e$.
\end{proposition}

	In the following, we prove that the enforcement induced by our synthesised monitors enjoys the properties stated in the Introduction: \emph{determinism preservation},  \emph{transparency}, \emph{soundness}, \emph{deadlock-freedom}, \emph{divergence-freedom}, and \emph{scalability}. 
        %
        In this section, with a small abuse of notation, given a set of observable actions $\PSet$, we will use $\PSet$ to denote also the set of the corresponding events.



        Given a deterministic global property $e$, our synthesis algorithm returns a   deterministic enforcer (according to Definition~\ref{def:semantic-deterministic-enfrocement}), {\color{black} \emph{i.e.}, an enforcer that can be effectively implemented. } Formally, 
        {\color{black}
     	\begin{proposition}[Deterministic preservation]
	\label{prop:sematic-determinisct-enforcement}
	Given a deterministic global property $e \in \PropG$ over a set of events $\PSet$. The edit automaton
	$\funEditP{e}{\PSet}$ is  deterministic. 
        \end{proposition}
}

	Let us move to   \emph{transparency}. 
	Intuitively, the enforcement induced by a {\color{black} deterministic} property $e \in \PropG$  {\color{black} should preserve} any execution trace satisfying  $e$ itself  (Definition~2 at pag.\ 5 of~\cite{Ligatti2005}).

{\color{black}           
	\begin{theorem}[Transparency]
		\label{prop:transparency-logic} 
		Let $e \in \PropG$ be a deterministic global property, $\PSet$ be a set of observable actions such that $\events e \subseteq \PSet$, and $P\in \mathbbm{Ctrl}$ be a controller. 
		Let $t= \alpha_1 \cdots \alpha_n $ be a 
                trace of the controller $P$ with $t\in \regSemantics{e}$.  Then, (1) 
               $t$ is a trace of the edit automaton $\funEditP{e}{\PSet}\!$,  and 
(2) there is no trace $t'=\alpha_1 \cdots \alpha_k\cdot \lambda $  for  $\funEditP{e}{\PSet}\!$ such that  $0 \leq k  <n$ and  
$\lambda \in \{  \suppressE{\alpha_{k+1}}, \insertE{\alpha}{\alpha_{k+1}}\}$, for  
some $\alpha$. 
         
	\end{theorem}
        Basically, conclusion (1) says that all execution trace $t$ (of a controller $P$) satisfying the enforcing property $e$ are \emph{allowed} by the associated enforcer  $\funEditP{e}{\PSet}\!$, while conclusion (2) says that allowing the trace $t$ is the only possible option in the enforcement (this follows by the determinism of $e$). 
         } 
         
	Another important property of our enforcement is  \emph{soundness}~\cite{Ligatti2005}. Intuitively,  a controller under the scrutiny of  a monitor $\funEditP{e}{\PSet}$ should only yield execution traces which satisfy the enforced property $e$, \emph{i.e.}, execution traces which are consistent with its semantics $\regSemantics{e}$ (up to $\tau$-actions).
	\begin{theorem}[Soundness]
	\label{thm:safety-logic}
	Let $e \in \PropG$ be a global property,  $\PSet$ be a set of observable actions such that $\events e \subseteq \PSet$, and $P\in \mathbbm{Ctrl}$ be a controller.
	If $t$ is a  trace of the monitored controller $\confCPS{\funEditP{e}{\PSet}}{P}$ then  $\hat{t}$ is a prefix of some trace in $\regSemantics{e}$ (see Notation~\ref{not:trace} for the definition of the trace $\hat t$). 		
 \end{theorem}

        Here, it is important to stress that in general soundness does not ensure deadlock-freedom of the monitored controller. That is,  it may be possible that the enforcement of some property $e$ causes a deadlock of the controller $P$ under scrutiny. In particular, this may happen in our controllers only when the initial sleeping phase 
        does not match  the enforcing property ({\color{black}\emph{e.g.}, $P=\tick.c.\fineC.P$ and $e=(c.\fineC)^\ast$}).
        Intuitively, a local property will be called a $k$-sleeping property  if it allows $k$ initial time instants of sleep. Formally, 
        \begin{definition}  For $k \in  \mathbb{N}^+ $,
          we say that $p \in \PropL$ is a \emph{$k$-sleeping local property},  only if $\regSemantics{p} =\{t | t=t_1\cdot... \cdot  t_n, \text{ for } n>0, \text{ s.t. } t_i = \tick^k{\cdot } t_i'{\cdot} \fineC, \fineC\notin t_i', \text{ and } 1\leq i \leq n \}$. We say that  $p^\ast$  is a \emph{$k$-sleeping global property} only if $p$ is, and $e=e_1\cap e_2$ is $k$-sleeping only if both $e_1,e_2$ are $k$-sleeping.
	\end{definition}
	
        The enforcement of $k$-sleeping properties does not introduce deadlocks in  $k$-sleeping controllers.
         This is because our synthesised monitors suppress all incorrect actions of the controller under scrutiny, driving it to the end of its scan cycle. Then,  the controller remains in stand-by while the monitor yields a safe sequence of actions to mimic a safe completion of the current scan cycle. 
	\begin{theorem}[Deadlock-freedom]
		\label{thm:deadlock}
		Let $e\in \PropG$ be a $k$-sleeping global property, and $\PSet$ be a set of observable actions such that $\events e \subseteq \PSet$. 
                Let $P\in \mathbbm{Ctrl}$ be a controller of the form $P=\tick^{k}.S$  
                whose set of observable actions is contained in $\PSet$. Then,  $ \confCPS {\funEditP{e}{\PSet}} {P}$ does not deadlock. 
	\end{theorem}

		Another important property of our enforcement mechanism
      	is \emph{divergence-freedom}. In practice, the enforcement  does not introduce divergence: monitored controllers  will always be able to complete their scan cycles by executing a finite number of actions. 
      	This is because we limit our enforcement to well-formed properties (Definition~\ref{def:well-formedness}) which always terminates with an $\fineC$ event. In particular,  the  well-formedness of local properties  ensures us 
        that in a global property of the form $p^\ast$  the number of events  within two subsequent $\fineC$ events  is always finite.

	\begin{theorem}[Divergence-freedom]
			\label{thm:divergence}	
			Let $e\in \PropG$ be a global property,  $\PSet$ be a set of observable actions such that $\events e \subseteq \PSet$, and $P\in \mathbbm{Ctrl}$ be a controller. Then, there exists a  $k \in  \mathbb{N}^+$ such that whenever \smash{$\confCPS{\funEditP{e}{\PSet}}{P}\trans{t}\confCPS{\Edit}{J}$}, 
                        if \smash{$\confCPS{\Edit}{J}\trans{t'}\confCPS{\Edit'}{J'}$}, 
                        with $|t'| \ge k$, then $\fineC\in t'$.	
		\end{theorem}

          Notice that all properties  seen up to now scale to \emph{field communications networks} of controllers. This means that they are  preserved when the  controller under scrutiny is running in parallel with other controllers in the same field communications network. As an example,  by an application of Theorems~\ref{prop:transparency-logic}  and \ref{thm:safety-logic},  we show how both transparency and soundness scale to field networks. A similar result applies to the remaining properties. 

\begin{corollary}[Scalability to networks of PLCs]
  \label{cor:transparency-logic}
  	\label{cor:enforcement-logic} 
		Let $e \in \PropG$ be a global property  and  $\PSet$ be a set of observable actions, such that $\events e \subseteq \PSet$. Let $P\in \mathbbm{Ctrl}$ be a controller and $N \in \mathbbm{FNet}$ be a field network. 
If
{\small $(\confCPS{\funEditP{e}{\PSet}}{P}) \parallel N \trans{t} (\confCPS{\Edit}{J}) \parallel N'$}, for some  $t$, $\Edit$, $J$ and $N'$, then
\begin{itemize}

\item
  {\color{black}
    whenever {\small $P \trans{t'} J$}, 
  with 		$t' =  \alpha_1 \cdots \alpha_n  \in  \regSemantics{e}$, the trace  $t'$ is a trace of ${\funEditP{e}{\PSet}}$ and
  there is no trace $t''=\alpha_1 \cdots \alpha_k\cdot \lambda $  of  $\funEditP{e}{\PSet}\!$ such that  $0 \leq k  <n$ and  
$\lambda \in \{  \suppressE{\alpha_{k+1}}, \insertE{\alpha}{\alpha_{k+1}}\}$, for  
some $\alpha$; 
   }
  \item whenever {\small $\confCPS{\funEditP{e}\PSet}{P} \trans{t'} \confCPS{\Edit}{J}$}
		the trace $\widehat{t'}$ is a prefix of some trace in $\regSemantics{e}$. 
\end{itemize}
	\end{corollary}

\section{Our enforcement mechanism at work}
\label{sec:implementation}
	\begin{figure}[t]
	\centering

		\includegraphics[
		height=4cm,keepaspectratio=true,angle=0]{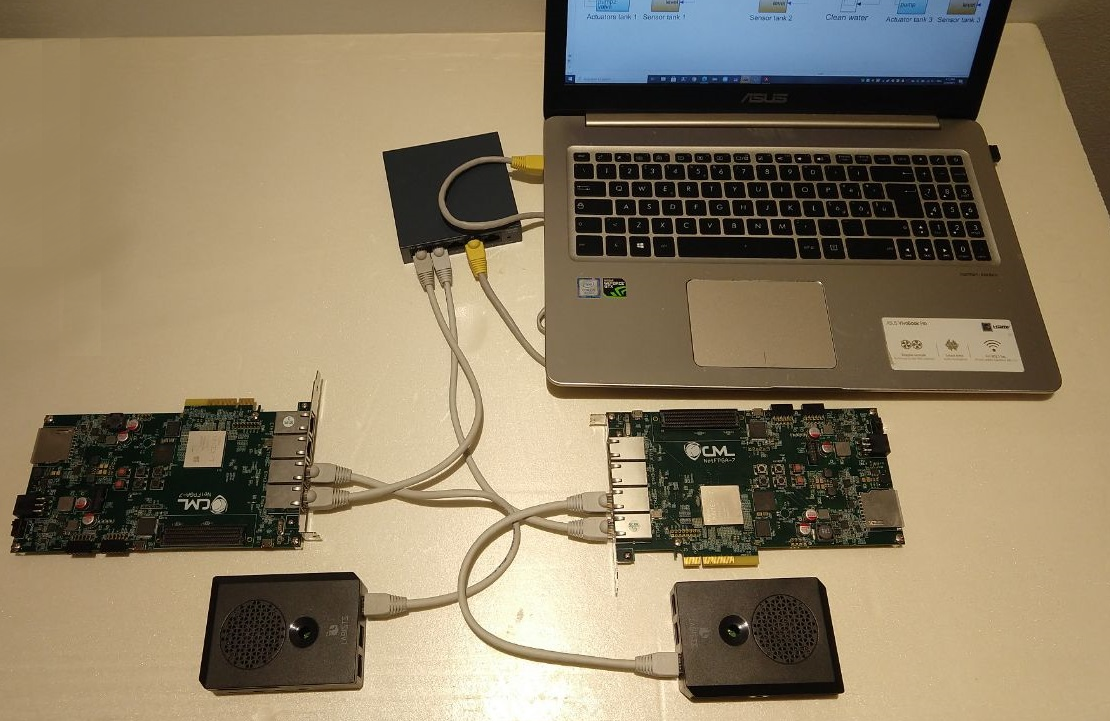}
	\caption{Some physical components of our implementation.}
	\label{fig:setup}
\end{figure}
In this section, we propose an implementation of our enforcement mechanism in which  monitors, running on \emph{field-programmable gate arrays} (FPGAs)~\cite{FPGA}, enforce \emph{open source PLCs}~\cite{OpenPLC}, running on Raspberry Pi devices~\cite{RaspberryPI}, and governing a physical plant simulated in  \emph{Simulink}~\cite{MATLAB}.  The section has the following structure. In Section~\ref{sec:FPGA},  we argue why FPGAs are good candidates for implementing \emph{secure proxies}. In Section~\ref{sec:Swat-impl}, we describe  how we implemented the {\color{black} whole enforcement architecture for the use case of Section}~\ref{sec:case-study}. In Section~\ref{sec:attacks}, we test our implementation  injecting the enforced PLCs with five different malware aiming at causing three different physical perturbations: tank overflow, valve damage, and pump damage.  {\color{black}
  The attacks have been chosen to cover as much as possible the attacker model of Section~\ref{sec:attacker-model}. In particular, they include: a drop of the actuator commands of the valve,  an integrity attack on the water-level sensors, a forgery of the actuator commands of the valve,  a forgery of the message requests to open/close  the valve, and a forgery of the  actuator commands of the pumps. {\color{black}Section~\ref{sec:discussion} discusses the performance of our implementation. }
}


\subsection{FPGAs as secure proxies for ICSs}
\label{sec:FPGA}

	Field-programmable gate arrays (FPGAs) are semiconductor devices that can be programmed to run specific applications. 
	An FPGA
        consists of  (configurational) logic blocks, routing channels and I/O blocks. The logic blocks can be configured to perform complex combinational functions and are further made up of transistor pairs,
	logic gates, lookup tables and multiplexers.
	The applications are written using hardware description languages, such as Verilog~\cite{Verilog}. Thus, in order  to execute an application on the FPGA, its Verilog code is converted into a sequence of bits, called  \emph{bitstream}, that 
	is loaded into the FPGA.


	 FPGA are assumed to be secure when the adversary does not have physical access to the device, \emph{i.e.}, the bitstream cannot be compromised~\cite{FPGA_security}. Recent  FPGAs support remote updates of the bitstream by relying on authentication mechanisms to prevent unauthorised uploads of malicious logic~\cite{FPGA_security}. Nevertheless, as said in the Introduction and  advocated by McLaughlin and Mohan~\cite{McLaughlin-ACSAC2013,Mohan-HiCONS2013}, any form of runtime reconfiguration  should be prevented. Summarising, under the assumption that the adversary does not have   physical access to the FPGA and she cannot do remote updates, FPGAs represent a good candidate for the implementation of secure enforcing proxies. 
	
		\begin{figure}[t]
	\centering
	\includegraphics[trim=0 210 0 200,clip,width=11cm,keepaspectratio=true,angle=0]{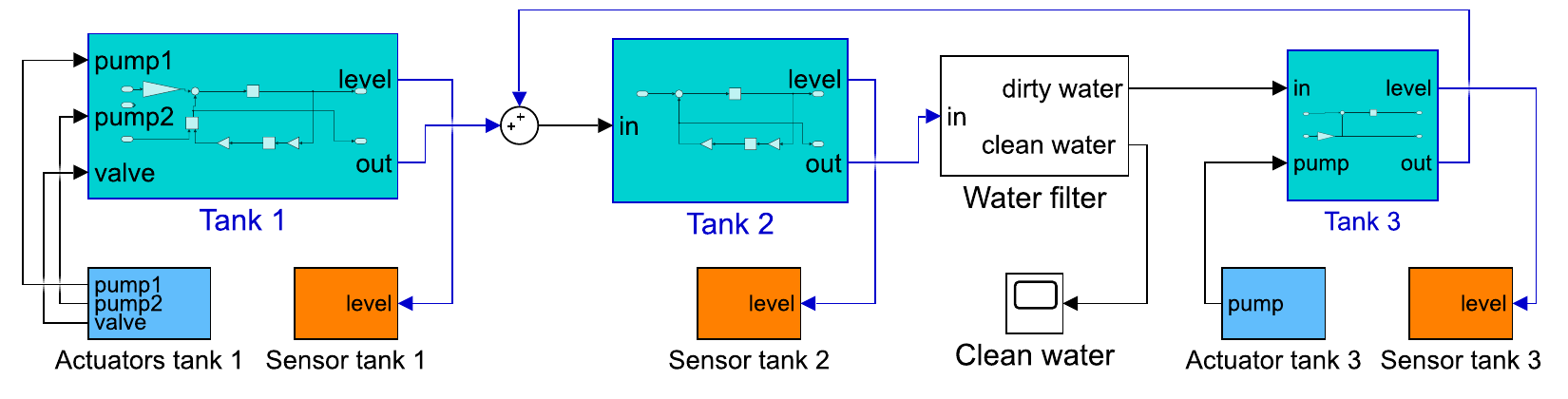}
	\caption{An implementation in Simulink of the plant of the SWaT system.}
	\label{fig:case-study-simulink}
\end{figure}
	\subsection{An implementation of the enforcement of the SWaT system of Section~\ref{sec:case-study}}
        \label{sec:Swat-impl}
        The proposed implementation adopts different approaches for plant, controllers and enforcers.

	\paragraph{Plant}
        The plant of the SWaT system is simulated in Simulink}~\cite{MATLAB},   a 
	framework 
	to model, simulate and analyse cyber-physical systems, widely adopted in industry and research.
	A Simulink model is given by \emph{blocks} interconnected via \emph{wires}.
	Our Simulink model contains blocks to simulate water tanks, actuators (\emph{i.e.}, pumps and valves) and sensors (see Figure~\ref{fig:case-study-simulink}). 
	In particular, water-tank blocks implement the differential equations that model the dynamics of the tanks according to the physical constraints obtained from~\cite{SWaT, swat-dataset}. 
	Actuation blocks receive commands from  PLCs, whereas sensor blocks send measurements to  PLCs.
	For simplicity, state changes of both pumps and valves do not occur instantaneously; they take 1 second.
	We ran our Simulink model  on a laptop with 2.8 GHz Intel i7 7700 HQ,  16GB memory, and Linux Ubuntu 20.04 LTS OS.
	
		\paragraph{Controllers}
	        Controllers are defined in \emph{OpenPLC}~\cite{OpenPLC},  an open source PLC capable of running user programs in all five IEC61131-3 defined languages~\cite{61131-3}.
	Additionally, OpenPLC supports standard SCADA protocols, such as Modbus/TCP, DNP3 and  Ethernet/IP.
	OpenPLC can  run on a variety of hardware, from a simple Raspberry Pi to robust industrial boards.
 	We installed OpenPLC on three Raspberry Pi 4~\cite{RPi4Modelb}; each instance  runs one of the three ladder logics seen in 
        Figure~\ref{fig:ladder-logic-plcs}.

		\paragraph{Enforcers}
	Enforcers are implemented using three NetFPGA-CML development boards~\cite{NetFPGA}.
	Our synthesis algorithm  is implemented in Python to return  enforcers written in Verilog, and  checked for correctness using ModelSim. The Verilog code is then compiled into a bitstream and  executed in the FPGA. 
	{\color{black}
	  More precisely, our algorithm in Python   takes as input a JSON file containing the property to be synthesised and other relevant informations,  such as the number of input/output signals and a fixed  priority among admissible safe output signals. 
          Then,  the property is parsed by means of the  ANTLR parser~\cite{ANTLR}. 
          After the parsing, our algorithm implements the synthesis  of Table~\ref{tab:synthesis-logic} to derive the enforcing edit automaton; this is written down into a JSON file. 
          At this stage, the derived edit automaton is still somewhat abstract, as both $\fineC$- and $\tick$-actions are explicitly represented. 
          Finally, the algorithm compiles the edit automaton into an enforcer written in Verilog, where the above abstractions are implemented. In particular, the passage of time  (\emph{i.e.},  $\tick$-actions) is  represented and monitored via \emph{clock variables},  
while the end of scan cycles (\emph{i.e.},  $\fineC$-actions) is implemented via specific code  to synchronise enforcers and controllers, relying on clock variables.
Thus, before  each scan cycle the enforcer forwards the current inputs (coming from the plant) to the controller. Then, when the scan cycle is completed, it receives from the controller all the current outputs, and forwards them to the actuators. {
In the meanwhile, the enforcer monitors the passage of time via its clock variables, and when  the scan cycle is completed (\emph{i.e.}, the controller sends all outputs) it moves to the state corresponding to the following scan cycle. 
   }         
	%
	%
                        Finally,  in our FPGAs we also write some  code to implement an UDP-based network connecting together  enforcers, PLCs, and the simulated plant. 
		
	}
%


\smallskip
        The code of the three PLCs, the  algorithm in Python,  the enforcers  written in Verilog,  and the Simulink simulations can be found at: {\footnotesize \url{https://bitbucket.org/formal\_projects/runtime\_enforcement
			
	}}.

	\subsection{The enforced SWaT system under attack}
\label{sec:attacks}
	In this section, we consider five different attacks targeting the PLCs of the SWaT system to achieve three possible malicious goals: (i) overflow the water tanks, (ii) damage of the valve, (iii) damage of the pumps.
	 In order to simulate the injection of  malware in the PLCs, 
	we reinstall the original PLC ladder logics with compromised ones, containing some  additional logic  intended to disrupt the normal operations of the PLC~\cite{Ladder-logic-bombs}.  
	In the following, we will discuss these  attacks, grouped by goals, showing how the enforcement of specific properties  mitigates  the attacks by preserving the correct behaviour of the monitored PLCs.

\paragraph*{Tank overflow}
	Our \emph{first attack} is a DoS attack targeting $\mathrm{PLC}_{1}$ by  dropping  the commands  to close the valve. In the left-hand side of Figure~\ref{fig:ld-plc1-attack1} we show a possible implementation of this attack in ladder logic. Basically, the malware remains silent  for $500$ seconds and then it sets  true a malicious $\mathit{drop}$ variable (highlighted in  yellow). Once the variable $\mathit{drop}$ becomes true, the $\mathit{valve}$ variable is forced to be false (highlighted in  red), thus preventing the closure of the valve.
    	\begin{figure}[t]
	\centering
	\begin{minipage}{0.39\textwidth}
		\includegraphics[
		width=5.6cm,keepaspectratio=true,angle=0]{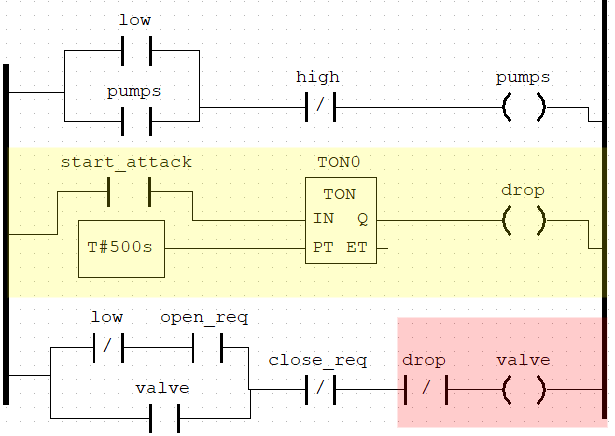}
	\end{minipage}\Q
	\begin{minipage}{0.4\textwidth}
		\includegraphics[
		width=4.8cm,keepaspectratio=true,angle=0]{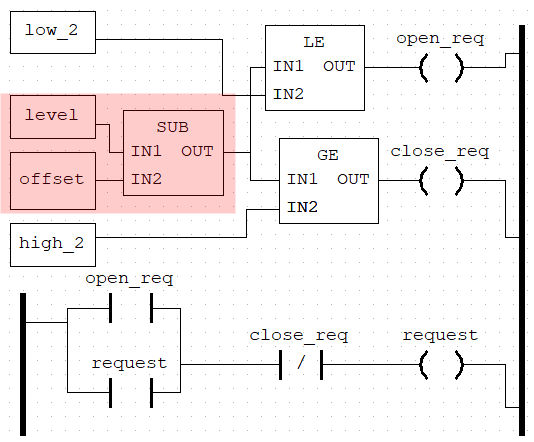}
	\end{minipage}
	\caption{Tank overflow: Ladder Logic of the first (left) and the second attack (right).}
	\label{fig:ld-plc1-attack1}
\end{figure}

        In order to prevent attacks aiming at overflowing the tanks,   we propose the following three enforcing properties, one for each PLC, respectively:
	
\begin{itemize}	
		\item $e_1\defn(\CBP{h_1}{\overline{\off{1}}}{1}{m})^\ast \cap\, (\CBP{h_1}{\overline{\off{2}}}{1}{m})^\ast$, an intersection between two conditional bounded persistency properties to enforce $\mathrm{PLC_{1}}$  to prevent water overflow in $T_1$. This  property ensures that both pumps $\mathit{pump}_1$ and $\mathit{pump}_2$ are off, for $m$ consecutive scan cycles, when the level of  $T_1$ is high (measurement $h_1$). Here,   $m<n$ for $n\in \mathbb{N}$ is the number of scan cycles required to empty  $T_1$ when its level  is high, both pumps are off, and the valve is open.
		
		
		\item $e_2\defn (\CBP{h_2}{\overline{\ask{close}}}{1}{u})^\ast$, a 
		conditional bounded persistency property for  $\mathrm{PLC_{2}}$ 
		ensuring that requests to close the valve (event $\overline{\ask{close}}$) are sent for $u$ consecutive scan cycles
		when the level of water in tank $T_2$ is high (measurement $h_2$). Here,   $u<v$ for  $v \in \mathbb{N}$ is the number of scan cycles required to empty the tank $T_2$ when the level is high and the valve is closed. 
	      \item $e_3\defn(\CBP{h_3}{\overline{\on{3}}}{1}{w})^\ast$, a conditional bounded persistency property for $\mathrm{PLC_{3}}$ to ensure that $\mathit{pump}_3$ is on for $w$ consecutive scan cycles when the level of water in tank $T_3$ is high (measurement $h_3$). Here,  $w<z$ for $z\in \mathbb{N}$ is the time 
		(expressed in scan cycles) required to empty the tank $T_3$ when the level  is high and  $\mathit{pump}_3$ is on.
	\end{itemize}

	Now, let us analyse the effectiveness of the  enforcement induced by these three properties.
        For instance, in the upper graphs of Figure~\ref{figure:goal-1-attack-1} we report the impact  on the tanks $T_1$ and $T_2$ of the  DoS attack previously described,  when enforcing the three properties $e_1, e_2$ and $e_3$ in the corresponding PLCs. Here,  the red region denotes when the attack becomes active. As the reader may notice, despite repeated requests to close the valve coming from  $\mathrm{PLC}_2$, the compromised $\mathrm{PLC}_1$ never closes the valve causing  the overflow of tank $T_2$. So, the enforced property $e_1$ is not up the task. 

		\begin{figure}[t]
		\centering
		\begin{minipage}{0.44\textwidth}
			\includegraphics[
			width=6cm,keepaspectratio=true,angle=0]{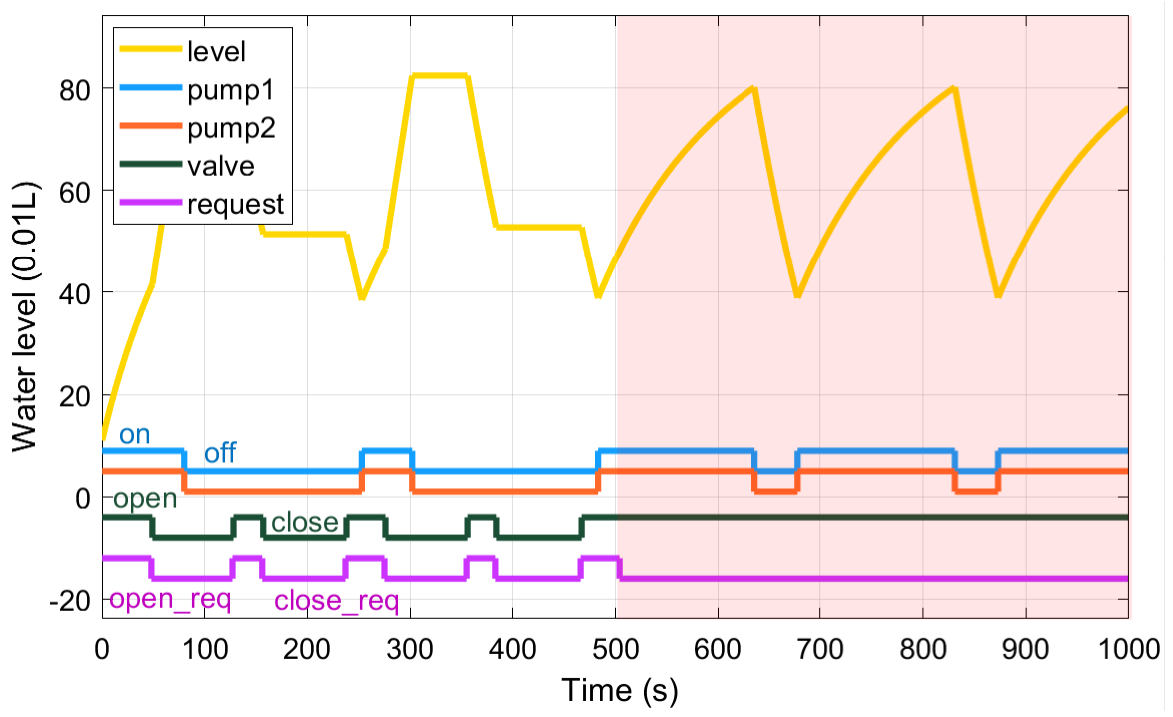}
		\end{minipage}\q
		\begin{minipage}{0.44\textwidth}
			\includegraphics[
			width=6cm,keepaspectratio=true,angle=0]{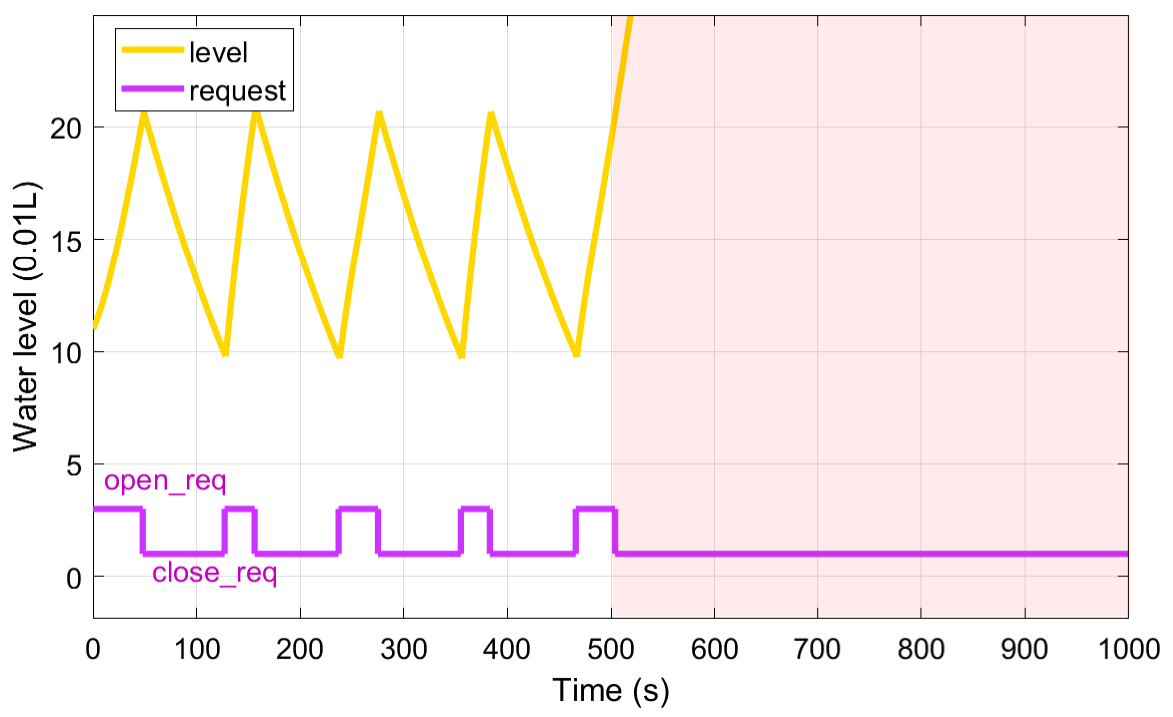}
		\end{minipage}\Q
                \\[2pt]
		\begin{minipage}{0.44\textwidth}
			\includegraphics[
			width=6cm,keepaspectratio=true,angle=0]{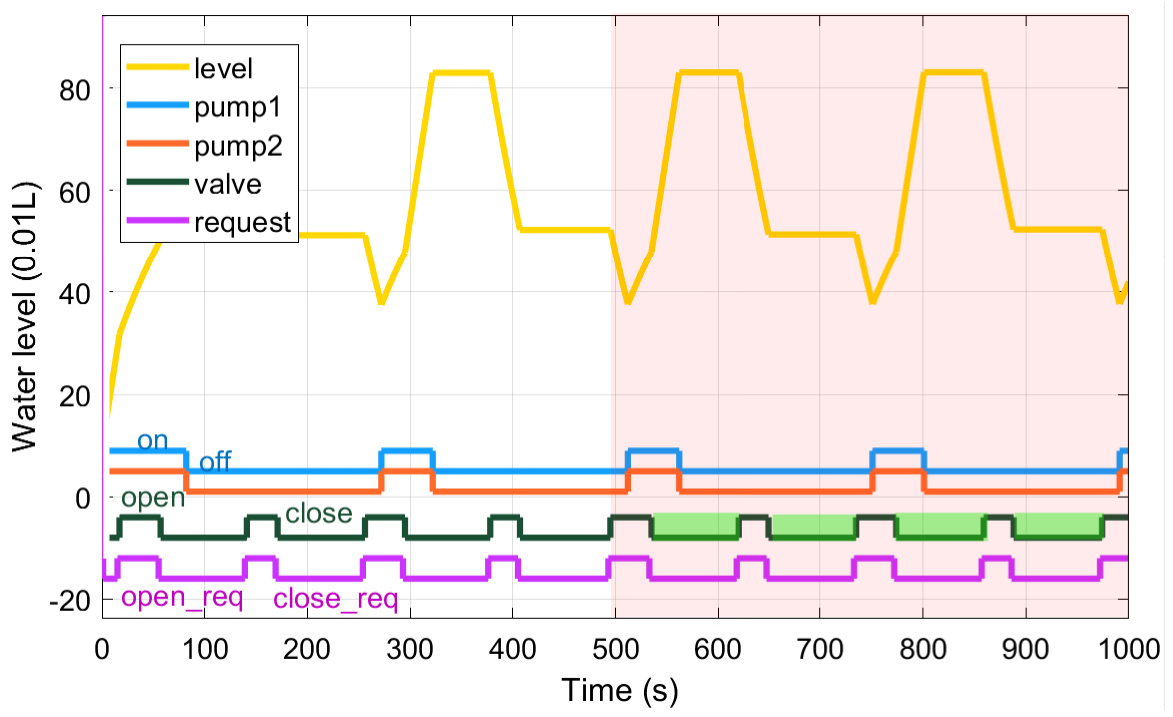}
		\end{minipage}\q
		\begin{minipage}{0.44\textwidth}
			\includegraphics[
			width=6cm,keepaspectratio=true,angle=0]{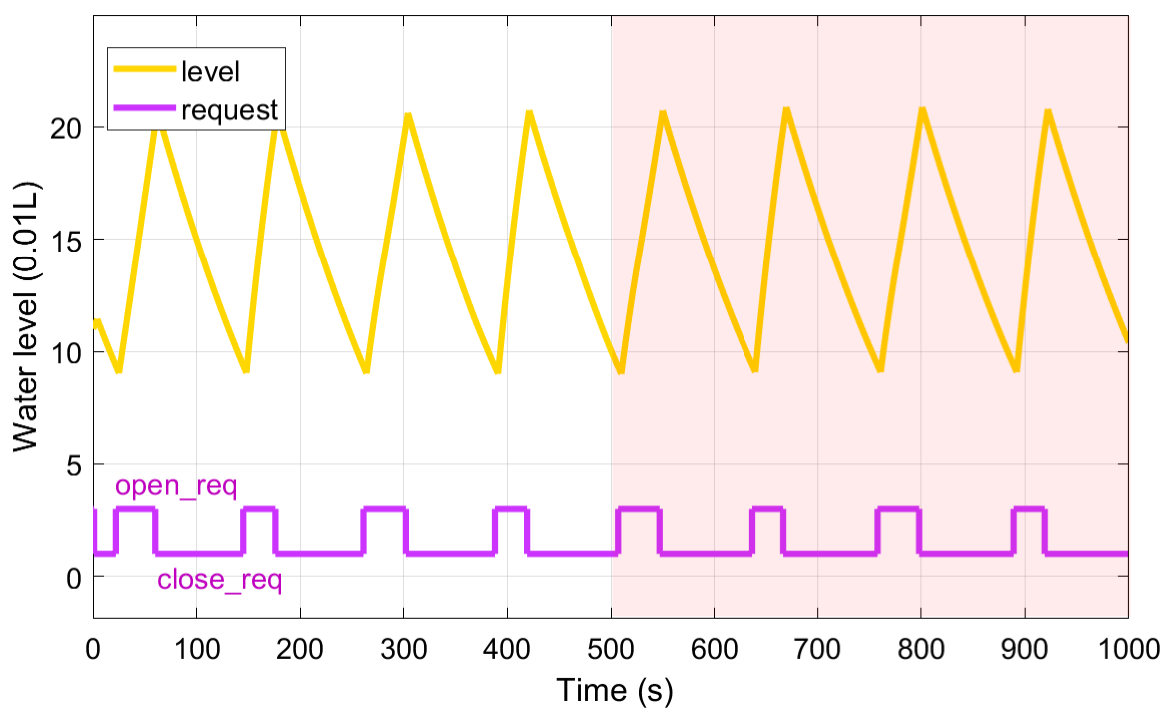}
		\end{minipage} 
		\caption{Tank overflow: DoS attack on $\mathrm{PLC}_1$ when enforcing $e_1, e_2,e_3$ (up) and $e_1',e_2, e_3$ (down).}
		\label{figure:goal-1-attack-1}
	\end{figure}

	In order to prevent this attack, we must guarantee that $\mathrm{PLC}_1$ closes the valve when  $\mathrm{PLC}_2$ requests so. Thus, we should enforce in $\mathrm{PLC}_1$ a more demanding    property $e_1'$ defined as follows: $
        e_1 \, \cap \,   \CBE{\ask{close}}{\overline{\close}}{1}{1}$. Basically, the last part of the property ensures that every request to close the valve
        is followed by an actual closure of the valve 
        in the same scan cycle. 
	The impact of the malware on $\mathrm{PLC}_{1}$  when enforcing the properties $e_1',e_2,e_3$ is represented in the lower graphs of Figure~\ref{figure:goal-1-attack-1}. Now, the correct behaviour of $\mathrm{PLC}_1$ is ensured, thus  preventing the overflowing of the water tank $T_2$. In these graphs, the green highlighted regions denote when the monitor \emph{detects} the attack and \emph{mitigates} the activities of the compromised $\mathrm{PLC}_1$. In particular, the monitor \emph{inserts} the commands to close the valve on behalf of $\mathrm{PLC}_1$ when $\mathrm{PLC}_2$ sends requests to close the valve.


	Having strengthened the enforcing property for $\mathrm{PLC}_1$ one may think that the enforcement of $e_2$ in $\mathrm{PLC}_2$ is now superfluous to prevent water overflow in $T_2$. 
	However, this is not the case if the attacker can compromise 
        $\mathrm{PLC}_{2}$. Consider a \emph {second attack} to  $\mathrm{PLC}_2$, an  \emph{integrity attack}  that adds an offset of $-30$ to the measured water level of $T_2$. 
        We show a ladder logic implementation of such attack in the right-hand side of Figure~\ref{fig:ld-plc1-attack1} where, for simplicity, we omit the initial silent phases lasting $500$ seconds.
	The impact on the tanks $T_1$ and $T_2$ of the malware injected in $\mathrm{PLC}_{2}$ in the presence of the enforcing of the properties $e'_1$ and $e_3$, respectively,  is represented on the upper graphs of Figure~\ref{figure:goal-1-attack-2}.
\begin{figure}[t]
	\centering
	\begin{minipage}{0.44\textwidth}
		\includegraphics[
		width=6cm,keepaspectratio=true,angle=0]{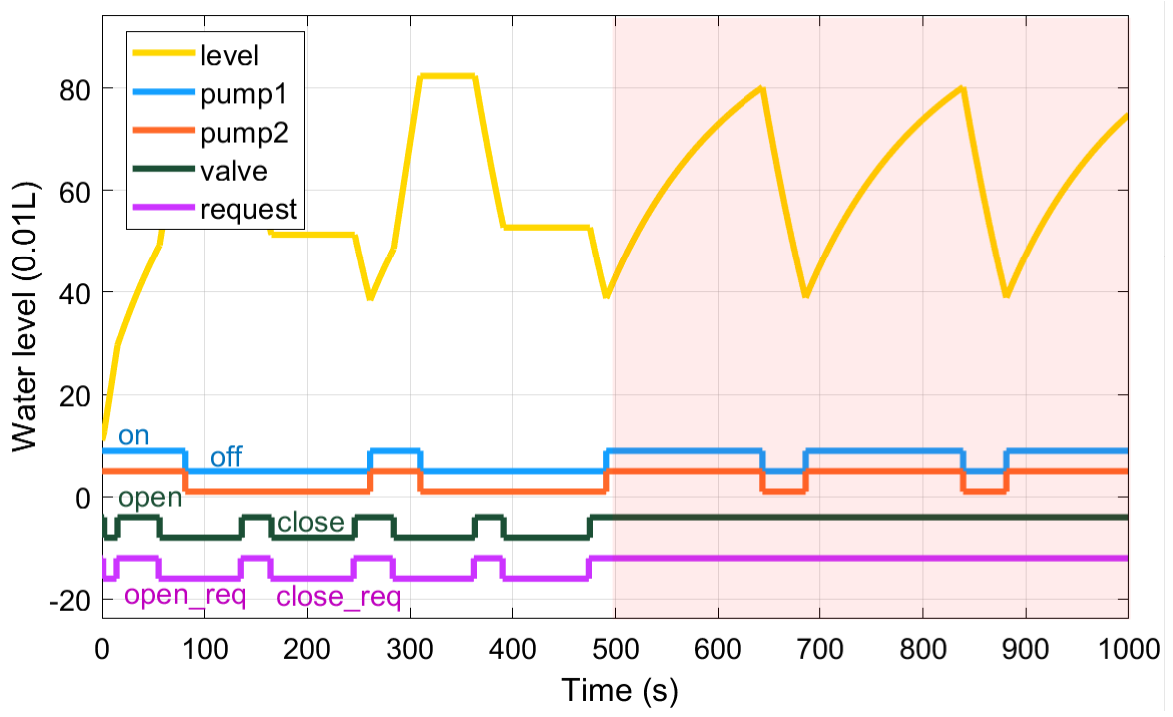}
	\end{minipage}\q
	\begin{minipage}{0.44\textwidth}
		\includegraphics[
		width=6cm,keepaspectratio=true,angle=0]{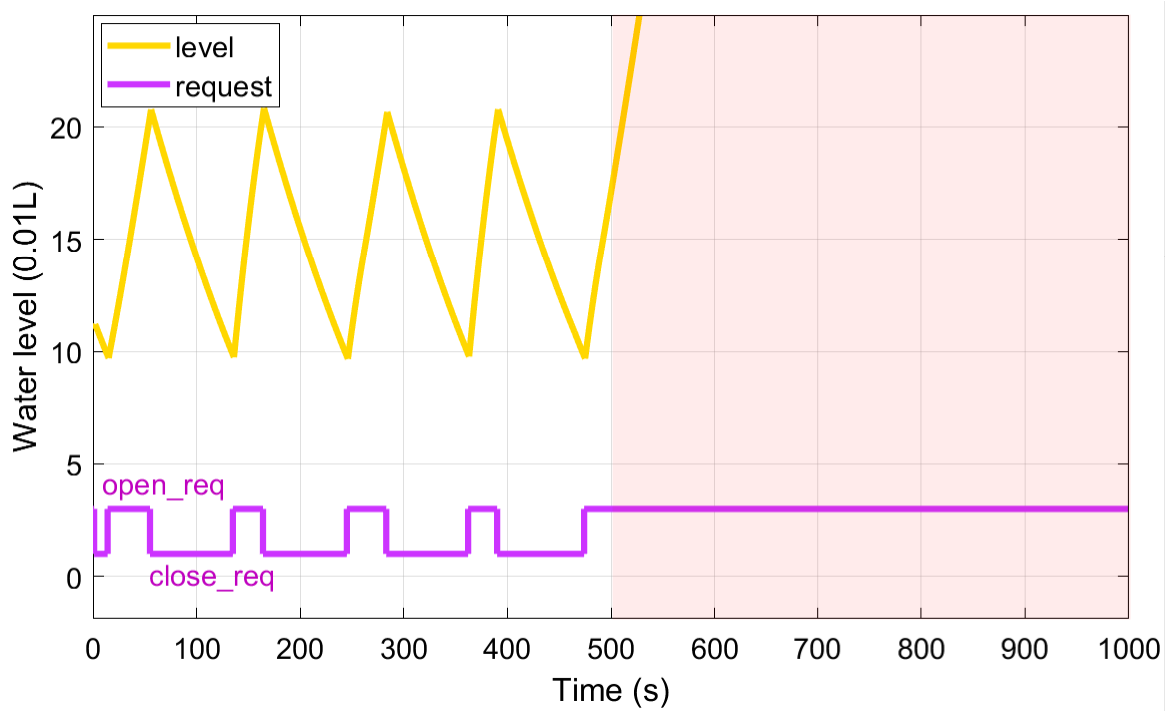}
	\end{minipage}
        \\[2pt]
	\begin{minipage}{0.44\textwidth}
		\includegraphics[
		width=6cm,keepaspectratio=true,angle=0]{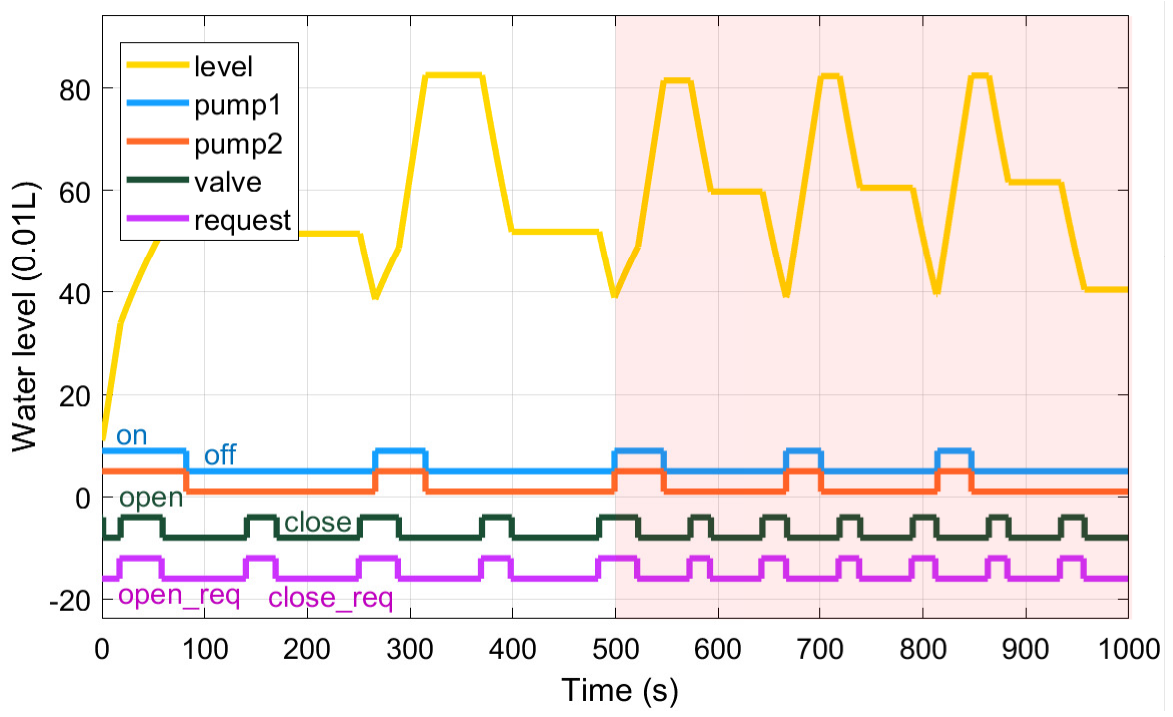}
	\end{minipage}\q
	\begin{minipage}{0.44\textwidth}
		\includegraphics[
		width=6cm,keepaspectratio=true,angle=0]{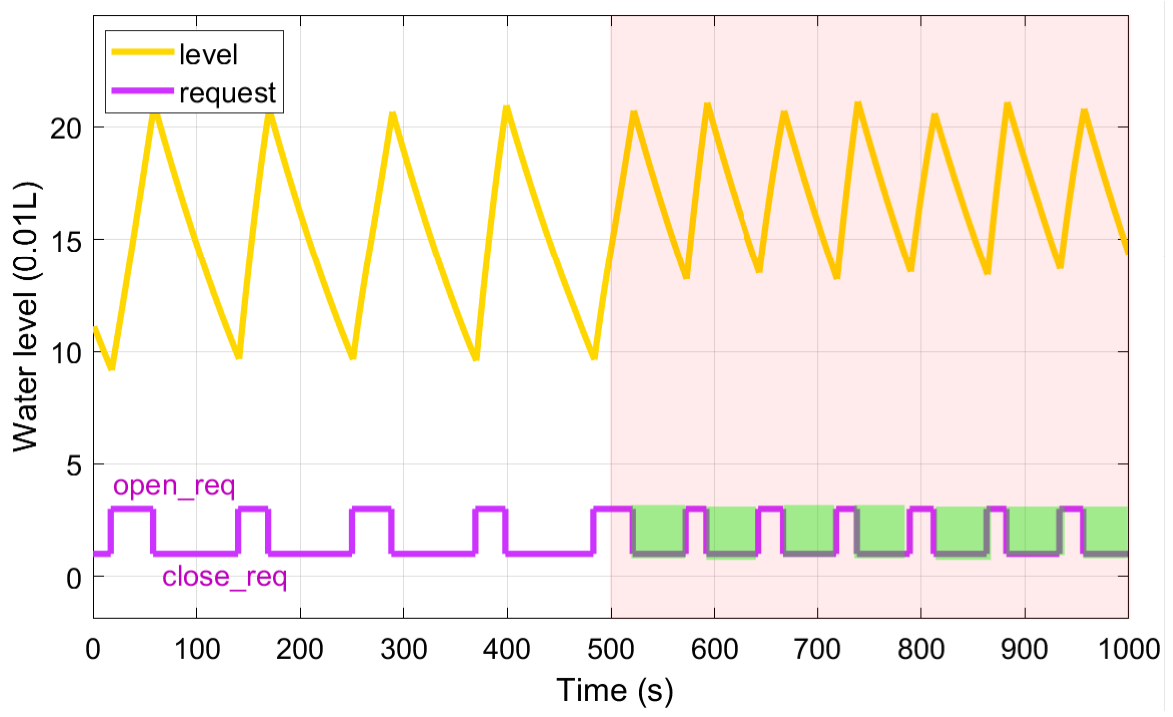}
	\end{minipage}
	\caption{Tank overflow: integrity attack on $\mathrm{PLC}_2$ when enforcing $e_1',e_3$ (up) and $e_1',e_2, e_3$ (down).}
	\label{figure:goal-1-attack-2}
\end{figure}
        Again, the red region shows when the attack becomes active. As the reader may notice, the compromised $\mathrm{PLC}_2$ never sends requests to close the valve causing  the overflow of the water tank $T_2$. On the other hand, when enforcing the three properties $e'_1, e_2, e_3$ in the three PLCs,  
the lower graphs of Figure~\ref{figure:goal-1-attack-2} shows that the overflow of  tank $T_2$ is prevented. Again,  the green highlighted regions denote when the monitor \emph{detects} the attack and \emph{mitigates} the commands of the compromised $\mathrm{PLC}_2$. Here, the monitor \emph{inserts} the request to close the valve on behalf of $\mathrm{PLC}_2$ when $T_2$ reaches a high level.


\paragraph*{Valve damage} We now consider attacks whose goal is to damage the valve via \emph{chattering}, \emph{i.e.},  rapid alternation of openings and closings of the valve that may cause  mechanical failures on the long run. In the left-hand side of Figure~\ref{fig:goal-2-ld-plc1-attack1} we show a possible ladder logic implementation of a \emph{third attack} that does \emph{injection} of the commands to open and close the valve. In particular, the attack repeatedly alternates  a \emph{stand-by phase}, lasting $70$ seconds, and a \emph{injection phase}, lasting $30$ seconds (yellow region); then, in the injection phase  the valve is opened and closed rapidly (red region).
	\begin{figure}[t]
		\centering
		\begin{minipage}{0.31\textwidth}
			\includegraphics[
		height=4.75cm,keepaspectratio=true,angle=0]{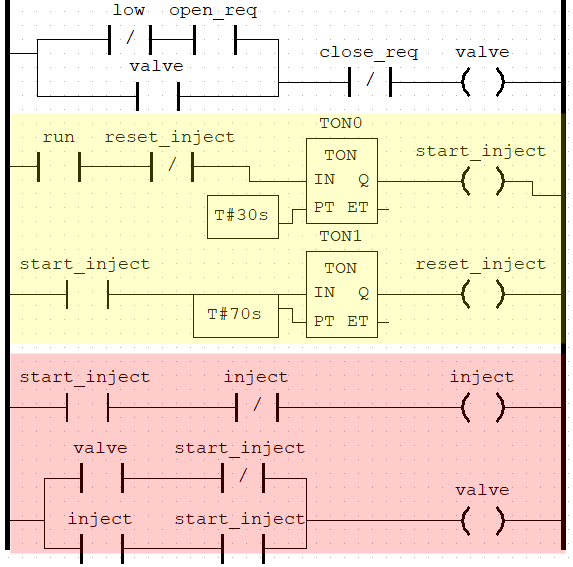}
		\end{minipage}\hspace{1cm}
		\begin{minipage}{0.4\textwidth}
		\includegraphics[
		height=4.25cm,keepaspectratio=true,angle=0]{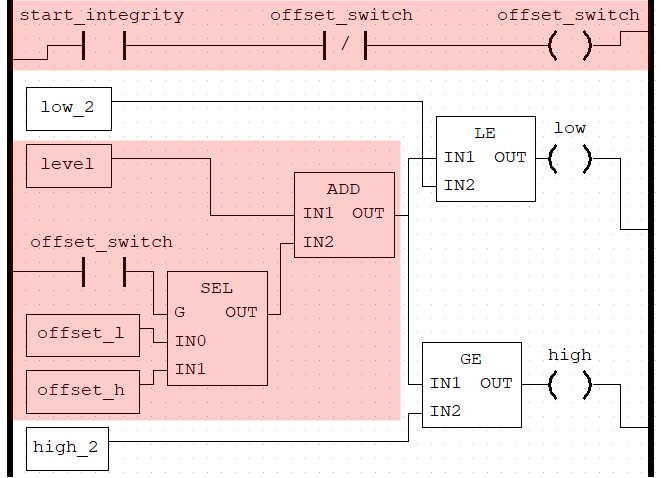}
	\end{minipage}
		\caption{Valve damage: Ladder logic of the first (left) and the second attack (right).}
		\label{fig:goal-2-ld-plc1-attack1}
	\end{figure}
        With no enforcement, the  impact of  the attack on the tanks $T_1$ and $T_2$   is represented on the upper graphs of Figure~\ref{figure:goal-2-attack-1}, where the red region denotes when the attack becomes active. From the graph associated to the execution of $T_1$ the reader can easily see that the valve is chattering. Note that  this is  a \emph{stealthy attack} as the water level of $T_2$ is maintained within the normal operation bounds. 

In order to prevent this kind of attacks, we might consider to enforce in $\mathrm{PLC}_1$ a bounded mutual exclusion property of the form  {\color{black}$e_1'' \defn (\BME{\{\overline{\open}, \overline{\close}\}}{10000})^\ast$} to ensure that within {\color{black}$10000$} consecutive scan cycles (10 seconds) openings and the closings of the valve may only occur in mutual exclusion. 
When the property $e''_1$ is enforced in  $\mathrm{PLC}_1$,  the lower graphs of Figure~\ref{figure:goal-2-attack-1} shows that the
chattering of the valve is prevented. 
In particular, the green highlighted regions denote when the monitor \emph{detects} the attack and \emph{mitigates} the commands on the valves of the compromised $\mathrm{PLC}_1$. 

\begin{figure}[t]
	\centering
	\begin{minipage}{0.44\textwidth}
	\includegraphics[
	width=6cm,keepaspectratio=true,angle=0]{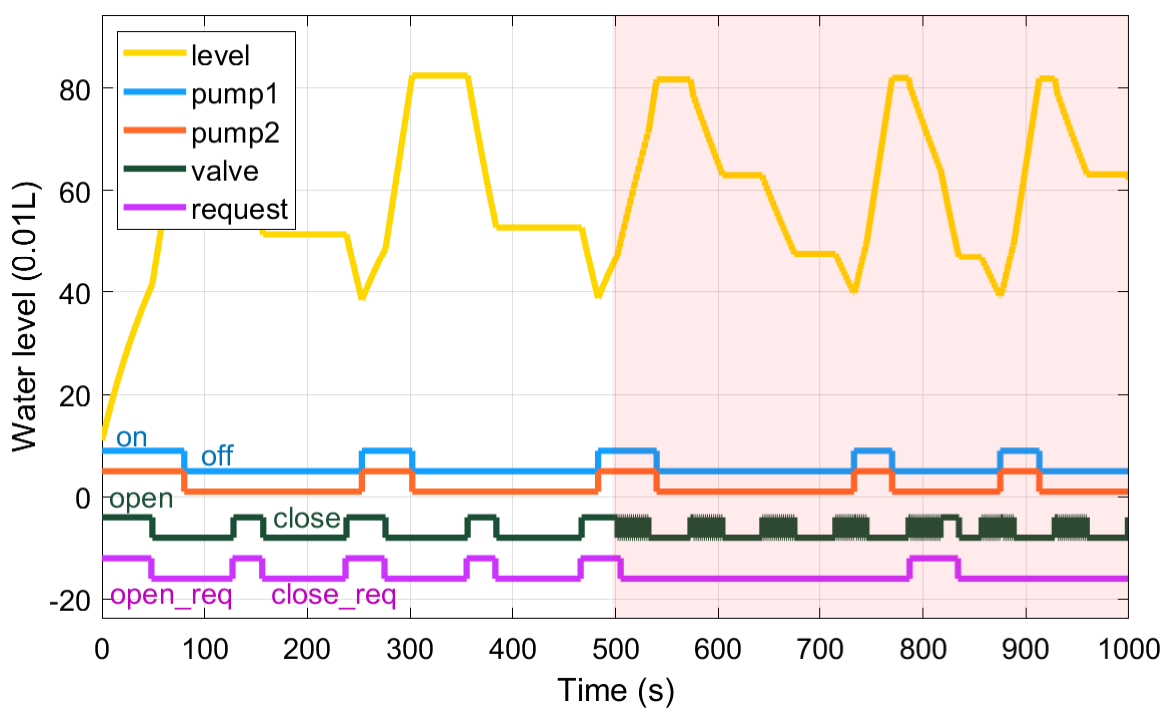}
\end{minipage}\q
\begin{minipage}{0.44\textwidth}
	\includegraphics[
	width=6cm,keepaspectratio=true,angle=0]{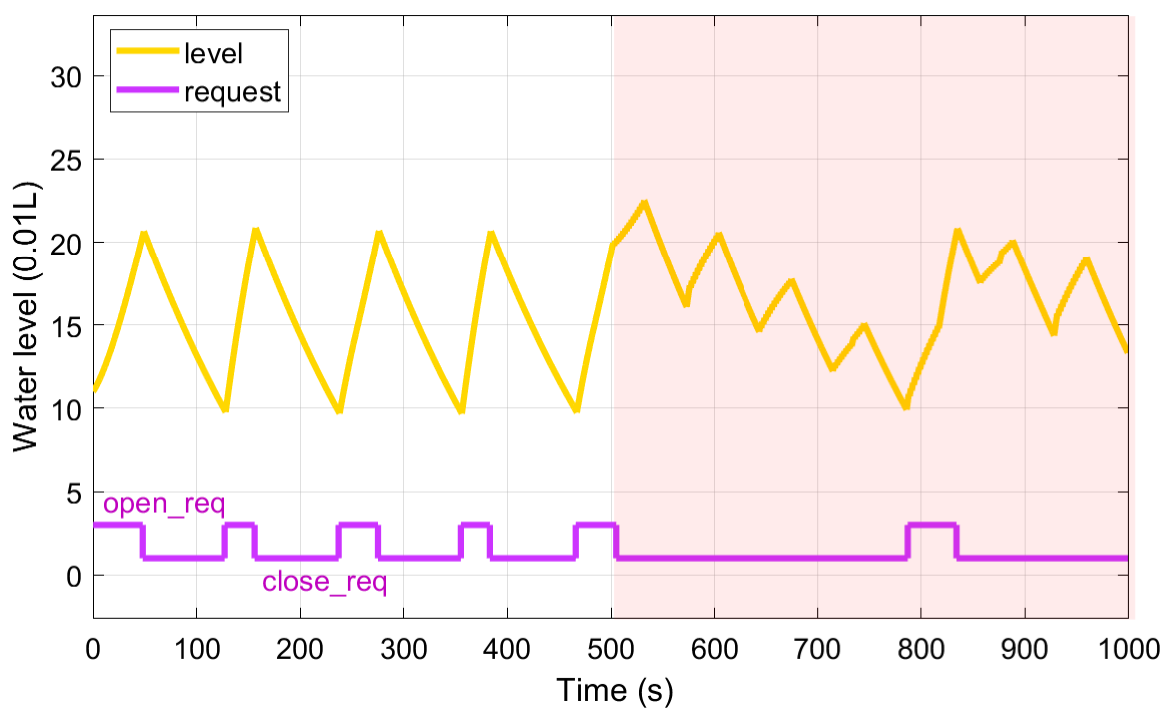}
\end{minipage}
\\[2pt]
\begin{minipage}{0.44\textwidth}
	\includegraphics[
	width=6cm,keepaspectratio=true,angle=0]{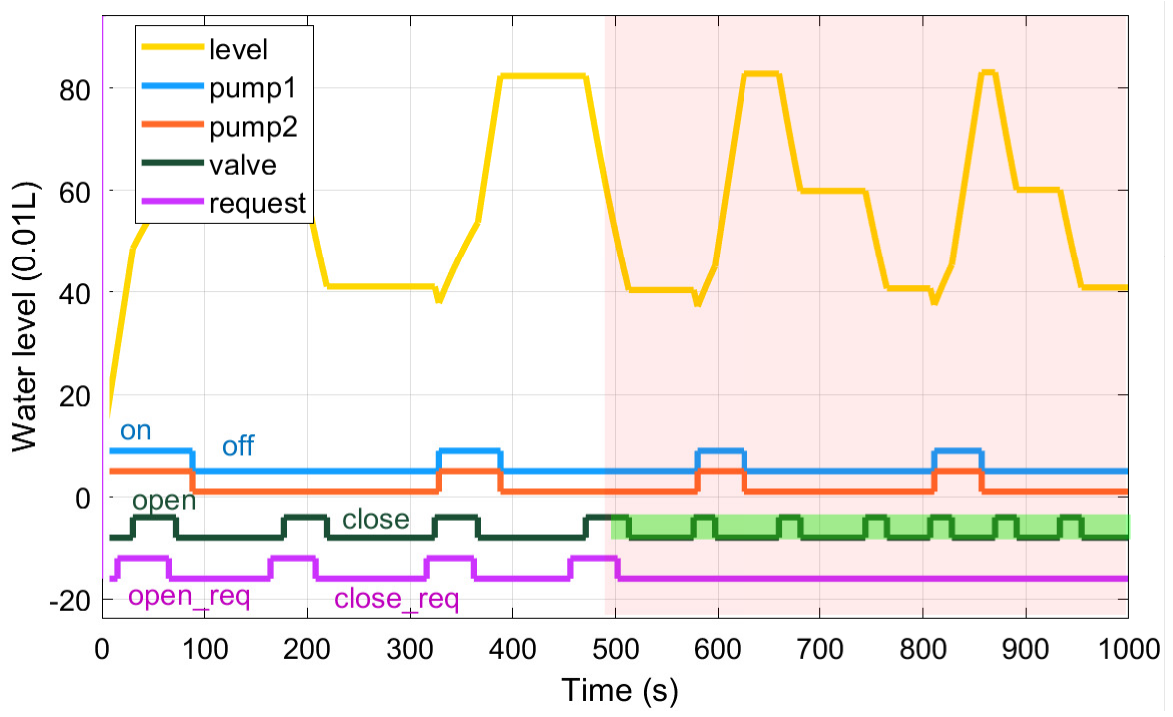}
\end{minipage}\q
\begin{minipage}{0.44\textwidth}
	\includegraphics[
	width=6cm,keepaspectratio=true,angle=0]{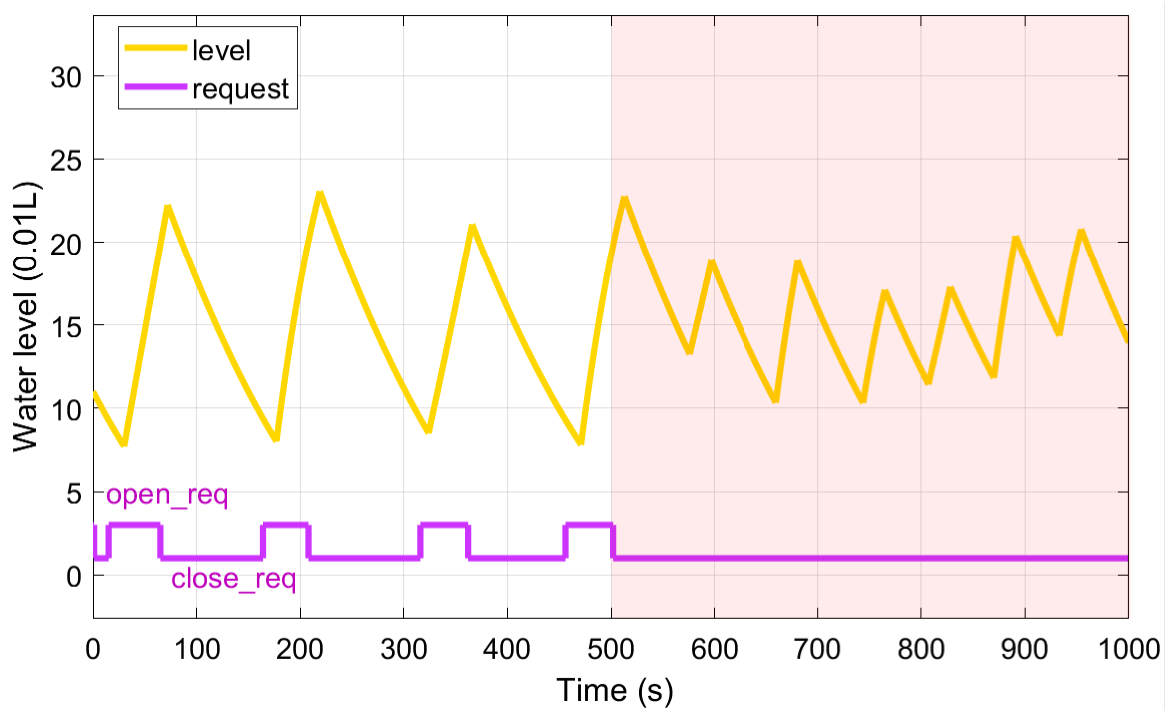}
\end{minipage}
	\caption{Valve damage: injection attack on $\mathrm{PLC}_1$ in the absence (up) and in the presence (down) of enforcement.}
	\label{figure:goal-2-attack-1}
\end{figure}
              
	A  \emph{fourth attack}  with the same goal of chattering the valve may be launched on  $\mathrm{PLC}_{2}$, by sending rapidly alternating requests to  open and close the valve. 
	This can be achieved by means of an \emph{integrity attack}  on the sensor of the tank $T_2$ by rapidly switching  the measurements between low and high. In the right-hand side of Figure~\ref{fig:goal-2-ld-plc1-attack1} we show parts of the ladder logic implementation of this attack on $\mathrm{PLC}_2$, where, for simplicity, we omit the machinery for dealing with the alternation of phases. Again, the attack repeatedly alternates between a \emph{stand-by phase}, lasting $70$ seconds, and a \emph{active phase}, lasting $30$ seconds. When the attack is in the active phase  (red region) the measured water level of $T_2$ rapidly switches between low and high, thus, sending requests to $\mathrm{PLC}_1$ to rapidly open and close the valve in alternation.

	The  impact of  this attack targeting on $\mathrm{PLC}_{2}$  in the absence of an enforcing monitor  is represented in the upper graphs of Figure~\ref{figure:goal-2-attack-2}, where the red region shows when the attack becomes active. Notice that  the rapid alternating 
  requests originating from $\mathrm{PLC}_2$ cause a chattering of the valve. 
	On the other hand, with the enforcement of the property $e''_1$ in $\mathrm{PLC}_{1}$ , the lower graph of Figure~\ref{figure:goal-2-attack-2} shows that the correct behaviour of tanks $T_1$ and $T_2$ is ensured.  In that figure, the green highlighted regions denote when the enforcer  of  $\mathrm{PLC}_{1}$  \emph{detects} the attack and \emph{mitigates} the commands (on the valve) of the compromised $\mathrm{PLC}_2$. Notice that in this case no enforcement is required in $\mathrm{PLC}_2$. 

			\begin{figure}[t]
	\centering
	\begin{minipage}{0.44\textwidth}
		\includegraphics[
		width=6cm,keepaspectratio=true,angle=0]{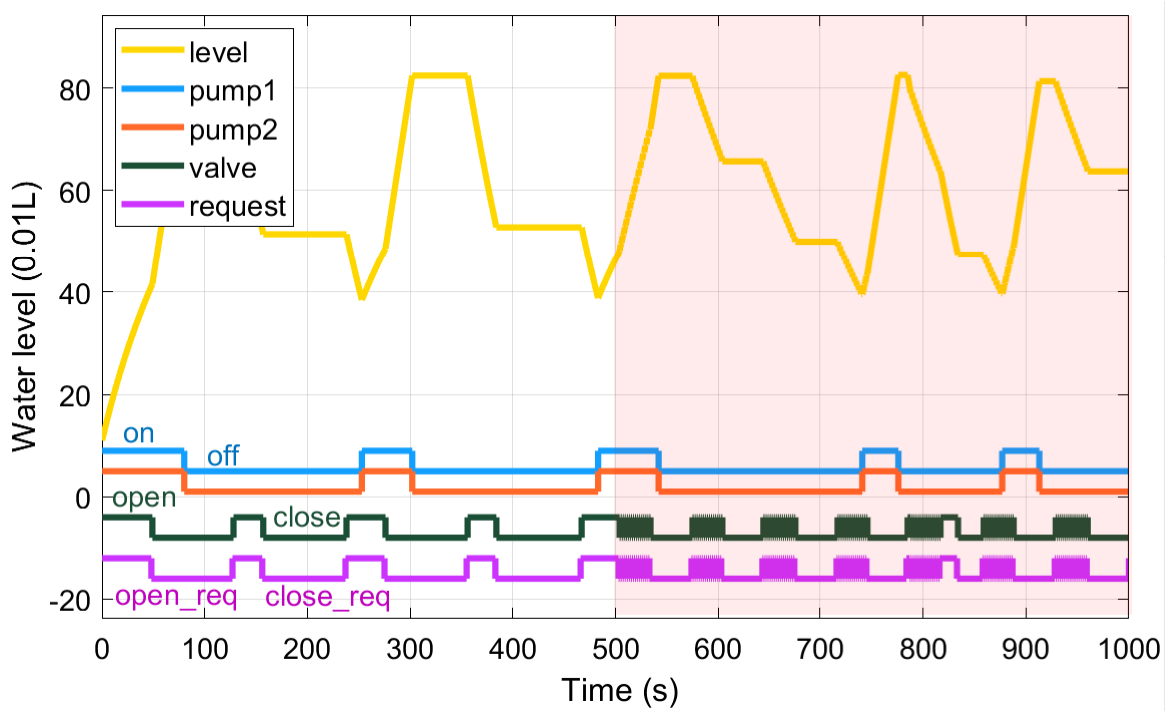}
	\end{minipage}\q
	\begin{minipage}{0.44\textwidth}
		\includegraphics[
		width=6cm,keepaspectratio=true,angle=0]{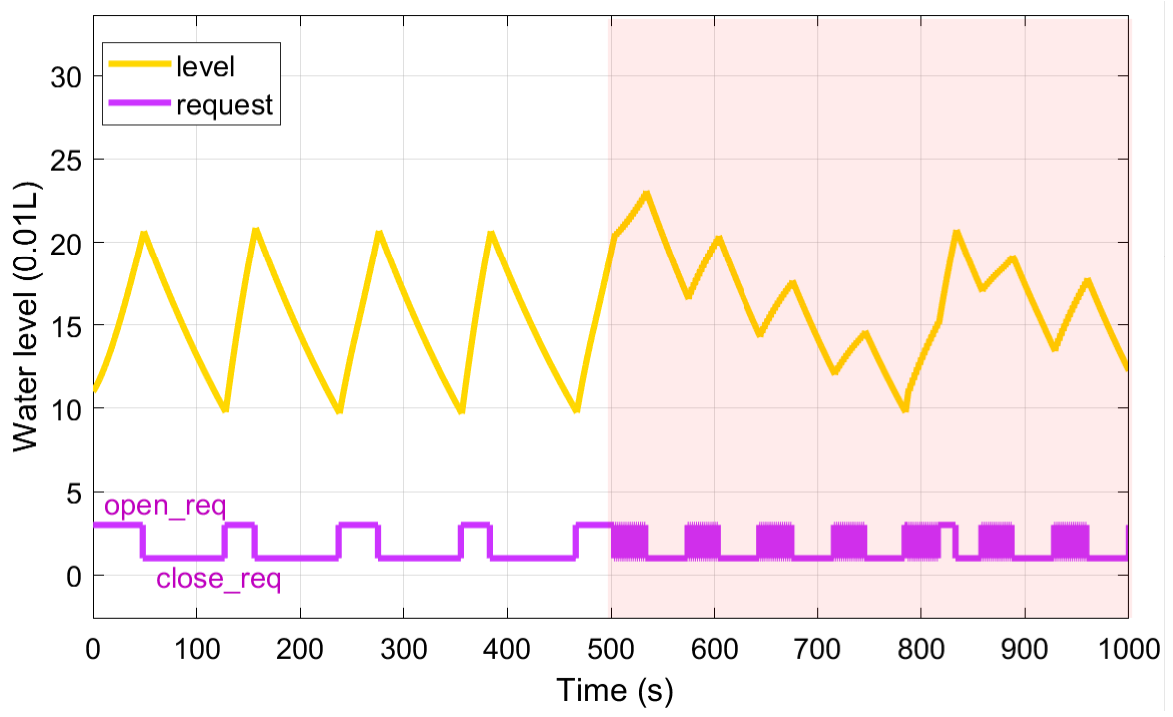}
	\end{minipage}
        \\[2pt]
	\begin{minipage}{0.44\textwidth}
		\includegraphics[
		width=6cm,keepaspectratio=true,angle=0]{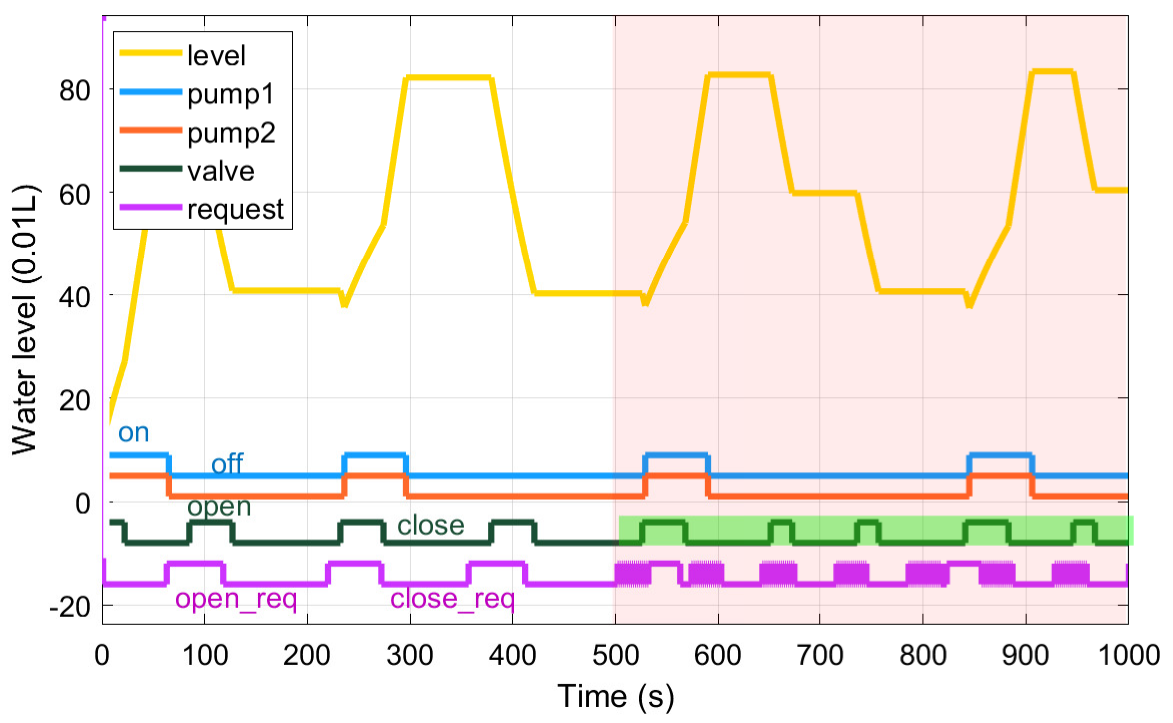}
	\end{minipage}
        \, 
	\begin{minipage}{0.44\textwidth}
		\includegraphics[
		width=6cm,keepaspectratio=true,angle=0]{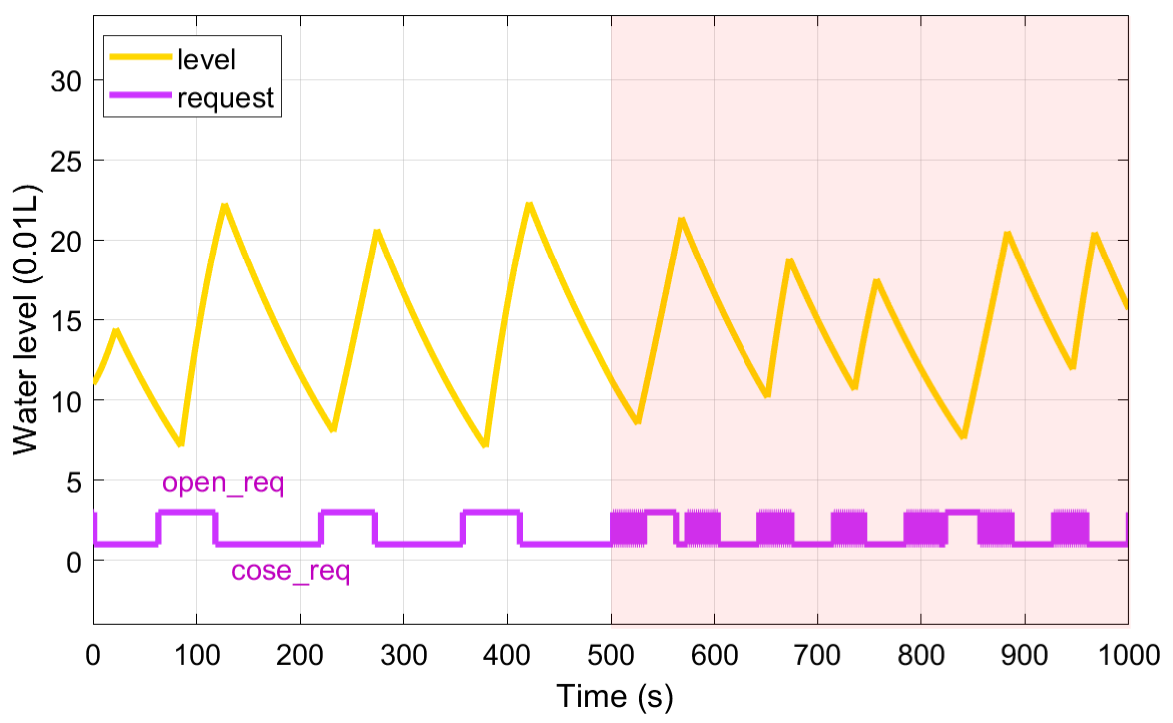}
	\end{minipage}
	\caption{Valve damage: integrity attack on $\mathrm{PLC}_2$ in the absence (up) and in the presence (down) of enforcement.}
	\label{figure:goal-2-attack-2}
\end{figure}

	                \paragraph*{Pump damage} Finally, we consider attacks whose goal is the damage of the pumps, and in particular  $\mathit{pump}_3$. In that case, an attacker may force  the pump to start when the water tank $T_3$ is empty. This can be done with a  \emph{fifth attack} that \emph{injects} commands to turn on the pump based on a ladder logic implementation similar to that seen in Figure~\ref{fig:ld-plc1-attack1}. The  impact of  this attack to  tank $T_3$  in the absence of  enforcement  is represented on the left-hand side graphs of Figure~\ref{figure:goal-3-attack-1}, where the red region shows when the attack becomes active. As the reader may notice, $\mathit{pump}_3$ is turned on when $T_3$ is empty.

                        Now, we can prevent damage on  $\mathit{pump}_3$ by enforcing on  $\mathrm{PLC}_3$ the following conditional bounded persistent property: $e_3'\defn(\CBP{l_3}{\overline{\off{3}}}{1}{w})^\ast$. The enforcement of this property  ensures that $\mathit{pump}_3$ is off for $w$ consecutive scan cycles when the level of water in tank $T_3$ is low, for $w<z$ and $z\in \mathbb{N}$ being the time (expressed in scan cycles) required fill up tank $T_3$  when the pump is off. Thus, when the enforcement of the $e'_3$  is active, the lower graphs of Figure~\ref{figure:goal-3-attack-1} shows that the correct behaviour of $T_3$ is ensured, thus  preventing pump damage. In that figure, the green highlighted regions denote when the monitor \emph{detects} the attack and \emph{mitigates} the commands (of the pumps) of the compromised $\mathrm{PLC}_3$. 
	More precisely, the enforcer \emph{suppresses} the commands to turn on the pump when the tank is empty, for $w$ consecutive scan cycles.

                        \subsection{Discussion}
                        \label{sec:discussion}

                        {\color{black}
                          In this section, we rely on the Vivado Design Suite 15.2 analysis tool to do a performance analysis  of our implementation.

As to the  \emph{hardware resources} used by our 
FPGAs, we measured them in terms of lookup tables and registers used during the enforcement.
The number of them 
depends on the number of states of the  enforcers implemented in the FPGAs. And this number 
is proportional to  the number of scan cycles involved in the enforced (local) property. In particular,
for each scan cycle, the number $k$ of states of the enforcer depends on the monitored input/output signals and their admissible values.  For instance, for scan cycles taking 10 ms  (0.1kHz), an enforced local property lasting 10  seconds 
will cover 1000 consecutive scan cycles, and the synthesised enforcer would have  $k \ast 1000$  states. In our experiments,
when
enforcing properties covering  $1000$ scan cycles the hardware resource use reaches $5$\%; for $10000$  scan cycles the resource use  rises to $13$\%. 

As for the \emph{execution speed} of the enforcement, in general all FPGAs are capable of running at a speed of 100 MHz (or higher). The actual execution speed  depends on the complexity of \nolinebreak the underlying code, in our case the  enforcer, plus some extra code  to implement the network communication protocol (UDP). 
In our experiments, FPGAs ran with a frequency of 1 MHz while PLCs ran with a frequency of 0.1-1kHz. Thus, the overhead introduced by the FPGAs is negligible, independently on the size  (the number of states) of the enforcer implemented in the FPGAs. 
We recall that in Remark~\ref{rem:maximum-time} we assumed that our enforced controllers successfully complete their scan cycle in
less than half of the maximum cycle limit (just in case the scan cycle should be entirely corrected by the enforcer). However, using FPGAs as enforcers this constraint can be actually relaxed.

			\begin{figure}[t]
	\centering
	
	\begin{minipage}{0.44\textwidth}
		\includegraphics[
		width=6cm,keepaspectratio=true,angle=0]{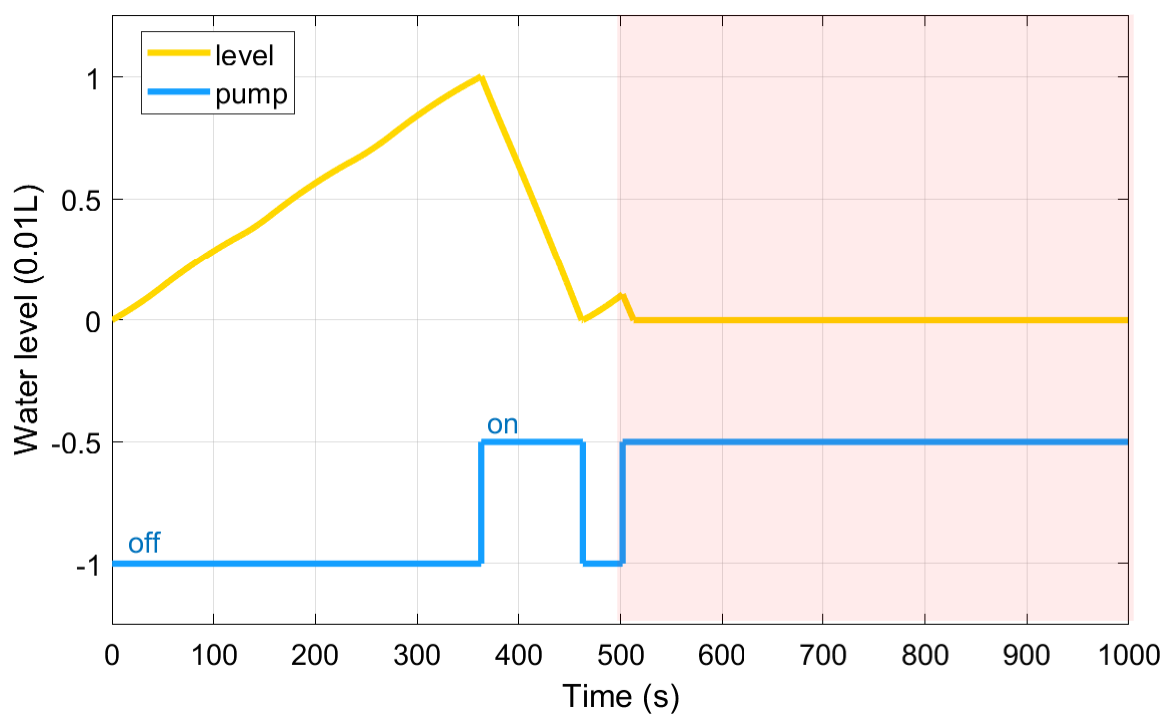}
	\end{minipage}\q
	\begin{minipage}{0.44\textwidth}
		\includegraphics[
		width=6cm,keepaspectratio=true,angle=0]{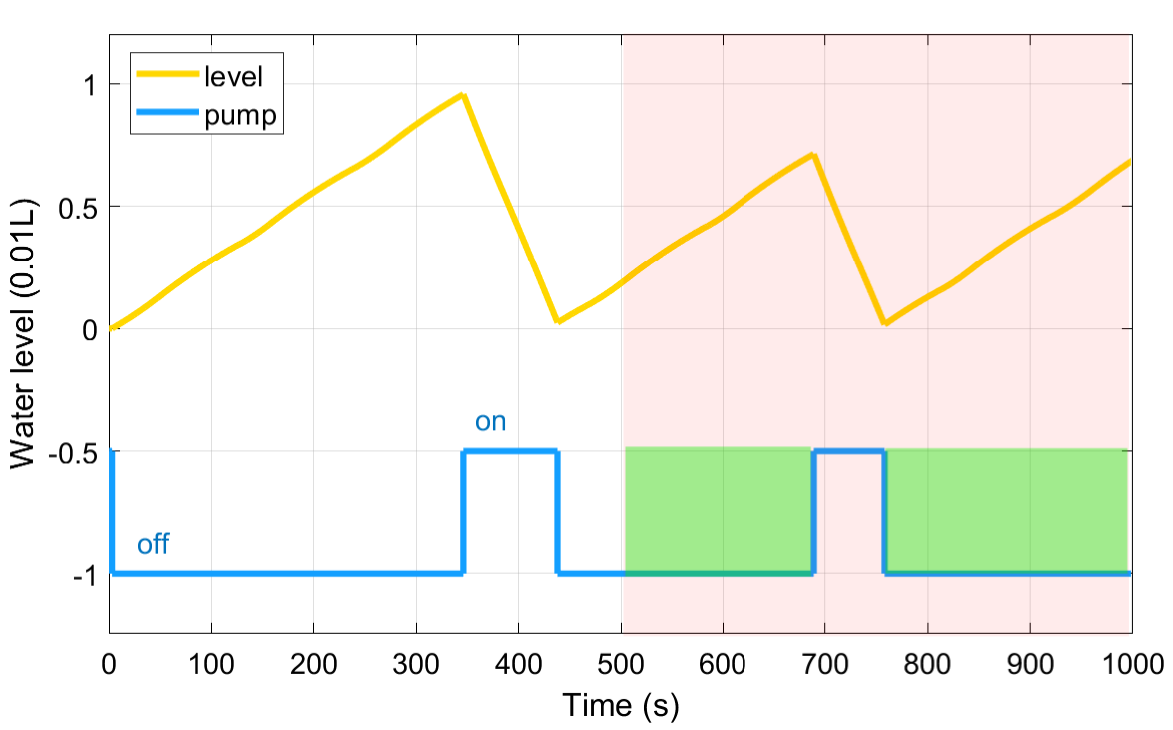}
	\end{minipage}
	\caption{Pump damage: injection attack on $\mathrm{PLC}_3$ in the absence (up) and in the presence (down) of enforcement. }
	\label{figure:goal-3-attack-1}
\end{figure}

Finally, concerning the \emph{communication latency} between enforcers, many FPGAs support high speed and low latency communications, which are the ones used in  industrial control contexts~\cite{intro_automation}. We used  FPGAs
with  Ethernet ports supporting 1 Gbps speed, \emph{i.e.}, with 
100 microseconds latency.   
Furthermore, thanks to our  result of scalability (Corollary~\ref{cor:enforcement-logic}), a network of enforcing FPGAs introduces a negligible overhead in terms of communication latency and hardware resources.


                        }

%
%
%
%
%
%
%

                        \section{Related  work}
                        \label{sec:related}

     	The notion of \emph{runtime enforcement} was introduced by Schneider~\cite{Schneider2000} to enforce security policies via \emph{truncation automata}, a kind of automata that terminates the monitored system in case of violation of the property. Thus, truncation automata can only enforce safety properties.  
{\color{black} Furthermore, the resulting enforcement may obviously lead to deadlock (actually termination) of the monitored system with no room for mitigation.}

	Ligatti et al.~\cite{Ligatti2005} extended Schneider's work by proposing the notion of  \emph{edit automata}, \emph{i.e.}, an enforcement mechanism able of \emph{replacing}, \emph{suppressing} and \emph{inserting} system actions. Edit automata are capable of enforcing instances of safety and liveness properties, along with other properties such as renewal properties~\cite{Bielova2011,Ligatti2005}.
	In general, Ligatti et al.'s edit automata {\color{black} are deterministic automata with}
      an enumerable number of states,  whereas in the current paper we restrict ourselves to finite-state edit automata equipped with Martinelli and Matteucci's operational semantics~\cite{MartinelliMatteucci2007}.
      {\color{black} 	Ligatti et al.~\cite{Ligatti2005}  studied a hierarchy of enforcement mechanisms, each with different transformational capabilities: Schneider's \emph{truncation automata}, \emph{suppression automata}, \emph{insertion automata}, and finally, \emph{edit automata} that combine the power of suppression and insertion automata. They defined different notions of enforcement, and in particular the so called \emph{precise enforcement} (Definition~2, pag.\ 5) which basically corresponds to the combination of our notion of transparency and soundness, proved in Theorems~\ref{prop:transparency-logic}  and \ref{thm:safety-logic}, respectively. }

	Bielova and Massacci~\cite{Bielova2011,Bielova2011Jrn} 
	provided a stronger notion of enforceability by introducing a \emph{predictability} criterion to prevent monitors from transforming invalid executions in an arbitrary manner. Intuitively, a monitor is said predictable if one can predict the number of transformations used to correct invalid executions.
        {\color{black} In our setting, in case of injection of a  malware which may act in an unpredictable manner, this approach appears unfeasible. 
          }  

    Falcone et al. \cite{Falcone2,Falcone3} proposed a synthesis algorithm, relying on \emph{Streett automata}, to translate most of the property classes defined within the \emph{safety-progress hierarchy}~\cite{Hierarchy_of_Temporal}
	into (a slight variation of) edit automata.     
        In the safety-progress hierarchy, our global properties can be seen as \emph{guarantee properties}  for which all execution traces that satisfy a property contain at least one prefix that still satisfies the property. {\color{black} However, it should be noticed that they consider untimed properties only; as already pointed out before, timed actions play a special role in our enforcement  and they  cannot be treated as untimed actions.}


 Beauquier et al.~\cite{BCL13} proved that finite-state edit automata (\emph{i.e.} those edit automata we are actually interested in) can only enforce a sub-class of regular properties. Actually they can enforce \emph{all and only} the  regular properties that can be recognised using finite automata whose cycles always contain at least one final state. This is the case of our enforced regular properties,  as well-formed local properties in $\PropL$  always terminate with the ``final'' atomic property $\fineC$.

 {\color{black}
 Pinisetty and Tripakis~\cite{Pinisetty-Tripakis2016} studied the \emph{compositionality} of the enforcement of different regular properties $p_1, \ldots, p_n$ at the same time, by composing the associated enforcing monitors. The idea is to replace a monolithic approach, in which a monitor is sinthesised from the property $p_1 \cap \ldots \cap p_n$, with a compositional one, where the $n$ monitors enforcing the properties $p_i$ are somehow put together to enforce  $p_1 \cap \ldots \cap p_n$. 
 The authors of~\cite{Pinisetty-Tripakis2016} proved that runtime enforcement is not compositional with respect to general regular properties, neither with respect to serial nor parallel composition. On the other hand compositionality holds for certain sub-classes of regular properties such as safety (or co-safety) properties. Here, we wish to point out that our notion of scalability is different from their notion of compositionality, as we aim at scaling our enforcement on a network of PLCs and not on multiple regular properties on \nolinebreak the \nolinebreak same \nolinebreak PLC. 
 }

  {\color{black}
  Bloem et al.~\cite{TACAS2015-Bloem} defined a synthesis algorithm that given a safety property returns a monitor, called \emph{shield}, to enforce untimed properties in  \emph{reactive systems} (which have many aspects in common with control systems).  Their algorithm rely on a notion called $k$-stabilization: when the design reaches a state where a property violation becomes unavoidable for some possibile future inputs, the shield is allowed to deviate for at most $k \in \mathbb N$ steps; if a second violation happens during the $k$-step recovery phase, the shield enters a fail-safe mode where it only enforces correctness, but no longer minimises the deviation.
  However, The $k$-stabilizing shield synthesis problem is unrealisable for many safety-critical systems, because a finite number of deviations cannot be guaranteed.  Humphrey et al.~\cite{HCV2016} addressed this problem by proposing the notion of \emph{admissible shields} which was  extended and generalised in  K\"onighofer et al.~\cite{Konighofer2017} by assuming that systems  have a cooperative behaviour with respect to the shield, \emph{i.e.}, the shield ensures a finite number of deviations if the system chooses certain outputs.  The authors presented a synthesis procedure that maximises the cooperation between system and environment for satisfying the required enforced properties.
  This approach has some similarities with our enforcement in which a violation of a property during a scan cycle induces the suppression of all subsequent controller actions until the PLC reaches the end of the scan, so the monitor can insert a safe trace before permitting the completion of the scan cycle. 
        }

  {\color{black}
  Pinisetty et al.~\cite{Pinisetty_2017} proposed a \emph{bi-directional runtime enforcement} mechanism  for reactive systems, and more generally for cyber-physical relying on Berry and Gonthier's {synchronous hypothesis}~\cite{Esterel}, to correct both inputs and outputs.
  Pinisetty et al.\ express safety properties in terms of \emph{Discrete Timed Automata}  (DTA)
  which are more expressive than our class of regular properties.
  Thus, an execution trace satisfies a required property only if it ends up on a final state of the corresponding DTA.   However, as not all regular properties can be enforced~\cite{BCL13}, they proposed a more permissive enforcement mechanism  that accepts execution traces as long as  there is still the possibility of reaching a final state. 
  Furthermore, due to the instantaneity of the synchronous approach, their enforcement actions are applied in the same reaction step to ensure reactivity. On the contrary, in our approach the enforcement takes places before the conclusion of scan cycles which are clearly delimited via $\fineC$-actions. Our notion of deterministic enforcers is taken from Pinisetty et al.~\cite{Pinisetty_2017}. 
    Moreover, when inserting safe actions,  our synthesised enforcers follows Pinisetty et al.'s \emph{random edit} approach, where the inserted  safe action is  randomly chosen from a list of admissible actions. 
  }

  {\color{black}
    Pearce et al.~\cite{Pearce_2019} proposed a bi-directional runtime enforcement  over valued signals for PLCs, by introducing \emph{smart I/O modules} (similar to our secure proxy) between the PLCs and the controlled physical processes, to act as an effective line of defence. The authors express security properties in terms of \emph{Values Discrete Timed Automata}  (VDTA), inspired by the DTA of   Pinisetty et al.~\cite{Pinisetty_2017}. Unlike DTA, VDTA support valued signals, internal variables, and  guard conditions.  As in   Pinisetty et al.~\cite{Pinisetty_2017}, the paper adopts the synchronous hypothesis~\cite{Esterel} to correct both inputs and outputs; thus,  their enforcement actions are applied in the same reaction step to ensure instantaneous reactivity. The authors  do not consider attacks that may tamper with inter-controller communications: their attackers may only manipulate sensor signals and/or actuator commands. Finally, their semantics  requires that every enforcer knows the state of all relevant signals and commands in a given system. Thus, as written by the same authors,  a networked system featuring multiple I/O modules may significantly complicate the enforcement, as pertinent I/O for a security policy may not be locally available. As a consequence, unlike us, their enforcement does not naturally scale to networks of controllers; we believe this is  basically due to the fact that they do bi-directional enforcement. 
      Last but not least, like them, we implement enforcers via FPGAs to ensure efficiency and security at the same time. In particular, when inserting safe actions our implementation fixes a  priority between admissible safe actions, 
       similarly to their \emph{selected edit}  approach. However, our  implementation differs from theirs in at least the  following aspects: 
    \begin{inparaenum}
    \item  our FPGAs do enforce PLC transmissions 
      (with a negligible latency);
    \item our enforcement is uni-directional and hence our FPGAs need to know
      only the state of  signals and commands of the corresponding enforced PLCs; 
      \item as a consequence, our FPGAs can be networked to monitor field communications networks paying only negligible overhead in terms of computational resources and communication latency. 
    \end{inparaenum}
           }

  Aceto et al.~\cite{AcetoCC2018}  developed an operational framework to  enforce properties in HML logic with recursion ($\mu$HML) relying on suppression. More precisely, they achieved the  enforcement of  a  safety  fragment of $\mu$HML by providing a linear automated synthesis algorithm that generates correct suppression monitors from formulas.  Enforceability of  modal $\mu$-calculus (a reformulation of $\mu$HML) was previously tackled  by Martinelli and Matteucci~\cite{MartinelliMatteucci2007} by means of a synthesis algorithm which is exponential in the length of the enforceable formula.    Cassar~\cite{Cassar_PhD} defined a general framework to compare different enforcement models and different correctness criteria, including optimality. His works focuses on the enforcement of a safety fragment of $\mu$HML, paying attention to both uni-directional and bi-directional notions of enforcement. {\color{black}
    More recently,  Aceto et al.~\cite{AcetoFORTE2021} developed an operational framework for bi-directional enforcement and used it to study the enforceability of the aforementioned safety fragment of HML with recursion, via a specific type of bi-directional enforcement monitors called \emph{action disabling monitors}. 
}

        As regards papers in the context of \emph{control system security}  closer to our objectives, 	McLaughlin~\cite{McLaughlin-ACSAC2013} proposed the introduction of an enforcement  mechanism, called $\mathrm{C}^2$, similar to our secure proxy, to mediate the control signals $u_k$ transmitted by the PLC to the plant. Thus, like our secured proxy, $\mathrm{C}^2$ is able to suppress commands, but unlike our proxy, it cannot autonomously send commands to the physical devices in the absence of a timely correct action  from the PLC. Furthermore, $\mathrm{C}^2$ does not seem to cope with inter-controller communications,  and hence with colluding malware operating on  PLCs of the same  field network.

        Mohan et al.~\cite{Mohan-HiCONS2013} proposed a different approach by defining an ad-hoc security architecture, called \emph{Secure System Simplex Architecture} (S3A), with the intention to generalise the notion of ``correct system state'' to include not just the physical state of the plant but also the \emph{cyber state} of the PLCs \nolinebreak of the system. 	In S3A,  every PLC runs under the scrutiny \nolinebreak  of a \emph{side-channel monitor} which looks for deviations with respect to \emph{safe executions}, taking  care of real-time constraints, memory usage, and communication patterns.  If the information obtained via the monitor differs from the expected model(s)  of  the  PLC, a \emph{decision module} is informed to decide whether to pass the control  from the ``potentially compromised'' PLC to a \emph{safety controller} to maintain the plant within the required safety margins. As reported by the same authors, S3A has a number \nolinebreak of limitations comprising: (i) the possible compromising of the side channels used for monitoring, (ii)  the tuning of the timing parameters of the state machine, which is still a manual  process.

    The present work is a revised extension of the  conference version appeared in~\cite{CSF2020}.  
Here, we provide a detailed comparison with that paper. In Section~\ref{sec:attacker-model} we specified the attacker model and the attacker objectives.  In Section~\ref{sec:calculus}, we adopted a simplified operational semantics for edit automata, in the style of  Martinelli and Matteucci~\cite{MartinelliMatteucci2007}. 
         In Section~\ref{sec:logics}, we have extended our language of regular properties with intersection of both local and global properties. With this extension we have expressed a wide family of correctness properties that can be combined in a modular fashion; these properties  include and extend the three classes of properties appearing in the conference paper.
          In Section~\ref{sec:synthesis}, we have extended our synthesis algorithm  to deal with our extended properties: both local and global intersection of properties are synthesised in terms of cross products of edit automata. Notice that, compared to the conference paper, our enforcement mechanism does not rely anymore on an ad-hoc semantic rule \rulename{Mitigation} to insert safe actions at the end of the scan cycle, but rather on the more standard rule \rulename{Insert} together with the syntactic structure of synthesised enforcers.
          As stated in Proposition~\ref{prop:poly2}, now our synthesis algorithm depends on the size  and the number of occurrences of intersection operators of the property in input.
 Last but not least, in this journal version we provide an implementation of our use case based on: (i) Simulink to simulate the physical plant, (ii) OpenPLC on Raspberry Pi to run open PLCs, and (iii) FPGAs to implement enforcers.  We have then exposed our implementation  to five different attacks targeting the PLCs and discussed the effectiveness of the proposed enforced mechanism. 

 In a preliminary work~\cite{TCS-2021}, we  proposed an extension of our process calculus with an explicit representation for malware  code. In that paper,  monitors are synthesised   from  PLC code rather than correctness properties. The focus of that paper was mainly  on: (i)  deadlock-free enforcement, and (ii)  intrusion detection via secure proxies. 
 {\color{black}Here, it is worth pointing out that the work in~\cite{TCS-2021} shares some similarities with \emph{supervisory control theory}~\cite{DES1987,DEST1994},  a general theory for automatic synthesis of controllers (supervisors) for \emph{discrete event systems}, given a plant model and a specification for the controlled behaviour.  Fabian and Hellgren~\cite{Fab-Hell1998} have pointed out a number of issues to be addressed when adopting supervisory control theory in industrial PLC-based facilities, such as causality, incorrect synchronisation, and choice between alternative paths. However,  as our syntheses regard only logical devices (no plant involved), we are not affected from similar problems.
   }

 {\color{black} Finally, Yoong et al.~\cite{Yoong2009} proposed a synchronous semantics for \emph{functions blocks},
   a component-oriented model at the core of the IEC 61499 international standard~\cite{61499}  used to design distributed industrial process measurement and control systems. In contrast to the scan cycle model followed in the current paper (IEC 61131~\cite{61131-3}) prescribing the execution of a sequential portion of code at each scan cycle, the \emph{event-driven} model for function blocks relies on the occurrence of asynchronous events to trigger program execution. Yoong et al.~\cite{Yoong2009} adopted a synchronous approach to define an execution semantics to function blocks by translating them into a subset of Esterel~\cite{Esterel}, a well-known synchronous language. Here, we wish to point out that our PLC specification is given at a more abstract level compared to that of \cite{Yoong2009}, and it complies with the sequential scan cycle standard IEC 61131, rather than the event-driven standard IEC 61499.
          }
 

	\section{Conclusions and future work}
	\label{sec:conclusion}

        We have defined a formal language to express  networks of monitored controllers, potentially compromised with colluding malware that may forge/drop actuator commands, modify sensor readings, and forge/drop inter-controller communications. The enforcing monitors have been expressed via a finite-state sub-class of Ligatti et al.'s edit automata. {\color{black} In this manner, we have provided a formal representation of field communications networks in which controllers are enforced via secure monitors, as depicted in Figure~\ref{fig:Sys-structure}. The room of manoeuvre of attackers is defined via a proper attacker model. }
        Then, we have defined a simple description language to express
        \emph{timed regular properties} that are recognised by finite automata whose cycles always contain at least one final state (denoted via an 
        $\fineC$-action). {\color{black}We have used that language to build up formal definitions for \emph{pattern templates} suitable for expressing a broad family of correctness properties that can be combined in a modular fashion to prescribe precise controller behaviours.   As an example,  our description language allows us to capture all (bounded variants of the)  controller properties studied in Frehse et at.~\cite{toolchain}.} 
        Once defined a formal language to describe controller properties, we  have provided  a \emph{synthesis function}  $\funEdit{\!-\!}$ that, given an alphabet $\PSet$ of observable controller actions 
        and  a deterministic regular property $e$ consistent with $\PSet$, returns a finite-state deterministic 
edit automaton $\funEditP{e}\PSet$.
The resulting enforcement mechanism will  ensure the required features advocated in the Introduction:  transparency, soundness, deadlock-freedom, divergence-freedom, mitigation and scalability.

         %

As a final contribution, we have provided a full implementation of a non-trivial case study in the context of industrial water treatment, where enforcers are implemented via FPGAs.  In this setting, we showed the effectiveness our enforcement mechanism when exposed to  five carefully-designed  attacks targeting the PLCs of our use case.


{\color{black}
  As \emph{future work}, we wish to test our enforcement mechanism in different  domains, such as industrial and cooperative robotic arms (\emph{e.g.}, Kuka, ABB, Universal Robots, etc) which are endowed with control architectures working at a fixed rate~\cite{siciliano2009modelling}. 
More generally, we would like to consider physical plants with significant uncertainties, in terms of measurements noise and physical process uncertainty.  This is because significant plant perturbations  might falsely indicate that the monitored  controller is under attack, inducing our enforcers to take erroneous correcting actions.  To address such challenges we would like to implement in  our enforcers  well-know control-theory algorithms, based on linear difference equations, to correctly estimate the state of the physical plant even when  affected by significant uncertainties.
  Finally, we would like to enhance our enforcers to deal with malicious alterations of sensor measurements due to compromised sensor devices. In order to  do so, we intend to  integrate our secured proxies with physics-based attack detection mechanisms~\cite{Cardenas-etal2011,ACM-survey2018}.
}

\begin{acks}
  We thank the anonymous reviewers for their insightful and careful reviews. 
 We  thank Adrian Francalanza, Yuan Gu, Marjan Sirjani and  Davide Sangiorgi for their comments on early drafts of the paper. The authors have been partially supported by the project ``Dipartimenti di Eccellenza 2018--2022'' funded by the Italian Ministry of  Universities and Research (MUR).
\end{acks}

	\bibliographystyle{ACM-Reference-Format}
	\bibliography{main}

\appendix

\section{Proofs}

    {\color{black}
	In order to prove the results of  Section~\ref{sec:synthesis}, in Table~\ref{tab:edit-automata-product}
        we provide the technical definition of cross product between two edit automata used in the synthesis of
        Table~\ref{tab:synthesis-logic}.  As the first three cases are straightforward, we explain only the fourth case,
  the cross product associated to $\TimesP{\sum_{i \in I}\lambda_i.\Edit_i}{\sum_{j \in J}\nu_j.\Edit_j}{\Zrec}$ 
  Here, we use the abbreviation  $\lambda.\Edit  \oplus \lambda'.\Edit$ to denote the automaton $\lambda.\Edit  + \lambda'.\Edit$, if $\lambda \neq \lambda'$, and, the automaton $\lambda.\Edit $, if $\lambda  = \lambda'$. 
  Thus, the product does the intersection of those addends $\lambda_i.\Edit_i$ and $\nu_j.\Edit_j$, with $(i,j) \in H$,
  for which: (a) the prefixes have the same output (\emph{e.g.}, $\lambda_i=\alpha$ and $\nu_j=\alpha < \alpha'$), (b)
  the prefixes are not suppressions, (c) the product of their continuations $\Edit_i$ and $\Edit_j$ ``is not empty'', \emph{i.e.}, it is not a suppression-only automaton. For the other addends $\lambda_i.\Edit_i$ and $\nu_j.\Edit_j$ which do not comply with the above conditions (\emph{i.e.}, $(i,j) \not \in H$),   the product results in a suppression-only automaton. 
}

        	\begin{table}[t]
	{\color{black}
			{\footnotesize		
		\begin{displaymath}
				\begin{array}{rcl}
					\TimesP{\Xrec_1}{\Xrec_2}{\Zrec} &\defn& 
					{ \TimesP{\Edit_1}{\Edit_2}{\Zrec}}, \text{ if } \Xrec_1=\Edit_1 \text{ and } \Xrec_2 = \Edit_2 \\[5pt]
					\TimesP{\Xrec}{\sum\limits_{i \in I}
						\lambda_i.\Edit_i}{\Zrec} &\defn& 
				{ \TimesP{\Edit}{\sum\limits_{i \in I}\lambda_i.\Edit_i}{\Zrec}}, \text{ if } \Xrec =\Edit 
					\\[5pt]
					\TimesP{\sum\limits_{i \in I}
						\lambda_i.\Edit_i}{\Xrec}{\Zrec} &\defn& 
				\TimesP{\sum\limits_{i \in I}\lambda_i.\Edit_i}{\Edit}{\Zrec}, \text{ if } \Xrec =\Edit \\[5pt]

					\TimesP{\sum\limits_{i \in I}
						\lambda_i.\Edit_i}{\sum\limits_{j \in J}
					  \nu_j.\Fdit_j}{\Zrec} &\defn&
	
  \sum\limits_{(i,j)\in H} (  {\lambda_i}. \Xrec_{i,j}  \oplus  {\nu_j}. \Xrec_{i,j})   
		                                + \sum\limits_{  \alpha \in      {\mathcal Q} } {^{-}}\alpha.\Zrec,      \textrm{ for }
                                                			                        \\[4pt]
                                           &&    
                                                                                                 \Xrec_{i,j}=\TimesP{\Edit_i}{\Fdit_j}{ \Xrec_{i,j}}    \\[1pt]
                                           &&    
                                                                                                     {\mathcal Q}= (\PSet \setminus \{\tick,\fineC\}) \setminus  \bigcup_{\tiny (i,j) \in H}\{\lambda_i,\nu_j\}
                                                    \\[1pt]                                                         
                                                    &&
                                                     H =\{ (i,j) {\in} I {\times} J: \sysAct{\lambda_i} {=} \sysAct{\nu_j}
                                                    {\neq} \tau { \text{ and } \TimesP{\Edit_i}{\Fdit_j}{  \Xrec_{i,j} }}\neq \sum\limits_{  \alpha \in  \PSet \setminus  \{\tick,\fineC\} }\!\!{^{-}}\alpha. \Xrec_{i,j}  \}
			\end{array}	
		\end{displaymath}
                }	
		\caption{Cross product between two edit automata with alphabet $\PSet$.}
		\label{tab:edit-automata-product}
	        }
	        \end{table}

        Let us prove the complexity of the synthesis algorithm formalised in Proposition~\ref{prop:poly2}. For that we need three technical lemmata. 
        {\color{black}
        	The first lemma shows that our  synthesis algorithm always returns an edit automaton  in  a specific canonical form.  }
        {\color{black}
        	\begin{lemma}[Canonical Form] 
        		\label{lem:canonical}
        		Let $e \in \PropG$   and $\PSet$ be a set of actions such that $\events {e} \subseteq \PSet$. 
        		Then,  either  $\synthESP{e}{}{\PSet}=\Edit $ or  $\synthESP{e}{}{\PSet} = \Zrec$, with $\Zrec = \Edit$,  for $\Edit$ of the following form: 
        		\[
        		\Edit=
        		\begin{cases}
        			\sum\limits_{\substack{{\scriptscriptstyle i\in I} }}
        			\alpha_i.\Edit_i 
        			+ \sum\limits_{\substack{{\scriptscriptstyle i\in I} }}
        			{\scriptstyle \eact{\, \fineC}{\alpha_i}}.\Edit_i 
        			
        			+ \sum\limits_{\subalign{{\scriptscriptstyle \alpha  \in  {\mathcal Q} }}} {^{-}}\alpha.\Fdit, \textrm{  if }  \fineC   \not \in \cup_{i \in I}\alpha_i  \\[2pt]
        			
        			\sum\limits_{ \substack{{\scriptscriptstyle i\in I} }}
        			\alpha_i.\Edit_i 
        			+ \sum\limits_{\subalign{{\scriptscriptstyle \alpha  \in  {\mathcal Q}   }}} {^{-}}\alpha.\Fdit  , \textrm{ otherwise.}                      
        		\end{cases}
        		\]	 where $\alpha_i \in \PSet$, ${\mathcal Q}= {\PSet} \setminus  (\cup_{i \in I}\alpha_i\cup\{\tick,\fineC\})$, and $\Edit_i$ and $\Fdit$ edit automata. A similar result holds when $e$ is replaced with some local property  $p \in \PropL$. 
        	\end{lemma}
        	\begin{proof}
        		The proof is  by induction on the structure of the property $e$.  The most interesting case is when $e= e_1\cap e_2$. Then,  \smash{$\synthESP{e_1\cap e_2}{}{\PSet}$} returns \smash{$\TimesP{\synthESP{e_1}{}{\PSet}}{\synthESP{e_2}{}{\PSet}}{\Xrec}$}.
        		By inductive hypothesis,  $\synthESP{e_1}{}{\PSet}$ and 
        		$\synthESP{e_2}{}{\PSet}$ have the required form. We prove the case when
        		\begin{itemize}
        			\item $\synthESP{e_1}{}{\PSet}= \sum\limits_{\substack{{\scriptscriptstyle i\in I} }}
        			\alpha_i.\Edit_i
        			+ \sum\limits_{\substack{{\scriptscriptstyle i\in I} }}
        			{\scriptstyle \eact{\, \fineC}{\alpha_i}}.\Edit_i 
        			+ \sum\limits_{\subalign{{\scriptscriptstyle \alpha \in  {\mathcal Q}_1 }}} {^{-}}\alpha.\Edit'$,   
        			with {\small ${\mathcal Q}_1={\PSet} \setminus  (\cup_{i \in I}\alpha_i\cup\{\tick,\fineC\})$} and $\fineC \not \in \cup_{i \in I}\alpha_i$
        			\item 
        			$\synthESP{e_2}{}{\PSet}=  \sum\limits_{ \substack{{\scriptscriptstyle j\in J} }}
        			\alpha_j.\Fdit_j
        			+ \sum\limits_{\subalign{{\scriptscriptstyle \alpha \in  {\mathcal Q}_2 }}} {^{-}}\alpha.\Fdit'$,
        			with {\small ${\mathcal Q}_2={\PSet} \setminus  (\cup_{i \in I}\alpha_i\cup\{\tick,\fineC\})$} and $\fineC   \in \cup_{j \in J}\alpha_j$. 
        		\end{itemize}
        		The other cases are similar or simpler. 
        		For any $i \in I$ and $j\in J $, we have: (i) 
        		$\sysAct{\alpha_i}= \sysAct{ \alpha_j}$  if and only if  $\alpha_i =\alpha_j$;  (ii)   $\sysAct{\eact{\! \fineC}{\alpha_i\!}}=\sysAct{\alpha_j}$ holds if and only if   $\alpha_i =\alpha_j$. We recall that
        		$\sysAct{^{-}\alpha}= \tau$. 
        		Thus, the set $H$ of the definition of cross product in Table~\ref{tab:edit-automata-product} for 
        		$\TimesP{\synthESP{e_1}{}{\PSet}}{\synthESP{e_2}{}{\PSet}}{\Xrec}$
        		is equal to $\{(i,j) \in I \times J:  \alpha_i  =  \alpha_j  { \text{ and } \TimesP{\Edit_i}{\Fdit_j}{  \Xrec_{i,j} }}\neq \sum_{  \alpha \in \PSet \setminus  \{\tick,\fineC\} }{^{-}}\alpha. \Xrec_{i,j}  \}$, with $\Xrec_{i,j}=\TimesP{\Edit_i}{\Fdit_j}{ \Xrec_{i,j}} $.
        		As a consequence, we  derive   
        		\[ \TimesP{\synthESP{e_1}{}{\PSet}}{\synthESP{e_2}{}{\PSet}}{\Xrec}=  
        		\sum\limits_{(i,j)\in H}     \alpha_i. \Xrec_{i,j} +  \sum\limits_{(i,j)\in H} {\scriptstyle \eact{\, \fineC}{\alpha_i}}. \Xrec_{i,j}+
        		\sum\limits_{\subalign{{\scriptscriptstyle \alpha \in  {\mathcal Q}  }}} {^{-}}\alpha.\Xrec \]			  with {\small ${\mathcal Q} ={\PSet} \setminus  ( \cup_{(i,j) \in H}\alpha_i\cup\{\tick,\fineC\})$}.
        		It remains to prove that $\fineC \not \in \cup_{(i,j) \in H}\alpha_i$. Since  $\fineC \not \in \cup_{i \in I}\alpha_i$  and  
        		$\fineC   \in \cup_{j \in J}\alpha_j$, then there is no $(i,j)\in H$ such that 
        		$\alpha_i=\fineC$. Thus,   $\fineC \not \in \cup_{(i,j) \in H}\alpha_i$,  as required.
        	\end{proof}
        }

	By an application of Lemma~\ref{lem:canonical}, we derive a second lemma which extends a   classical result on the complexity of the cross product of finite state automata to the cross product of (synthetised) edit automata. 

	\begin{lemma} 
		\label{lem:complexity_cross_product}
		Let $e_1, e_2\in \PropG$ and $\PSet$ be a set of observable actions. Let $v_1,v_2$ be the number of derivatives of $\synthES{e_1}{}^\PSet$ and $\synthES{e_2}{}^\PSet$, respectively.\footnote{These numbers are finite as we deal with finite-state edit automata.}    
		The complexity of the algorithm to compute $\TimesP{\synthES{e_1}{}^\PSet}{\synthES{e_2}{}^\PSet}{\Xrec}$ is $\mathcal{O}({|}\PSet{|} \cdot v_1 \cdot v_2)$. A similar result holds for edit automata derived from local properties $p_1, p_2\in \PropL$. 
	\end{lemma}

	The  third lemma provides an upper bound to the number of derivates of 
	the  automaton  $\synthES{e}{}^{\PSet}$.

	\begin{lemma}[Upper bound of number of derivatives]
		\label{lem:num-edit-derivatives}
		Let $e \in \PropG$  be a global property with  $m=\propDim{e}$, and $\PSet$ be a set of observable actions. Then, the number of derivatives of $\synthES{e}{}^{\PSet}$ is at most $m^{k+1}$, where $k$  is the number of occurrences of the symbol $\cap$ in $e$. 
	\end{lemma}
	\begin{proof}		
		The proof is by structural induction on $e$. 
	Let $e\equiv  e_1 \cap e_2$ and $m = \propDim{ e_1 \cap e_2}$. By definition, the synthesis function recalls itself on $ e_1$ and $ e_2$. Obviously,   $m_1+ m_2=m- 1$ with $m_1 = \propDim{e_1} $ and $m_2 = \propDim{e_2}$. Let $k$, $k_1$ and $k_2$ be the number of occurrences of the symbol $\cap$ in 
	$e_1 \cap e_2$,  $e_1$ and $e_2$, respectively. We deduce that $k_1+k_2=k-1$.			
	By inductive hypothesis, $\synthESP{e_1}{ }{\PSet}$ has at most \smash{$m_1^{k_1+1}$}  derivatives, and,  
	$\synthESP{e_2}{}{\PSet}$  has at most 
	\smash{$m_2^{k_2+1}$} derivatives. {As the synthesis returns the cross product between $\synthESP{e_1}{ }{\PSet}$ and $\synthESP{e_2}{}{\PSet}$, we derive that the resulting edit automaton will have at most \smash{$ m_1^{k_1+1} \cdot m_2^{k_2+1}$} derivatives}. The result follows because \smash{$ m_1^{k_1+1} \cdot m_2^{k_2+1} \leq m^{k_1+1}\cdot m^{k_2+1} \leq m^{k_1 + k_2+2} \leq m^{k-1+2} \leq m^{k+1} $}. 
	
	Let $e\equiv  p^\ast$, for $p\in \PropL$. 
	In order to analyse this case, as  $m=\propDim{p^*} =\propDim{ p } +1$ and   $\transf{p^{\ast}}^{\PSet} \defn \mathsf X$, for $ \mathsf X =\synthES{p}{\mathsf X}^{\PSet}$, we proceed by structural induction of $p\in \PropL$. We focus on the most significant case  $p\equiv \bigcup_{i\in I}\pi_i.p_i$.	     We have that $m-1=\propDim{\bigcup_{i\in I}\pi_i.p_i} $. By definition the synthesis produces $|I|$ derivatives, one for each $\pi_i \in I$, and also the derivative $\Zrec$.
	Furthermore, the synthesis algorithm re-calls itself $|I|$ times on $p_i$, with $m_i=\propDim{p_i}$ such that  $m-1 = |I| + \sum_{i \in I} m_i$,  for $i\in I$. Let $k$ and $k_i$ be the number of occurrences of $\cap$ in $p$ and in $p_i$, respectively, for $i\in I$. We deduce that $\sum_{i \in I}k_i=k$. 
	By inductive hypothesis, the synthesis produces  \smash{$m_i^{k_i+1}$} derivatives on each property $p_i$, for $i\in I$. 
	Summarising, in this case the number of derivatives is \smash{$1+|I| + \sum_{i \in I}m_i^{k_i+1}$}. Finally, the thesis follows as \smash{$  1+|I| + \sum_{i \in I}m_i^{k_i+1} \leq \sum_{i \in I}m^{k_i+1} \leq m^{k+1} $}. 
	\end{proof}

	\begin{proof}[Proof of Proposition~\ref{prop:poly2} (Complexity)]
	  For any property $e\in\PropG$ and any set of observable actions $\PSet$, 
         we prove that the recursive structure of the  function returning $\funEditP{e}{\PSet}$ can be characterised in the following form: 
          \smash{$T(m) = T(m-1) +{|}\PSet{|}\cdot  m^{k}$}, with $m=\ctrlDim{e}$, 
          and $k$ the number of occurrences of $\cap$ in $e$. The result follows because \smash{$T(m) = T(m-1) +{|}\PSet{|}\cdot  m^{k} $} is \smash{$\mathcal{O}({|}\PSet{|}\cdot m^{k+1} )$}. 
	  The proof is by  case analysis on the structure of  $e$, by examining each synthesis step in which the synthesis process $m = \ctrlDim{e}$ symbols.

		Case $ e \equiv e_1\cap e_2 $. Let $m=\propDim{e_1\cap e_2}$. By definition, the synthesis \smash{$\synthESP{e_1\cap e_2}{}{\PSet}$} 
		call itself on $e_1 $ and $e_2$, with $m_1=\propDim{e_1}$ and $m_2=\propDim{e_2}$ symbols, respectively, where $m_1+m_2=m-1$.
		Let $k$ be the number of occurrences of $\cap$ in $e$ and $k_1,k_2$ be the number of occurrences of $\cap$ in $e_1$ and $e_2$, respectively. We deduce that  $k_1 + k_2 =k -1$. 
		By an application of Lemma~\ref{lem:complexity_cross_product}, 
		the complexity of the algorithm to compute $\TimesP{\synthES{e_1}{}^\PSet}{\synthES{e_2}{}^\PSet}{\Xrec}$ is $\mathcal{O}({|}\PSet{|} \cdot v_1 \cdot v_2)$, where    $v_1$ and $v_2$ are the number of derivatives of $\synthES{e_1}{}^\PSet$ and $\synthES{e_2}{}^\PSet$, respectively.		
		By an application of Lemma~\ref{lem:num-edit-derivatives}, we have that 
		 $v_1 \leq m_1^{k_1+1}$
		and  $v_2\leq m_2^{k_2+1}$.
		Thus, the number of operations required for the cross product between \smash{$\synthESP{e_1}{ }{\PSet}$} and \smash{$\synthESP{e_2}{ }{\PSet}$} is   
		$\mathcal{O}({|}\PSet{|}\cdot m_1^{k_1+1}\cdot m_2^{k_2+1})$. 		
		Thus, we can characterise the recursive structure as:
		$T(m) = T(m_1) + T(m_2) + {|}\PSet{|}\cdot  m_1^{k_1+1}\cdot m_2^{k_2+1}$. We notice that the complexity of this recursive form is smaller than the complexity of  $T(m-1) +{|}\PSet{|}\cdot m^{k} $.
	
		Case $	e \equiv p^\ast$. In order to prove this case, as   	$m=\propDim{p^*} =\propDim{ p } +1$ and  \smash{$\transf{p^{\ast}}^{\PSet} \defn \mathsf X$}, for \smash{$ \mathsf X =\synthES{p}{\mathsf X}^{\PSet}$}, we proceed by case analysis on $p\in \PropL$. Thus, we consider the local properties $p\in \PropL$. We focus on the most significant case $p\equiv \bigcup_{i\in I}\pi_i.p_i$.     We have that $m-1=\propDim{\bigcup_{i\in I}\pi_i.p_i} $. By definition, the synthesis \smash{$\transf{\bigcup_{i\in I}\pi_i.p_i}^{\PSet}$} consumes all events $\pi_i$, for $i\in I$. The synthesis algorithm re-calls itself ${|I|}$ times on $p_i$,  with $\propDim{p_i}$ symbols, for $i \in I$. Furthermore, let $l$ be the size of the set $\PSet$, the  algorithm performs at most $l$ operations  due to a summation over over $\alpha \in \PSet \setminus (\bigcup_{i \in I}\pi_i\cup\{\tick,\fineC\})$, with $|\! \PSet \setminus (\bigcup_{i \in I}\pi_i\cup\{\tick,\fineC\}) \!| < l$. Thus, we can characterise the recursive structure as $T(m) =  \sum_{i \in I} T(  \propDim{p_i})  + l $. Since  $\sum_{i \in I} \propDim{p_i}= m-1-{|I|}\le m-1$. The resulting complexity  is smaller than that  of $T(m-1)+{|}\PSet{|}\cdot m^{k}$.   
\end{proof}

{\color{black} 
\begin{proof}[Proof of Proposition~\ref{prop:sematic-determinisct-enforcement} (Deterministic preservation)]
We reason by contradiction. Suppose  there is a sum $\sum_{i \in I} \lambda_i.\Edit_i $ appearing in $\funEditP{e}{\PSet}$ such that
                $\trigger{\lambda_k} =  \trigger{\lambda_j}$ and 
        $\mathit{out}({\lambda_k}) =  \mathit{out}({\lambda_j})$, for some $k,j \in I$, $k\neq j$. 
We proceed by case analysis on the structure of the property $e$. Let us focus on the case $e=\bigcup_{i\in I}\pi_i.p_i $. The other cases are simpler. Then, $\funEditP{e}{\PSet}$
  is equal to $\Zrec$,   for 
 \[
 		\Zrec =
                        \begin{cases}
                        \sum\limits_{\substack{{\scriptscriptstyle i\in I} }}
			\pi_i.\synthES{p_i}{\mathsf X}^{\PSet}
			+ \sum\limits_{\substack{{\scriptscriptstyle i\in I} }}
			{\scriptstyle \eact{\,\fineC}{\pi_i}}.\synthES{p_i}{\mathsf X}^{\PSet} 
			+ \sum\limits_{\subalign{\scriptscriptstyle \alpha \in  {\mathcal Q}}}{^{-}}\alpha.\Zrec, \textrm{  if } \fineC \not \in \cup_{i \in I}\pi_i  \\[2pt]
			
                         \sum\limits_{ \substack{{\scriptscriptstyle i\in I} }}
			\pi_i.\synthES{p_i}{\mathsf X}^{\PSet}
		+ \sum\limits_{\subalign{{\scriptscriptstyle \alpha \in  {\mathcal Q}  }}} {^{-}}\alpha.\Zrec, \textrm{ otherwise}                      
\end{cases}
\]
and ${\mathcal Q} = {\PSet} \setminus  (\cup_{i \in I}\pi_i\cup\{\tick,\fineC\})$.
Since $e$ is deterministic (Definition~\ref{def:deterministic-prop}) it follows that  $\pi_h \neq \pi_l$, for any $h,l \in I$, $h\neq l$. As a consequence, it cannot be  $\lambda_k= \eact{\! \fineC}{\pi_h\!}$, for $h \in I$, and $\lambda_j= \eact{\! \fineC}{\pi_l\!}$, for $l \in I$, $h \neq l$, because    $\mathit{out}({\lambda_k})= \pi_h \neq \pi_l = \mathit{out}({\lambda_j})$. 
Thus, the only chance for $\Zrec$ to be nondeterministic is that  $\lambda_k = \pi_h$, for  $h \in I$, and $\lambda_j= \eact{\! \fineC}{\pi_l\!}$, for $l \in I$, in the case $\fineC \not \in \cup_{i \in I}\pi_i $. However, this is not admissible because
 $\fineC \not \in \cup_{i \in I}\pi_i $ implies $\trigger{\lambda_k} = \pi_h \neq  \fineC = \trigger{\lambda_j}$. 	
\end{proof}
	}

{

In order to prove Theorem ~\ref{prop:transparency-logic}, we need  prove that the cross product between edit automata satisfies a standard correctness result
saying that  any execution trace associated to the intersection of two regular properties is \nolinebreak also a trace of the 
	 the cross product of the edit  automata associated to the two properties, and vice \nolinebreak versa.

  \begin{lemma}[Correctness of Cross Product] 
  	\label{lem:sound_cross}
  	Let $e_1,e_2 \in \PropG$  (\emph{resp.}, $p_1,p_2 \in \PropL$)  and $\PSet$ be a set of actions such that $\events {e_1\cap  e_2} \subseteq \PSet$
  	(\emph{resp.},  $\events {p_1\cap  p_2} \subseteq \PSet$). 
  	Then, it holds that:
  	\begin{compactitem}
  	\item
  If	  $t$ is a trace of 	
   	{\small ${\TimesP{\synthESP{e_1}{ }{\PSet}}{\synthESP{e_2}{ }{\PSet}}{\Xrec}} $} (\emph{resp.}, {\small ${\TimesP{\synthESP{p_1}{\Xrec}{\PSet}}{\synthESP{p_2}{\Xrec}{\PSet}}{\Xrec}} $}), then  
   	{\small $\widehat{\sysAct{t}}$} is prefix of some trace  in the semantics 
  	$\regSemantics{e_1\cap e_2}$ (\emph{resp.},  $\regSemantics{p_1\cap p_2}$). 
  	\item  If $t$ is a trace  in 
  	$\regSemantics{e_1\cap e_2}$ (\emph{resp.},  $\regSemantics{p_1\cap p_2}$) then there exists a trace  $t'$ of 
   	{\small ${\TimesP{\synthESP{e_1}{ }{\PSet}}{\synthESP{e_2}{ }{\PSet}}{\Xrec}} $} (\emph{resp.}, {\small ${\TimesP{\synthESP{p_1}{\Xrec}{\PSet}}{\synthESP{p_2}{\Xrec}{\PSet}}{\Xrec}} $})  such that {\small $\widehat{\sysAct{t'}}=t$}.
  \end{compactitem}
  \end{lemma}

\begin{proof}[Proof of Theorem~\ref{prop:transparency-logic} (Transparency)]  
{\color{black} 
We  prove  a stronger result.  
	Let $e\in \PropG$ be a global deterministic property and $P \in \mathbbm{Ctrl}$  be a controller such that 	\smash{${P}\trans{t} J$}, for some trace $t=\alpha_1 \cdots \alpha_n $.  If $t$ is the prefix of some trace in the semantics $\regSemantics{e}$ then the following sub-results hold: 
	\begin{compactenum}
		\item There exists a unique $\Edit$ such that  
		{\small $\synthESP{e}{ }{\PSet} \trans{t}\Edit $} 
		where either $\Edit = \synthESP{p'}{\Xrec}{\PSet}$ or $\Edit = \Zrec$, with $\Zrec = \synthESP{p'}{\Xrec}{\PSet}$, for some $p'\in \PropL$  and some automaton variable $\Xrec$. 
		\item There is a trace $t'\in \regSemantics{p'}$ such that $t\cdot t'$ is a prefix of some trace in $\regSemantics{e}$. 
		\item  There is no trace $t''=\alpha_1 \cdots \alpha_k\cdot \lambda $  for  $\funEditP{e}{\PSet}\!$ such that  $0 \leq k  <n$ and   $\lambda \in \{  \suppressE{\alpha_{k+1}}, \insertE{\alpha}{\alpha_{k+1}}\}$, for  some $\alpha$.
	\end{compactenum}
        These three sub-results imply the required result. 
We proceed by induction on the length $n$ of  \nolinebreak trace \nolinebreak $t$.
		
\noindent - \emph{Base case: }$n=1$.  That is $t=\alpha$, 	with	\smash{$\alpha\in\SensSet\cup\ChanSet^\ast\cup \overline\ActSet\cup \{\tick,\fineC\}$}.
                 We proceed by induction on the structure of $e$.

	\emph{Case} $e \equiv p^\ast$, for some \smash{$p \in \PropL$}.
        We prove the following three results:
        \begin{itemize}
          \item
		i)	there exists a unique $\Edit$ such that \smash{ $\synthESP{p}{\Xrec'}{\PSet} \trans{\alpha} \Edit$} and either \smash{$\Edit = \synthESP{p'}{\Xrec'}{\PSet}$} or $\Edit = \Zrec$, with \smash{$\Zrec = \synthESP{p'}{\Xrec'}{\PSet}$}, for some $p'\in \PropL$ and some automaton variable $\Xrec'$;
 \item 	ii) there is a trace $t'  \in \regSemantics{p'}$ such that $\alpha\cdot t' $ is a prefix of some trace in $\regSemantics{p}$;
 \item 	iii) there is no $\lambda \in \{  \suppressE{\alpha }, \insertE{\alpha'}{\alpha }\}$ such that
   \smash{$\synthESP{p}{\Xrec'}{\PSet} \trans{\lambda} \Edit'$}, for some $\Edit'$.
   \end{itemize}
	As	\smash{$\funEditP{p^\ast}{\PSet} \defn \mathsf X$}, with  \smash{${\mathsf X} = \synthESP{p}{\mathsf X}{\PSet}$}, results i) and ii) and iii) imply the required facts (1) and (2) and (3) for $e=p^\ast$. 
        We proceed as follows: first, we prove items i) and ii)  by induction on the structure of $p$,  and then we prove item iii) by contradiction.

        We prove  items i) and ii).
We focus on  the most significant cases:  $p= \bigcup_{i\in I}\pi_i.p_i$   and $ p \equiv p_1 \cap p_2$.  The other cases are similar or simpler.

\noindent
Let $p \equiv  \bigcup_{i\in I}\pi_i.p_i$. 
In this case, $\alpha$ is a prefix of some trace in  $\regSemantics{p}$ and $\synthESP{p}{\mathsf X}{\PSet}$ returns 
$\Zrec'  $, for    
 \[
 		\Zrec' =
                        \begin{cases}
                        \sum\limits_{\substack{{\scriptscriptstyle i\in I} }}
			\pi_i.\synthES{p_i}{\mathsf X}^{\PSet}
			+ \sum\limits_{\substack{{\scriptscriptstyle i\in I} }}
			\eact{\! \fineC}{\pi_i\!}.\synthES{p_i}{\mathsf X}^{\PSet} 
			+ \sum\limits_{\subalign{{\scriptscriptstyle \alpha'' \in {\mathcal Q} }}} {^{-}}\alpha''.\Zrec', \text{  if } \fineC \not \in \cup_{i \in I}\pi_i  \\[2pt]
			
                         \sum\limits_{ \substack{{\scriptscriptstyle i\in I} }}
			\pi_i.\synthES{p_i}{\mathsf X}^{\PSet}
		+ \sum\limits_{\subalign{{\scriptscriptstyle \alpha'' \in  {\mathcal Q}  }}} {^{-}}\alpha''.\Zrec', \text{  otherwise.}                      
\end{cases}
\]	
where ${\mathcal Q} = {\PSet} \setminus  (\cup_{i \in I}\pi_i\cup\{\tick,\fineC\})$.
Since $\alpha$ is a prefix of some trace in  $\regSemantics{p}$ and $\pi_i\neq \epsilon$, 
for any $i\in I$, and $e$ is deterministic, then we derive that 
$\alpha=\pi_k$, for a unique index $k \in I$. 
                \begin{compactitem}
                 \item 
                Let us prove i).
Since $k$   is the unique index such that $\alpha=\pi_k$, 
we derive that
  {$\synthESP{p}{\Xrec}{\PSet} \trans{\alpha} \Edit $} is the
   unique transition labeled $\alpha$ such that  either \smash{$\Edit = \synthESP{p_k}{\Zrec'}{\PSet}$} or $\Edit = \Zrec_1$, with \smash{$\Zrec_1 = \synthESP{p_k}{\Zrec'}{\PSet}$}.  
		                \item 
		Let us prove ii). Since  \smash{$P\trans{\alpha}J$} and $\alpha=\pi_k$, by inductive hypothesis there exists 
		$t'\in \regSemantics{p_k}$ such that $\alpha\cdot t'$ is a prefix of some trace   in $\regSemantics{\pi_k.p_k}$, and hence also in $\regSemantics{p}$, as required. 
  \end{compactitem}

		\noindent
		Let $p \equiv p_1\cap p_2$.
                In this case, we have that $\alpha$ is prefix of some trace in $\regSemantics{p}$ and  the synthesis  \smash{$\synthESP{p}{\mathsf X}{\PSet}$} returns the  edit automaton 
                \smash{$\Edit_p=\TimesP{\synthESP{p_1}{\Xrec}{\PSet}}{\synthESP{p_2}{\Xrec}{\PSet}}{\Xrec}$}.
\begin{compactitem}
  \item 
Let us prove i). By definition of cross product in Table~\ref{tab:edit-automata-product}, the most interesting case is when \smash{$\synthESP{p_1}{\Xrec}{\PSet} ={\sum_{i \in I}\lambda_i.\Edit_i} $} and \smash{$\synthESP{p_2}{\Xrec}{\PSet} = {\sum_{j \in J}\nu_j.\Fdit_j}$}. In this case,   
\[ \Edit_p=  \TimesP{\synthESP{p_1}{\Xrec}{\PSet}}{\synthESP{p_2}{\Xrec}{\PSet}}{\Xrec} = 
  \sum\limits_{(i,j)\in H} (  {\lambda_i}. \Xrec_{i,j}  \oplus  {\nu_j}. \Xrec_{i,j})   
		                                + \sum\limits_{  \alpha \in      {\mathcal Q} } {^{-}}\alpha.\Zrec, \]
for
$\Xrec_{i,j}=\TimesP{\Edit_i}{\Fdit_j}{ \Xrec_{i,j}} $
and
${\mathcal Q}= (\PSet \setminus \{\tick,\fineC\}) \setminus  \bigcup_{\tiny (i,j) \in H}\{\lambda_i,\nu_j\}$ 
 and 
$ H =\{ (i,j) \in I {\times} J: \sysAct{\lambda_i} = \sysAct{\nu_j}
                                                    \neq \tau { \text{ and } \TimesP{\Edit_i}{\Fdit_j}{  \Xrec_{i,j} }}\neq \sum\limits_{  \alpha \in  \PSet \setminus  \{\tick,\fineC\} }\!\!{^{-}}\alpha. \Xrec_{i,j}  \}$.	
Now, since $\alpha$  is a prefix of some trace in $\regSemantics{p }$,
then $\alpha$  is a prefix of some trace in both $\regSemantics{p_1}$
and  $\regSemantics{ p_2}$.
Thus, since \smash{$P \trans{\alpha} J $}, by inductive hypothesis there exists a unique $\Edit $ such that \smash{$ \synthESP{p_1}{\Xrec}{\PSet} \trans{\alpha}  \Edit $},  and either $\Edit = \synthESP{p_{1}'}{\Xrec}{\PSet}$ or $\Edit = \Zrec_{1}$, with $\Zrec_{1} = \synthESP{p_{1}'}{\Xrec}{\PSet}$, for some $p_1'\in \PropL$. 
Similarly, there exists  unique $\Fdit$ such that \smash{$ \synthESP{p_2}{\Xrec}{\PSet} \trans{\alpha}  \Fdit $},  and either $\Fdit = \synthESP{p_{2}'}{\Xrec}{\PSet}$ or $\Fdit = \Zrec_{2}$, with $\Zrec_{2} = \synthESP{p_{2}'}{\Xrec}{\PSet}$, for some $p_2'\in \PropL$. 
 Since  \smash{$\synthESP{p_1}{\Xrec}{\PSet} ={\sum_{i \in I}\lambda_i.\Edit_i} $} and \smash{$\synthESP{p_2}{\Xrec}{\PSet} = {\sum_{j \in J}\nu_j.\Fdit_j}$}, then we have that there exist $i\in I$ and $j\in J$ such that 
 $\Edit=\Edit_j$ and  $\Fdit=\Fdit_j$
By Lemma \ref{lem:sound_cross} and by definition of cross product, we have that
$(i,j)\in H$,  
$\alpha=\lambda_i$  and 
\smash{$ \Edit_p  \trans{\alpha} \Xrec_{i,j}$}, with $\Xrec_{i,j}=\TimesP{\Edit_i}{\Fdit_j}{ \Xrec_{i,j}}=\TimesP{\Edit}{\Fdit}{ \Xrec_{i,j}}$.
Thus, since $\Edit$ and $\Fdit$ are unique,
it follows that 
 \smash{$ \Edit_p \trans{\alpha} \Xrec_{i,j} $} is the only possible transition  for $\Edit_p$ with label $\alpha  $. 
Finally,	we have that $\TimesP{\Edit}{\Fdit}{ \Xrec_{i,j}}= \Times{\synthESP{p_{1}'}{\Xrec}{\PSet}}{\synthESP{p_{2}'}{\Xrec}{\PSet}}{\Xrec}=\synthESP{p'_1\cap p_2'}{\Xrec}{\PSet}$, as required.

\item
  Let us prove  ii).
{ As  \smash{$\Edit_p \trans{\alpha} \TimesP{\Edit}{\Fdit}{ \Xrec_{i,j}}= \TimesP{\synthESP{p_1'}{\Xrec}{\PSet}}{\synthESP{p_2'}{\Xrec}{\PSet}}{\Xrec} = \synthESP{p_1'\cap p_2'}{\Xrec}{\PSet}$}, by Lemma~\ref{lem:sound_cross} we derive that 
  $\regSemantics{p_1' \cap p_2'}\neq \emptyset$. Thus,    
there exists  $t'\in \regSemantics{p_1' \cap p_2'}$.
Again, by Lemma~\ref{lem:sound_cross} it follows that 
\smash{$\Edit_p  \trans{\alpha}\TimesP{\synthESP{p_1'}{\Xrec}{\PSet}}{\synthESP{p_2'}{\Xrec}{\PSet}}{\Xrec} \trans{t '}\Edit' $}, for some $\Edit' $, with
 $\alpha\cdot t'$ prefix of some trace in $\regSemantics{p_1 \cap p_2}$, as required. }
 \end{compactitem}
We have proved items i) and ii), for  any $p\in \PropL$. It remains to prove item iii) namely, if \smash{ $\synthESP{p}{\Xrec'}{\PSet} \trans{\alpha} \Edit$} then   there is no  $\lambda \in \{  \suppressE{\alpha }, \insertE{\alpha'}{\alpha }\}$ such  
$\synthESP{p}{\Xrec'}{\PSet}  \trans{\lambda}\Fdit$, for some    $\Fdit$.
By Lemma~\ref{lem:canonical} we have that 
 $\synthESP{p}{\Xrec'}{\PSet}=\Edit' $ or  $\synthESP{p}{\Xrec'}{\PSet} = \Zrec$, with $\Zrec = \Edit'$ for
 \[
\Edit'=
  \begin{cases}
                        \sum\limits_{\substack{{\scriptscriptstyle i\in I} }}
			\alpha_i.\Edit_i 
			+ \sum\limits_{\substack{{\scriptscriptstyle i\in I} }}
			\eact{\! \fineC}{\alpha_i\!}.\Edit_i 
 
			+ \sum\limits_{\subalign{{\scriptscriptstyle \alpha''  \in  {\mathcal Q} }}} {^{-}}\alpha''.\Edit'' , \textrm{  if }  \fineC   \not \in \cup_{i \in I}\alpha_i  \\[2pt]
			
                         \sum\limits_{ \substack{{\scriptscriptstyle i\in I} }}
			\alpha_i.\Edit_i 
		+ \sum\limits_{\subalign{{\scriptscriptstyle \alpha''  \in  {\mathcal Q}   }}} {^{-}}\alpha''.\Edit''  , \textrm{ otherwise.}                      
\end{cases}
\]	 for $\alpha_i \in \PSet$, ${\mathcal Q}= {\PSet} \setminus  (\cup_{i \in I}\alpha_i\cup\{\tick,\fineC\})$, and $\Edit_i$ and $\Edit'' $ edit automata. Since \smash{ $\synthESP{p}{\Xrec'}{\PSet} \trans{\alpha} \Edit$} it follows that
 $\alpha=\alpha_k$, for some  $k\in I$.  Let us assume by contradiction that
			{$\synthESP{p}{\Xrec'}{\PSet} \trans{\lambda} \Fdit $}, for some 	 $\lambda \in \{  \suppressE{\alpha }, \insertE{\alpha'}{\alpha }\}$ and automata $\Fdit$. Since $\alpha=\alpha_k$, with $k \in I$, we  derive that $\alpha \not\in {\mathcal Q} =  {\PSet} \setminus  (\cup_{i \in I}\alpha_i\cup\{\tick,\fineC\})  $, that is $\lambda$ is an insertion, 			
			$\lambda=\insertE{\alpha'}{\alpha }$, for some $\alpha'$.
                        As in $\Edit'$ the only insertions are of the form $	\eact{\! \fineC}{\alpha_i\!}$, it follows that
			$\alpha=\fineC$  and $\fineC \not \in \cup_{i \in I}\alpha_i$.
However, since $\fineC \not \in \cup_{i \in I}\alpha_i$ it follows that  $\alpha=\alpha_k \neq \fineC$. Contradiction.

\emph{Case} $e \equiv e_1 \cap e_2$, for some $e_1,e_2\in \PropG$. This case can be proved  with a reasoning similar to that  of the case $p_1 \cap p_2$.

\noindent - \emph{Inductive case:} $n> 1$, for $n\in \mathbb{N}$.
 Suppose 
\smash{$P \trans{t} J $} such that $t$ is a prefix of some trace in $\regSemantics{e}$. Since $n>1$,
\smash{$ P \trans{t'} J' \trans{\alpha} J $},  for some trace $t'$ such that $t=t' \cdot\alpha$.
As $t$ is a prefix of some trace in $\regSemantics{e}$ then 
 $t'$ is a prefix of some trace in $\regSemantics{e}$ as well. Thus, 
 by inductive hypothesis we have that: 
	\begin{compactenum}
	\item There exists a unique $\Edit'$ such that {\small $ \synthESP{e}{ }{\PSet} \trans{t'} \Edit' $}, and either \smash{$\Edit' = \synthESP{p'}{\Xrec}{\PSet}$} or $\Edit' = \Zrec$, with $\Zrec = \synthESP{p'}{\Xrec}{\PSet}$, for some $p'\in \PropL$   and some automaton variable $\Xrec$. 
	\item There is a trace $t''\in \regSemantics{p'}$ such that $t'\cdot t''$ is a prefix of some trace in $\regSemantics{e}$. 
			\item  There is no trace $t'''=\alpha_1 \cdots \alpha_k\cdot \lambda $  for  $\funEditP{e}{\PSet}\!$ such that  $0 \leq k  <n-1$ and   $\lambda \in \{  \suppressE{\alpha_{k+1}}, \insertE{\alpha'}{\alpha_{k+1}}\}$, for  some $\alpha'$.
\end{compactenum} 
Since from (1)
$\Edit'$ is unique and either \smash{$\Edit' = \synthESP{p'}{\Xrec}{\PSet}$} or $\Edit' = \Zrec$, with $\Zrec = \synthESP{p'}{\Xrec}{\PSet}$, 
 we have to prove:
	i)	there exists a unique $\Edit''$ such that \smash{ $\synthESP{p'}{\Xrec}{\PSet} \trans{\alpha} \Edit''$}, and either \smash{$\Edit'' = \synthESP{p''}{\Xrec'}{\PSet}$} or $\Edit'' = \Zrec$, with \smash{$\Zrec = \synthESP{p''}{\Xrec'}{\PSet}$}, for some $p''\in \PropL$ and some automaton variable 
	$\Xrec'$;
	ii) there is a trace $t'  \in \regSemantics{p''}$ such that $\alpha\cdot t' $ is a prefix of some trace in $\regSemantics{p'}$;
	iii)  there is no  $\lambda \in \{  \suppressE{\alpha }, \insertE{\alpha'}{\alpha }\}$, such $\synthESP{p'}{\Xrec}{\PSet} \trans{\lambda} \Fdit$, for some $\Fdit$.   These three facts can be proved as previously done for the base case,  $n=1$.
}	
\end{proof}

In order to prove  Theorem~\ref{thm:safety-logic} we  need a couple of  technical lemmata.
	\begin{lemma}[Soundness of the synthesis] 
		\label{lem:sound_edit}
		Let $e \in \PropG$ be a global property and $\PSet$ be a set of observable actions such that $\events e \subseteq \PSet$. Let \smash{$\funEditP{e}{\PSet}\trans{\lambda_1}\ldots \trans{\lambda_n}\Edit$} be an arbitrary execution trace of the synthesised automaton $\funEditP{e}{\PSet}$. Then, 
		\begin{compactenum}
			\item for $t=\sysAct{\lambda_1}\cdot\ldots\cdot\sysAct{\lambda_n}$ the trace $\hat{t}$ is a prefix of some trace in $\regSemantics{e}$; 
			\item either $\Edit = \synthESP{p'}{\Xrec}{\PSet}$ or $\Edit = \Zrec$, with $\Zrec = \synthESP{p'}{\Xrec}{\PSet}$, for some $p'\in \PropL$   and some automaton variable {$\Xrec$}.		
		\end{compactenum}
	\end{lemma}
	\begin{proof}
		We proceed by induction on the length of the execution trace  \smash{$\funEditP{e}{\PSet}\trans{\lambda_1}\ldots \trans{\lambda_n}\Edit$}. 
	%
		\noindent\emph{Base case:}  $n=1$. In this case,  \smash{$\funEditP{e}{\PSet}\trans{\lambda}\Edit$}. We proceed by induction on the structure of $e$.
		
		\emph{Case} $e \equiv p^\ast$, for some $p \in \PropL$.
		We prove by induction on the structure of $p$ the following two results:
		i) for $\beta=\sysAct{\lambda}$, $\hat{\beta}$ is a prefix of some trace in $\regSemantics{p}$, and
		ii) either \smash{ $\Edit = \synthESP{p'}{\Xrec'}{\PSet}$} or $\Edit = \Zrec$, with \smash{$\Zrec = \synthESP{p'}{\Xrec'}{\PSet}$}, for some $p'\in \PropL$ and some automaton variable {$\Xrec'$}. 
		As \smash{$\funEditP{p^\ast}{\PSet} \defn \mathsf X$, for ${\mathsf X} = \synthESP{p}{\mathsf X}{\PSet}$},  results i) and ii) imply the required results (1) and (2), for $e=p^\ast$.  We show the cases $p\equiv p_1;p_2$ and $p\equiv p_1 \cap p_2$, the others cases are similar or simpler.

	\noindent Let $p\equiv p_1;p_2$ and\smash{ $\synthESP{p_1; p_2}{\mathsf X}{\PSet}  \trans{\lambda}\Edit$}.
	We prove the two results i) and ii) for $p_1 \neq \epsilon$, the case $p_1 = \epsilon$ is simpler. By definition, \smash{$\synthESP{p_1; p_2}{\mathsf X}{\PSet}$} returns \smash{$\synthESP{p_1}{\Zrec'}{\PSet}$}, for \smash{$\Zrec'=\synthESP{p_2}{\mathsf X}{\PSet}$}, and $\Zrec' \neq \mathsf X$.	
	As a consequence, from $p_1 \neq \epsilon$ and \smash{$\synthESP{p_1; p_2}{\mathsf X}{\PSet}  \trans{\lambda}\Edit$} it follows that
	\smash{$\synthESP{p_1}{\Xrec}{\PSet}\trans{\lambda}\Edit_1$}, for  some 
	$\Edit_1$.
	\begin{compactitem}
	\item Let us prove  i).  Since \smash{$\synthESP{p_1}{\Xrec}{\PSet}\trans{\lambda}\Edit_1$},  by inductive hypothesis 
          we have that  $\hat{\beta}$ is a prefix of some trace in $\regSemantics{p_1}$. Thus,  $\hat{\beta}$ is a prefix of some trace in $\regSemantics{p_1;p_2}$, as required.
		\item Let us prove ii). 	Again, since \smash{$\synthESP{p_1}{\Xrec}{\PSet}\trans{\lambda}\Edit_1$}, by inductive hypothesis either \smash{$\Edit_1 = \synthESP{p_1'}{\Zrec'}{\PSet}$} or $\Edit_1 = \Zrec_1$, with \smash{$\Zrec_1 = \synthESP{p_1'}{\Zrec'}{\PSet}$}, for some $p_1'\in \PropL$   and some automaton variable {$\Zrec'$}.
		Let us analyse \smash{$\Edit_1 = \synthESP{p_1'}{\Zrec'}{\PSet}$} (the case  $\Edit_1 = \Zrec_1$, with $\Zrec_1 = \synthESP{p_1'}{\Zrec'}{\PSet}$, is similar).  
		As  \smash{$\Edit_1 = \synthESP{p_1'}{\Zrec'}{\PSet}$} with
		\smash{$\Zrec'=\synthESP{p_2}{\mathsf X}{\PSet}$},		by definition of the synthesis algorithm it follows that  \smash{$\Edit_1 =\synthESP{p_1';p_2}{\mathsf X}{\PSet}$}, as required.
	\end{compactitem}

	\noindent   {\color{black}
		Let $p\equiv p_1\cap p_2$ and {\small $\synthESP{p_1\cap p_2}{\mathsf X}{\PSet}  \trans{\lambda}\Edit$}.
By definition, the synthesis algorithm applied to \smash{$\synthESP{p_1\cap p_2}{\mathsf X}{\PSet}  $} returns {\small $\Edit_p=\Times{\synthESP{p_1}{\Xrec}{\PSet}}{\synthESP{p_2}{\Xrec}{\PSet}}{\Xrec}$}. Let us prove the results i) and ii).
\begin{compactitem}
	\item Result i) follows directly from Lemma~\ref{lem:sound_cross}.
	\item Let us prove ii). 
		By inspection of the definition of cross product in Table~\ref{tab:edit-automata-product}, the most interesting case is when {$\synthESP{p_1}{\Xrec}{\PSet} ={\sum_{i \in I}\lambda_i.\Edit_i} $} and \smash{$\synthESP{p_2}{\Xrec}{\PSet} = {\sum_{j \in J}\nu_j.\Fdit_j}$}.
In this case,   
\[ \Edit_p  =\Times{\synthESP{p_1}{\Xrec}{\PSet}}{\synthESP{p_2}{\Xrec}{\PSet}}{\Xrec} = 
  \sum\limits_{(i,j)\in H} (  {\lambda_i}. \Xrec_{i,j}  \oplus  {\nu_j}. \Xrec_{i,j})   
		                                + \sum\limits_{  \alpha \in      {\mathcal Q} } {^{-}}\alpha.\Zrec \]
for 
$\Xrec_{i,j}=\TimesP{\Edit_i}{\Fdit_j}{ \Xrec_{i,j}} $
and
${\mathcal Q}= (\PSet {\setminus} \{\tick,\fineC\}) {\setminus}  \bigcup_{\tiny (i,j) \in H}\{\lambda_i,\nu_j\}$ 
 and 
$ H =\{ (i,j) \in I {\times} J: \sysAct{\lambda_i} = \sysAct{\nu_j}
                                                    \neq \tau { \text{ and } \TimesP{\Edit_i}{\Fdit_j}{  \Xrec_{i,j} }}\neq \sum\limits_{  \alpha \in  \PSet \setminus  \{\tick,\fineC\} }\!\!{^{-}}\alpha. \Xrec_{i,j}  \}$.	
Hence $\Edit_p$ has following two (families of) transitions:  
        (a)
	\smash{$\Edit_p \trans{\lambda} {\Xrec_{i,j}}$}, for  $\lambda\in \bigcup_{\tiny (i,j) \in H}\{\lambda_i,\nu_j\}$; 
		(b)   \smash
	  {$\Edit_p\!\!\trans{\suppressE{\alpha}} \Zrec$}, for   $  \alpha \in \mathcal Q $.  
	We prove the result for the case (a);  the case (b)   can be proved in a similar manner. 
Since $\lambda\in \bigcup_{\tiny (i,j) \in H}\{\lambda_i,\nu_j\}$ we have that 
$\lambda=\lambda_i$	or $\lambda=\nu_j$, for some $(i,j)\in H$.  By definition of cross product,  it holds that \smash{$\synthESP{p_1}{\Xrec}{\PSet}\trans{\lambda_i}\Edit_i$ and $\synthESP{p_2}{\Xrec}{\PSet}\trans{\nu_j}\Edit_j$}, with $\sysAct{\lambda_i}=\sysAct{\nu_j}=\sysAct{\lambda}$. Thus, by inductive hypothesis 
        we have that: (1) either \smash{$\Edit_{i} = \synthESP{p_{1}'}{\Xrec}{\PSet}$} or $\Edit_{i}  = \Zrec_{1}$, with \smash{$\Zrec_{1} = \synthESP{p_{1}'}{\Xrec}{\PSet}$}, for some $p_1'\in \PropL$; (2) 
  either \smash{$\Fdit_{j} = \synthESP{p_{2}'}{\Xrec}{\PSet}$} or $\Fdit_{j}  = \Zrec_{2}$, with \smash{$\Zrec_{2} = \synthESP{p_{2}'}{\Xrec}{\PSet}$}, for some $p_2'\in \PropL$.
 	Therefore,  by definition of cross product  we derive that $\TimesP{\Edit_i}{\Fdit_j}{ \Xrec_{i,j}}=$   \smash{$\Times{\synthESP{p_{1}'}{\Xrec}{\PSet}}{\synthESP{p_{2}'}{\Xrec}{\PSet}}{\Xrec}$}. Finally, by definition of our synthesis it follows that \smash{$\Times{\synthESP{p_{1}'}{\Xrec}{\PSet}}{\synthESP{p_{2}'}{\Xrec}{\PSet}}{\Xrec} = \synthESP{p'_1\cap p_2'}{\Xrec}{\PSet}$}, as required.
\end{compactitem}
}

		\emph{Case} $e = e_1 \cap e_2$ for some $e_1,e_2\in \PropG$. This case can be proved  with a reasoning similar to that seen in the proof of case $p_1 \cap p_2$.
		
\emph{Inductive case:} $n>1$, for $n\in \mathbb{N}$. Suppose \smash{$\funEditP{e}{\PSet}\trans{\lambda_1}\ldots  \trans{\lambda_{n-1}} \Edit' \trans{\lambda_n}\Edit$}, for $n>1$.  $\funEditP{e}{\PSet}\trans{\lambda_1}$ \smash{$\ldots \trans{\lambda_{n-1}}\Edit'\trans{\lambda_n}\Edit$}. Thus, by induction, we have that:
\begin{compactenum}
	\item for \smash{$t'=\sysAct{\lambda_1}\cdot\ldots\cdot\sysAct{\lambda_{n-1}}$} the trace $\widehat{t'}$ is a prefix of some trace in $\regSemantics{e}$, and
	\item either \smash{$\Edit' = \synthESP{p'}{\Xrec}{\PSet}$} or $\Edit' = \Zrec$, with \smash{$\Zrec = \synthESP{p'}{\Xrec}{\PSet}$}, for some $p'\in \PropL$  and some automaton variables $\Zrec$ and $\Xrec.$		
\end{compactenum}
	Since either \smash{$\Edit' = \synthESP{p'}{\Xrec}{\PSet}$} or $\Edit' = \Zrec$, with \smash{$\Zrec = \synthESP{p'}{\Xrec}{\PSet}$}, then to conclude the proof it is sufficient to prove that given \smash{$\synthESP{p'}{\Xrec}{\PSet}\trans{\lambda_n}\Edit$} and $\beta_n=\sysAct{\lambda_n}$, it holds that $\hat{\beta_n}$ is a prefix of some trace in $\regSemantics{p'}$. For that we resort to the proof of the base case.
	\end{proof}

	In the next lemma, we prove that, given the execution traces of a monitored controller, we can always extract from them the traces performed by its edit automaton and its monitored controller in isolation. 
\begin{lemma}[Trace decomposition]
	\label{lem:execution_decomposition}
	Let $\Edit  \in\mathbbm{Edit}$ be an edit automaton and   
	$J  \in\mathbbm{Ctrl}$ be a controller.
	Then, for any execution 
	trace {\small $\confCPS{\Edit_0}{J_0}\trans{\beta_1}
		\ldots\trans{\beta_n}\confCPS{\Edit_n}{J_n}$},
		 with $\Edit_0=\Edit$ and $J_0=J$, it hold that
	%
		(1) $\Edit_{i-1} \trans{\lambda_i} \Edit_i$,  with 
		$\beta_i = \sysAct{\lambda_i}$, and  (2) 
		either {\small $J_{i-1}\trans{\alpha_i}J_{i}$}, with $\alpha_i = \trigger{\lambda_i}$,  or $J_{i}=J_{i-1}$, for   $1  \leq i \leq  n$. 
\end{lemma}
\begin{proof}
The transition
$\confCPS{\Edit_{i-1}}{J_{i-1}}\smash{\trans{\beta_i}}\confCPS{\Edit_i}{J_i}$,
for  $1  \leq i \leq  n$, can be only derived by applying one of the following  rule: \rulename{Allow}, \rulename{Insert}, \rulename{Suppress}.
In the case of an application of rule \rulename{Allow}, $\confCPS{\Edit_{i-1}}{J_{i-1}}\smash{\trans{\beta_i}}\confCPS{\Edit_i}{J_i}$ 
derives from \smash{$\Edit_{i-1} \trans{\alpha_i} \Edit_i$} and 
\smash{$J_{i-1}{\trans{\alpha_i }}J_i$} with $\beta_i=\alpha_i=\lambda_i$.
Hence,  
	 $\sysAct{\lambda_i}=\trigger{\lambda_i}=\alpha_i$, as required. 
In the case of rule \rulename{Insert},  $\confCPS{\Edit_{i-1}}{J_{i-1}}\smash{\trans{\beta_i}}\confCPS{\Edit_i}{J_i}$ 
derives from \smash{$\Edit_{i-1} \trans{\insertE{\alpha}{\alpha_i}} \Edit_i$} and 
\smash{$J_{i-1}{\trans{\alpha_i }}J $}, for some $\alpha$ and $J$,  with 
$\beta_i=\alpha$.
Thus,  
	 $\sysAct{\lambda_i} {=} \sysAct{\insertE{\alpha}{\alpha_i}}=\beta_i$ and 
 $J_i=J_{i-1}$, as required. 
Finally, in the case of rule \rulename{Suppress},
$\confCPS{\Edit_{i-1}}{J_{i-1}}\smash{\trans{\beta_i}}\confCPS{\Edit_i}{J_i}$ 
derives from \smash{$\Edit_{i-1} \trans{\suppressE{\alpha_i}} \Edit_i$} and 
\smash{$J_{i-1}{\trans{\alpha_i }}J_i$}, for some $\alpha_i$, with  $\beta_i=\tau$ and $\lambda_i = \suppressE{\alpha_i}$. Hence,  
$\sysAct{\lambda_i}=\tau$  and $\trigger{{\lambda_i}}=\alpha_i$, as required.
\end{proof}

	\begin{proof}[Proof of Theorem~\ref{thm:safety-logic} (Soundness)]
		Let  $t=\beta_1{\cdot}\ldots\cdot\beta_n $ be a trace s.t.\ $\confCPS{\funEditP{e}{\PSet}}{P} \trans{t} \confCPS{\Edit}{J}$, for some $\Edit \in \mathbbm{Edit}$ and some controller $J$.
		By an application of Lemma~\ref{lem:execution_decomposition}
                there exist $\Edit_i \in \mathbbm{Edit}$ and  $\lambda_i $, for $1 \leq i \leq n$, such that: 
         \smash{$\funEditP{e}{\PSet}\trans{\lambda_1}\Edit_1\trans{\lambda_2}\ldots \trans{\lambda_n}\Edit_n=\Edit$,  with $\beta_i = \sysAct{\lambda_i}$}.
         Thus, $t=\sysAct{\lambda_1}\cdot\ldots\cdot\sysAct{\lambda_n}$.  
         By    Lemma~\ref{lem:sound_edit},   $\hat{t}$ is a prefix of some trace in 
         $\regSemantics{e}$, as required. 
	\end{proof}

\begin{lemma}[Deadlock-freedom of the synthesis]
	\label{lem:deadlock-freedom-synthesis}
		Let $e \in \PropG$ be a global property and $\PSet$ be a set of observable actions s.t.\ $\events e \subseteq \PSet$. Then the edit automaton $\synthESP{e}{ }{\PSet}$ does not deadlock.
\end{lemma}
\begin{proof} Given an arbitrary execution \smash{$\synthESP{e}{ }{\PSet}\trans{\lambda_1} \ldots \trans{\lambda_n}\Edit$}, the proof is by induction on the length $n$ of the execution trace.  By an application of Lemma~\ref{lem:sound_edit} we have that 
	either $\Edit = \synthESP{p}{\Xrec}{\PSet}$ or $\Edit = \Zrec$, with $\Zrec = \synthESP{p}{\Xrec}{\PSet}$, for $p\in \PropL$ and some automaton variable $\Xrec$. Hence, the result follows by inspection of the synthesis function of Table~\ref{tab:synthesis-logic} and  by induction on the structure of $p$.
\end{proof}

	\begin{proof}[Proof of Theorem~\ref{thm:deadlock} (Deadlock-freedom)]
		Let $t$ be a trace such that ${\confCPS{\funEditP{e}{\PSet}}{P}\trans{t}\confCPS{\Edit}{J}}$, for some edit automaton $\Edit$ and controller $J$. By contradiction we assume that $\confCPS{\Edit}{J}$ is in deadlock. Notice that, by definition, our  controllers $J$ never deadlock. By Lemma~\ref{lem:deadlock-freedom-synthesis} the automaton $\funEditP{e}{\PSet}$ never deadlock as well. 
		Consequently, we have that for any transition 
		$J\trans{\alpha}J'$ there is no action $\lambda$ for $\Edit$, 
		such that the monitored controller  $\confCPS{\Edit}{J}$
		 may progress  
		according to one of the rules: \rulename{Allow}, \rulename{Suppress} and \rulename{Insert}.
By an application of Lemma~\ref{lem:sound_edit}, we have that 
	either $\Edit = \synthESP{p}{\Xrec}{\PSet}$ or $\Edit = \Zrec$, with $\Zrec = \synthESP{p}{\Xrec}{\PSet}$, for some $p\in \PropL$ and some automaton variable $\Xrec$. 
	{\color{black}
	Now, by Lemma~\ref{lem:canonical}, we have that 
\[
 \synthESP{p}{\Xrec}{\PSet} =
  \begin{cases}
                        \sum\limits_{\substack{{\scriptscriptstyle i\in I} }}
			\alpha_i.\Edit_i 
			+ \sum\limits_{\substack{{\scriptscriptstyle i\in I} }}
			\eact{\! \fineC}{\alpha_i\!}.\Edit_i 
 
			+ \sum\limits_{\subalign{{\scriptscriptstyle \alpha  \in  {\mathcal Q} }}} {^{-}}\alpha.\Fdit, \textrm{  if }  \fineC   \not \in \cup_{i \in I}\alpha_i  \\[2pt]
			
                         \sum\limits_{ \substack{{\scriptscriptstyle i\in I} }}
			\alpha_i.\Edit_i 
		+ \sum\limits_{\subalign{{\scriptscriptstyle \alpha  \in  {\mathcal Q}   }}} {^{-}}\alpha.\Fdit  , \textrm{ otherwise.}                      
\end{cases}
\]	 for $\alpha_i \in \PSet$, ${\mathcal Q}= {\PSet} \setminus  (\cup_{i \in I}\alpha_i\cup\{\tick,\fineC\})$, and $\Edit_i$ and $\Fdit$ edit automata.}
In both the cases  $\synthESP{p}{\Xrec}{\PSet}$  may only deadlock the enforcement when the controller may only perform $\tick$-actions. From this fact,  we derive 	 
		 $J = \tick^h.S$, for $0< h\leq k$. Since $\tick$-actions cannot be suppressed, we have that
		 $t=t'  \cdot \tick^{k-h}$, for some possibly empty trace $t'$ terminating with an $\fineC$. By Theorem~\ref{thm:safety-logic}, 
		 $t=t'  \cdot \tick^{k-h} \in \regSemantics{e}$.
		 And since  $e$ is $k$-sleeping 
		 we derive $p=\tick^h.p'$, for some $p'$.
		 Since $\funEditP{e}{\PSet}$ is sound (Lemma~\ref{lem:sound_edit}) we derive that   $\Edit =\synthESP{p}{\Xrec}{\PSet}=\synthESP{\tick^h.p'}{\Xrec}{\PSet}$. Finally, $h>0$ implies  \smash{$\Edit \trans{\tick}\Edit'$}, for some $\Edit'$, in contradiction with what stated four lines above. 
	\end{proof}

\begin{proof}[Proof of Theorem~\ref{thm:divergence} (Divergence-freedom)]
	 Let $e \in \PropG$ be a global property in its general form,  given by the intersection of $n \geq 1$ global properties $  p_1^\ast \cap \dots \cap p_n^\ast$, for  $p_i\in\PropL$, with $1\leq i \leq n$. As $e$ is well-formed, according to Definition~\ref{def:well-formedness} also all local properties $p_i$ are well-formed.
                This means that they all terminate with an $\fineC$ event.  Thus,  in all global properties $p_i^\ast$, for $1 \leq i \leq n$,  the number of events  within two subsequent $\fineC$ events  is always finite. The same holds for the property $e$. 
                Now, 	let 		$t$ be an arbitrary trace such that \smash{${\confCPS{\funEditP{e}{\PSet}}{P}\trans{t}\confCPS{\Edit}{J}}$}, for some edit automaton $\Edit$ and controller $J$. And let $k= \mathrm{max}_{1\leq i\leq n}k_i$, where $k_i$ is the  length of the longest trace of $\regSemantics{p_i}$, for $1\leq i\leq n$.
Thus, if \smash{$\confCPS{\Edit}{J}\trans{t'}\confCPS{\Edit'}{J'}$}, with $|t'|\ge k$,  and since  by Theorem~\ref{thm:safety-logic} we have that $t\cdot t'$ is a prefix of some trace $\regSemantics{e}$,
 then $\fineC \in t'$.   
\end{proof}

\end{document}